\documentclass{article}%
\usepackage{amsfonts}
\usepackage{makeidx}
\usepackage{amsmath}
\usepackage{amssymb}
\usepackage{graphicx}%
\providecommand{\U}[1]{\protect\rule{.1in}{.1in}}
\newtheorem{theorem}{Theorem}

\newtheorem{corollary}[theorem]{Corollary}

\newtheorem{proposition}[theorem]{Proposition}

\newenvironment{proof}[1][Proof]{\noindent\textbf{#1.} }{\ \rule{0.5em}{0.5em}}
\begin{document}

\title{The generic model of General Relativity}
\author{Michael Tsamparlis$^{1}$\thanks{Email: mtsampa@phys.uoa.gr}\\{\ \ }$^{1}${\textit{Faculty of Physics, Department of
Astronomy-Astrophysics-Mechanics,}}\\{\ \textit{University of Athens, Panepistemiopolis, Athens 157 83, Greece.}}
\and Andronikos Paliathanasis$^{2,3,4}$\thanks{Email: anpaliat@phys.uoa.gr}\\{\ }$^{2}${\textit{Instituto de Ciencias F\'{\i}sicas y Matem\'{a}ticas, }}\\{\textit{Universidad Austral de Chile, Valdivia, Chile.}}\\$^{3}${\textit{Department of Mathematics and Natural Sciences, }}\\{\textit{Core Curriculum Program, Prince Mohammad Bin Fahd University, }}\\{\textit{Al Khobar 31952, Kingdom of Saudi Arabia.}}\\$^{4}${\textit{Institute of Systems Science, Durban University of Technology}}\\{\textit{Durban 4000, Republic of South Africa.}}}
\date{}
\maketitle

\begin{abstract}
We develop a generic spacetime model in General Relativity which can be used
to build any gravitational model within General Relativity. The generic model
uses two types of assumptions: (a) Geometric assumptions additional to the
inherent geometric identities of the Riemannian geometry of spacetime and (b)
Assumptions defining a class of observers by means of their 4-velocity $u^{a}$
which is a unit timelike vector field. The geometric assumptions as a rule concern symmetry
assumptions (the so called collineations). The latter introduces the 1+3
decomposition of tensor fields in spacetime. The 1+3 decomposition results in
two major results. The 1+3 decomposition of $u_{a;b}$ defines the kinematic
variables of the model (expansion, rotation, shear and 4-acceleration) and
defines the kinematics of the gravitational model. The 1+3 decomposition of
the energy momentum tensor representing all gravitating matter introduces the
dynamic variables of the model (energy density, the isotropic pressure, the
momentum transfer or heat flux vector and the traceless tensor of the
anisotropic pressure) as measured by the defined observers and define the
dynamics of he model. The symmetries assumed by the model act as constraints
on both the kinematical and the dynamical variables of the model. As a second
further development of the generic model we assume that in addition to the
4-velocity of the observers $u_{a}$ there exists a second universal vector field
$n_{a}$ in spacetime so that one has a so called double congruence
$(u_{a},n_{a})$ which can be used to define the 1+1+2 decomposition of tensor
fields. The 1+1+2 decomposition leads to an extended kinematics concerning
both fields building the double congruence and to a finer dynamics involving
more physical variables. After presenting and discussing the results in their
full generality we show how they are applied in practice by considering in a
step by step approach the case of a string fluid in Bianchi I spacetime for
the comoving observers.

\end{abstract}


\section{Introduction}

General Relativity (GR) is the first theory of Physics which uses geometry to
a such a great extend \cite{ein1,ein2,ein3,ein4}. Newtonian Physics is also a
geometric theory of Physics but the difference is that it uses the 3-d Euclidian geometry
which has a direct sensory correspondence therefore `couples' much easier and
directly with the description of the physical phenomena.

In GR the space used (spacetime) and the geometry employed (Riemannian
Geometry) do not allow for a direct comparison of the mathematical description
with the direct sensorial reality. Most frequently this results in a confusion
as to where geometry starts and where she stops, how much Physics is used and
done, what is the physical interpretation of the geometrically derived results
etc. On top of that it does not seem to exist at a reasonable level a concise
exposition of the structure of GR as a theory of Physics and a clear
description of the various hypotheses made both at the geometric and at the
physical level when a gravitational model is presented. For example it is
widely misunderstood that Einstein's gravitational field equations \cite{ein4}
are equations in the usual sense; that is, they contain unknown quantities of
some sort which are specified as soon as one "solves" them. However on the
contrary these equations contain only unknowns and there is no way or point to
solve them. In fact they are generators of standard field equations which
result after certain assumptions are made which specify a "model"
gravitational universe and whose solutions reveal the properties of the
specific model universe. This is the reason that exist so many (and still
are produced) "solutions" of the GR\ field equations in the literature
\cite{exact,Kransinski}.

The purpose of this article is to develop in a systematic way the structure of
a generic model for GR so that one has a clear understanding of the
impact of the geometric and the physical assumptions made in the
construction of a certain gravitational model within the framework of GR.

The generic model of GR\ consists of two parts:

a) The first part concerns all possible parameters which are used in order to
build a GR model. These parameters are classified in two large sets (i)
Geometric restrictions / assumptions and (ii) Physical restrictions / requirements

b) The second part consists of all constrains the parameters of the first part
must satisfy in order to lead to a mathematically consistent relativistic
model. These constrains are also classified in two large sets\\
 (i)
Constraints resulting from geometric identities and geometric relations in
general and \\
(ii) Constraints due to physical simplifications which have the
general name equations of state.

The constrains in the two sets are not
independent, in the sense, that parameters from one set constraint the
parameters in the second set and vice versa. This is due to the GR\ field
equations which relate the geometry with the physics. It is apparent that when
one makes assumptions on the physics of a model these assumptions have to be
compatible with the geometric requirements of the model. If this is not the
case then there is the possibility that the proposed gravitational model is
inconsistent hence an invalid\ model.

The Physics today is still in its Aristotelian form, that is, a statement /
prediction / result is either "true" (i.e. justified by the world "out there")
or not. This is the logic of zero and one. As it is well known today the
Aristotelian point of view has been replaced in many areas (especially the
technological) by the Fuzzy Physics approach which is beyond the logic or true
or false. This type of Physics has yet to come. However because GR is
definitely a theory of the Aristotelian type we are not to worry about that
type of future developments and we shall follow the classical path.

\section{The generic model of GR}

Although the theories of Physics follow the general pattern of the theories of
other sciences, they have the unique characteristic that they use mathematical
entities to model actual physical quantities of the real world. The
association of these two very different type of entities imposes certain
requirements which a theory of Physics must satisfy in order to be meaningful.
In general terms every theory of Physics must have the following ingredients.

\section{The background mathematical space}

Physics does not build its models in the `real' space i.e. the space we live
in. This is the case only in Newtonian Physics where we have a direct sense of
the evolution of physical systems in space due to the direct sensory
conception we have for the Newtonian world. In relativistic models of Physics
we have an `indirect sensory' conception of the world by means of our
measuring instruments. The first relativistic theory which was proposed was of
course Special Relativity where people could not have a direct sensory feeling
therefore they were unable to understand and anxious to disprove it by means
of the many paradoxes proposed. In short it was violating their `common sense'
the latter be possible to be put forward as the rule\ "I\ see it, I\ believe it!"

The trouble is the use of the word `space' to describe both a mathematical
entity (software) and a physical entity (hardware). The physical entity is
what it is `out there'. The mathematical entity is a set with certain
mathematical structures. Our intention here is not to enter into details about
these delicate matters, which in any case are not first priority questions
to a general physicist.

The mathematical space used by the theories of Physics is a point set which
has the structure of a manifold. The main characteristic of the manifold
structure is that it is locally diffeomorphic to an open set in a flat space
of some dimension (called the dimension of the manifold) therefore it can be
covered by a set of coordinate patches from the flat space in such a way that
whenever two coordinate patches coincide there is a differentiable map which
relates them either way. In General Relativity the set of points are the
events which are assumed to correspond to the various events in the real
world. The coordinate patches are called coordinate systems and are
diffeomorphic to open neighborhoods of $R^{4}$. The differentiable maps which
relate two coordinate systems are called coordinate transformations. This
manifold is called spacetime\footnote{In addition to the local structure
spacetime, there exists a global structure described by its topology. However
these topological properties are not of our interests.}.

The coordinate transformations form a group under the action of composition of
maps. This group is an infinite dimensional Lie group called the Manifold
Mapping Group (MMG) \cite{anderson}. The mathematical quantities which in any
coordinate system are described by means of a set of components so that under
coordinate transformations they transform in a definite way we call geometric
objects. We say that the geometric objects form a representation of the MMG.
The nature of a geometric object is characterized by the its components
transform under coordinate transformations in the manifold. The geometric
objects which transform in a linear and homogeneous manner we call tensors.
In the following we shall restrict our considerations to tensors.

 This is as
far as one can go with the assumption of manifold structure of the background
mathematical space. However this mathematical structure is not enough to study the
physical phenomena because it lacks the concept of "measure".  It is safe to say that in all theories of Physics
the concept of "measure" is introduced by the
 requirement of the existence of a specific
geometric object on the manifold which is the metric. Let us see how this is done.

It is possible (but not necessary!) that besides the MMG group a theory of Physics introduces an
additional characteristic subgroup of the MMG. This is achieved by considering
an additional inherent structure by means of an absolute\footnote{By absolute
we mean that there exist coordinate systems called inertial coordinate systems
which are defined all over the space and in which the metric has its canonical
form, that is it is represented by a diagonal square matric whose elements are
$\pm1.$ The linear coordinate transformations which preserve the canonical
form of the metric forma group which is the characteristic group of the
theory.} metric tensor defined all over the space. This subgroup defines
special classes of tensors by the requirement that they transform covariantly
under the coordinate transformations derived by the special subgroup. For
example, in Special Relativity one assumes the Lorentz metric whose canonical
form\footnote{Canonical form of a metric in a coordinate system is defined by
the requirement that in that coordinate system the metric is represented
with a diagonal matrix whose entries are $\pm1.$} is preserved under the
coordinate transformations which we know as Lorentz transformations and are
elements of the Lorentz group. The geometric objects which transform
covariantly under the Lorentz transformations we call Lorentz tensors.
Similarly, in Newtonian Physics one assumes the Euclidian metric whose
canonical form is the tensor $\delta_{\mu\nu}.$ The coordinate transformations
which preserve $\delta_{\mu\nu}$ we call Galilean transformations and they
form the Galilean group which defines the Newtonian tensors. In GR\ we do not
assume an absolute metric\footnote{Assumes only that at each point there is a
coordinate system in which the metric takes its canonical form which is
the canonical form of the Lorentz metric i.e. $diag(-1,1,1,1).$%
}; hence, it does not\ exist a special subgroup and the geometric objects of
the theory are general tensors.

\subsection{The role of Geometry}

The introduction of a metric, which is a tensor or
order (0,2), in a general manifold defines an additional structure which we call
Geometry. Using geometry one is able to define a correspondence of the geometric objects with
physical entities and define the covariant derivative which is necessary in the formulation
of the Laws of Physics.  Let us see how this is done in the classical theories of Physics.

In Newtonian Physics the physical requirement is that there exist Euclidian
solids, that is objects whose points are such that their Euclidian distance
does not change as they change their state. This postulates the absolute character of the
Euclidian metric. The geometry defined in the 3d manifold $R^3$ by the Euclidian
metric is the so-called Euclidian geometry which makes the space to be the
Euclidian space $E^{3}.$ In practical terms this means that under linear
transformations which are isometries of the Euclidian metric $E^{3}$ the laws
of Physics stay the same. These transformations form a group which is called
the Galilean group. The tensors of the space which are transformed covariantly
with respect to the Galilean group are called Newtonian tensors. All Newtonian
Physics develops on the space $E^{3}$ and all Newtonian physical quantities are described
by Newtonian tensors. Furthermore all Newtonian Laws relate Newtonian tensors only.

In Special Relativity, the physical requirement is the invariance of the speed
of light. Extending this to more general than the light phenomena leads to the
postulation of the Lorentz metric and the introduction of the Minkowski space
$M^{4}$ where the Physics of Special Relativity is done. The isometry group
of the Lorentz metric is the Lorentz group (or more correctly the Poincar\'{e}
group)\ whose linear transformations are the Lorentz transformations. These
 transformations define the special relativistic tensors. It is important
to note that the Physics which is based on $M^{4}$ will be different form the
Newtonian Physics which is based on $E^{3}.$ This means that the relativistic
phenomena are not necessarily Newtonian and certainly the Newtonian are not
relativistic. For instance there are no Euclidian solids in Special
Relativity. Therefore there is no point to say that one theory is wrong and
the other correct. Simply each theory (provided it has a "touch" with
reality)\ applies to its own physical phenomena. Of course people would like
to have the "theory of everything" but this is another characteristic aspect of the
human utopia

Concerning GR the situation is different. The main interest of GR
is to develop a theory for the gravitational field while at the
same time to reduce to Special Relativity when the gravitational field is
switched off. Therefore in this theory one keeps the notion of metric and
associates it with the gravitational field. The difference with the previous
two theories is that there is no a globally defined metric therefore there
does not exist a characteristic subgroup of the MMG group for this theory.
This implies that the tensors which are used in GR are general tensors which
transform linearly and homogeneously with respect to the transformations of
the MMG\ group. Furthermore, because at each point we want the theory to
reduce to Special Relativity, it is postulated that the signature of the
metric will be that of the Lorentz metric, that is $-2.$

The result of considering a general metric is that in the manifold of GR\ one
defines a Riemannian geometry (i.e. a geometry with curvature) which is characterized by the introduction of
the notion of the Riemannian covariant derivative. The latter  is defined by the requirements
that it is symmetric (the torsion tensor vanishes) and also metrical (the covariant derivative of
the metric vanishes, or as we say, it has zero metricity). The manifold
structure endowed with the Riemannian geometry we call spacetime. Spacetime is
the background space of General Relativity where all gravitational phenomena
shall be considered.

The introduction of Riemannian geometry has many new effects which are hidden
in the spaces $E^{3}$ and $M^{4}$ whose geometry is also Riemannian but
trivial. These effects are the following.

\begin{itemize}
\item Using the metric and its derivatives one introduces many new geometric
objects which vanish in $E^{3}$and $M^{4}.$ These are the connection
coefficients, the Ricci tensor $R_{ab},$ the curvature tensor $R_{abcd}$
etc.$.$ With a collective name we shall call them metric tensors.

\item There are geometric identities which must be satisfied by the metric
tensors. These main identities are the following

A) The Ricci identity: For a tensor of type $(r,s)$ this identity is:%
\begin{equation}
T_{\phantom{.........}l_{1}...l_{s};h;k}^{j_{1}...j_{r}}%
-T_{\phantom{.........}l_{1}...l_{s};k;h}^{j_{1}...j_{r}}=\sum_{a=1}^{r}%
R_{m}^{\phantom{...}j_{a}}{}_{hk}%
T_{\phantom{.............................}l_{1}...l_{s}}^{j_{1}..j_{a-1}%
m.j_{a+1....}j_{r}}-\sum_{\beta=1}^{r}R_{m}^{\phantom{...}j_{a}}{}%
_{hk}T_{\phantom{.........}l_{1}..l_{\beta-1}ml_{b+1}...l_{s}}^{j_{1}...j_{r}}%
\end{equation}

In particular for a vector field $X^{a}$ Ricci identity reads:%
\begin{equation}
X_{a;bc}-X_{a;cb}=R_{dabc}X^{d}. \label{Key.1.a}%
\end{equation}
B) The Bianchi identities:%
\begin{align}
R_{ab[cd;e]}  &  =0\label{BI.01}\\
C_{\mbox{ \phantom ...};d}^{abcd}  &  =R^{c[a;b]}-\frac{1}{6}g^{c[a}R^{;b]}
\label{BI.02}%
\end{align}
where $C_{abcd}$ is the Weyl tensor\footnote{To be defined in the following.} In a
four dimensional Riemannian space these two identities are
equivalent\footnote{This equivalence holds for $n=4$ only. See Kundt and
Trumper Abh. Akad.Wiss. and Lit. Mains., Mat. Nat.Kl.,No \textbf{12} (1962).}
therefore in spacetime they reduce to one identity. \ These identities imply
the important contracted Bianchi identity:
\begin{equation}
R_{....;b}^{ab}-\frac{1}{2}R^{;a}=0\Leftrightarrow G_{...;b}^{ab}=0
\label{BI.03}%
\end{equation}
where:
\begin{equation}
G_{ab}=R_{ab}-\frac{1}{2}Rg_{ab} \label{BI.04}%
\end{equation}
is the Einstein tensor.

\item A third result of the geometrical structure is the possibility it
provides to introduce new geometric requirements which will act as
constraints additional to the inherent ones assumed by the Riemannian structure of spacetime.
The major type of such geometric constraints are the "symmetries" or collineations of
the metric objects. A symmetry or collineation is a relation of the form
\begin{equation}
\,L_{X}A=B \label{Key.4}%
\end{equation}
where $X$ is a vector field and $A,B$ are metric geometric objects with the
same number of indices and the same symmetry properties of the indices. The
following are examples of collineations \cite{katzin1}:%
\[%
\begin{array}
[c]{ll}%
\mathcal{L}_{\xi}g_{ab}=0 & \mbox{Killing Vector}\\
\mathcal{L}_{\xi}g_{ab}=2cg_{ab},\;c=const. & \mbox{Homothetic Killing
Vector}\\
\mathcal{L}_{\xi}g_{ab}=2{\psi}g_{ab},\;{\psi}\not =const. & \mbox{Conformal
Killing Vector}\\
\mathcal{L}_{\xi}{\Gamma}_{jk}^{i}=0 & \mbox{Affine collineation Vector}\\
\mathcal{L}_{\xi}R_{ab}=0 & \mbox{Ricci collineation}\\
\mathcal{L}_{\xi}R_{jkl}^{i}=0 & \mbox{Curvature collineation}\\
g^{ab}\mathcal{L}_{\xi}R_{ab}=0 & \mbox{Contracted Ricci collineation.}\\
\mathcal{L}_{\xi}W_{jkl}^{i}=0 & \mbox{Contracted Ricci collineation.}
\end{array}
\]

\item Other geometrical constraints.
These are mathematical requirements introduced `by hand' after the previous
two levels of simplifying assumptions have been exhausted. The purpose of
their introduction is to simplify further the equations obtained by the
previous simplifying assumptions. Their form is depends on the form of the
equations they have to simplify. For example in the case of a Conformal
Killing Vectors (CKV) such a requirement is $\psi_{;ab}=0$ because this
condition removes $\psi_{;ab}$ from the field equations. A CKV which satisfies
this condition is called a special CKV.
\end{itemize}

It is to be noted that up to now our discussion is limited to the background
space only, therefore the results apply to all physical fields introduced or,
to be more specific, to both the kinematics and the dynamics of the whatever
gravitational theory will be developed on this spacetime.

\subsection{Some useful material about collineations}

Each collineation has a different constraint `power'.
A collineation is called proper if it cannot be reduced to a `simpler' one.
For example, the CKVs contain the homothetic Killing vectors (HKV) when the
conformal function $\psi=$constant. A proper CKV is one for which $\psi$ is
not constant.

The vectors which satisfy $\mathcal{L}_{\xi}g_{ab}=0$ are called Killing
vectors (KVs) and the equation $\mathcal{L}_{\xi}g_{ab}=0$ Killing equation. In
an n-dimensional space with a non-degenerate metric there exist at most
$\frac{n(n+1)}{2}$ KVs. If this is the case then the space is called a space
of constant curvature. The KVs serve all standard geometric symmetry concepts
e.g. spherical symmetry, cylindrical symmetry etc.

The KVs form a finite dimensional Lie algebra. This algebra can be used to
classify the metrics in various classes in the sense that one is able to write
the metric in a form that takes into account the symmetry and it is written in
terms of a small number of parameters.

A well known example is the Friedmann - Robertson - Walker (FRW) spacetime
which is used in the standard cosmology. The metric of this spacetime has the
form:%
\begin{equation}
ds_{RW}^{2}=-dt^{2}+R(t)^{2}(dx^{2}+dy^{2}+dz^{2}) \label{Key.5}%
\end{equation}
where $(t,x,y,z)$ are Cartesian coordinates and \thinspace$R(t)$ is the only
unknown function of (the coordinate)\ time $t$ and can be defined by means of
the following symmetry assumptions \cite{M. Tsamparlis 1992}:

\begin{enumerate}
\item Spacetime admits a timelike gradient Conformal Killing Vector (CKV) such
that there exists another reduced metric for which the CKV\ is a KV

\item The 3-D hypersurfaces orthogonal to the CKV are spaces of constant
curvature, consequently the FRW spacetime admits six KVs.
\end{enumerate}

The collineations other than Killing vectors fix the metric to a lesser
degree, however they do act as constraints and can be used in any case as such.

\subsection{The generic collineation}

One can prove that the Lie derivative of every metric geometric object in
Riemannian\ Geometry can be expressed in terms of the quantity $\mathcal{L}%
_{\xi}g_{ab}$ and its derivatives. For instance, the following relations are true:%

\begin{align}
L_{\xi}\Gamma_{\text{ }ab}^{c}  &  =\frac{1}{2}g^{cd}\left[  \left(  L_{\xi
}g_{ad}\right)  _{;b}+\left(  L_{\xi}g_{bd}\right)  _{;a}-\left(  L_{\xi
}g_{ab}\right)  _{;d}\right] \label{Key.6}\\
L_{\xi}R_{\text{ }bcd}^{a}  &  =\left(  L_{\xi}\Gamma_{\text{ }bd}^{a}\right)
_{;c}-\left(  L_{\xi}\Gamma_{\text{ }bc}^{a}\right)  _{;d}\label{Key.7}\\
L_{\xi}R_{\text{ }ab}  &  =\left(  L_{X}\Gamma_{\text{ }ab}^{c}\right)
_{;c}-\left(  L_{X}\Gamma_{\text{ }ca}^{c}\right)  _{;b}. \label{Key.8}%
\end{align}

This observation leads us to introduce the concept of \textbf{generic
collineation} by the relation/identity:%
\begin{equation}
L_{\xi}g_{ab}=2\psi g_{ab}+2H_{ab} \label{Key.9}%
\end{equation}
where $H_{ab}$ is a symmetric traceless tensor. Then it is possible to express
every collineation in terms of the symmetry variables $\psi,H_{ab}$ and their
derivatives. This approach greatly unifies and simplifies the study of the
effects of each collineation and reveals its relative significance.

As an example we refer the following general result:%
\begin{equation}
L_{\xi}R_{ab}=-\left(  n-2\right)  \psi_{;ab}-g_{ab}\square\psi+2H_{.\left(
a;b\right)  d}^{d}-\square H_{ab}. \label{Key.10}%
\end{equation}
In terms of trace and traceless parts $L_{\xi}R_{ab}$ is written as follows:%
\begin{equation}
L_{\xi}R_{ab}=-\left(  n-2\right)  A_{ab}+2K_{ab}-\square H_{ab}+\frac{1}%
{n}g_{ab}\left[  -2\left(  n-1\right)  \square\psi+2H_{.;ab}^{ab}\right]
\label{Key.11}%
\end{equation}
where the traceless tensors:%
\begin{equation}
A_{ab}=\psi_{;ab}-\frac{1}{n}g_{ab}\square\psi~\text{and }K_{ab}=H_{.\left(
a;b\right)  d}^{d}-\frac{1}{n}g_{ab}H_{.;cd}^{cd}. \label{Key.12}%
\end{equation}
In spacetime $(n=4)$ the above formulae read:%
\begin{equation}
L_{\xi}R_{ab}=-2\psi_{;ab}-g_{ab}\square\psi+2H_{.\left(  a;b\right)  d}%
^{d}-\square H_{ab} \label{Key.13}%
\end{equation}

\begin{equation}
L_{\xi}R_{ab}=-2A_{ab}+2K_{ab}-\square H_{ab}+\frac{1}{4}g_{ab}\left(
-6\square\psi+2H_{.;ab}^{ab}\right)  \label{Key.14}%
\end{equation}

\begin{equation}
A_{ab}=\psi_{;ab}-\frac{1}{4}g_{ab}\square\psi~\text{and }K_{ab}=H_{.\left(
a;b\right)  d}^{d}-\frac{1}{4}g_{ab}H_{.;cd}^{cd}. \label{Key.15}%
\end{equation}

This is as far as one can go with the spacetime structure (i.e. the introduction of the metric
tensor). It is possible to continue making Geometry in spacetime but no
Physics. To do the latter one has to introduce more tensor fields which will
describe / correspond to the physical quantities.

\section{The observers - Kinematics}

The main purpose of all physical theories is to "explain" the real world
(whatever "real" means) to us, the observers who observe the physical
quantities and describe their evolution in the cosmos. Therefore all theories
of Physics must have two types of entities:\ Observers and observed phenomena.
In non-quantum Physics these two entities are separate in the sense that they
do not interact. In quantum Physics this is not the case and one has the so
called Uncertainty Principle which, however, will not concern us here. Both
types of entities must be described mathematically in the space of the theory
by a certain type of geometric object.

In GR the observers are represented in the same way as
in Special Relativity, that is, by a timelike world line whose tangent unit
four vector $u^{a}$ is the the four velocity of the observers. In GR there
is no way to define a special class of inertial observers because it is not
assumed to exist a characteristic group of transformations defined all over the
spacetime manifold. Instead of them, one defines the closest type of observers
which are the free falling observers whose spacetime trajectory is a geodesic
of the metric tensor considered by the specific model. All other classes of
observers in the same spacetime are characterized as accelerating observers. We
note that each gravitational model in GR has its own free falling observers
depending on the metric it adopts. Also when the gravitational field is
switched off the geodesics are straight lines in the space $M^{4}$ therefore
the observers coincide with those of Special Relativity.

The world lines of a class of observers in GR constitute the `fluid' of
observers in spacetime and this fluid is used to express the kinematic
properties of these observers in terms of this fluid flow quantities

As we have already mentioned the assumed structure of Riemannian geometry of
spacetime introduces identities among the tensor fields which act as
"internal" constraints universal to all models one may consider. In addition to
these \textquotedblright internal\textquotedblright\ identities, each model
theory introduces new geometric relations which act as \textquotedblright
external\textquotedblright\ constraints to the theory. The two types of
constraints must be satisfied by the four velocity vector of the observers of
the model theory. This fact defines the first level of the physical theory
which is called kinematics and it involves only the considered observers and the
geometric constraints (internal and external)\ imposed by the model. On a
given kinematics one is possible to build various  models by defining
a different dynamics of the theory.

The introduction of a non-null vector field in spacetime allows one to
decompose a tensor equation and an individual tensor parallel and normal to
the vector field. The vector field can be either timelike or spacelike. This
decomposition is called \textbf{1+3 decomposition} for obvious reasons. It
is of primary importance to Physics because it is covariant, in the sense that
each irreducible part transforms under coordinate transformations
independently of the other. Therefore, one can break the study the Physics
(Kinematics and Dynamics) of a tensor field or a field equation by studying the
simpler Physics of its irreducible parts.

Below we develop the basic mathematics of the 1+3 decomposition with respect
to  a non-null vector field $P^{a}$. The results of this analysis can be extended easily to
the $1+(n-1)$ decomposition in an $n-$dimensional spacetime by making the necessary adjustments.

\subsection{1+3 decomposition wrt a non-null vector field}

Consider a non-null ($P^{i}P_{i}\neq0)\ $vector field $P^{i}$ with signature
$\varepsilon(P)=P^{i}P_{i}/|P^{i}P_{i}|=\pm1,$ where the $+1$ applies to a
spacelike 4-vector and the $-1$ to a timelike 4-vector in a metric space with
metric $g_{ij}$. The projection tensor $h_{ij}(P)$ associated with $P^{i}$ is
defined by the equation:%

\begin{equation}
h_{ij}\left(  P\right)  =g_{ij}-\frac{\varepsilon\left(  P\right)  }{P^{2}%
}P_{i}P_{j} \label{d.one.1}%
\end{equation}
where $P^{2}=|P^{i}P_{i}|=\varepsilon(P)P^{i}P_{i}>0.$ In a four dimensional
manifold, it is easy to prove the properties:%
\[
h_{ij}(P)P^{j}=0,~h^{ij}(P)h_{jk}(P)=h_{k}^{i}(P),~h_{i}^{i}(P)=3.
\]
The tensor $h_{ij}\left(  P\right)  $ projects normal to $P^{a}$ and gives us
the possibility to decompose any tensor field in irreducible parts in a
direction parallel to $P^{a}$ and another normal to $P^{a},$ as given by the
following propositions.

\begin{proposition}
\qquad A general vector field $R^{i}$ \ can be 1+3 decomposed wrt $P^{i}$ as follows:%

\begin{equation}
R^{i}=\alpha P^{i}+\beta^{j}h_{j}^{i} \label{d.one.2}%
\end{equation}
where
\begin{align}
\alpha &  =\frac{\varepsilon(P)}{P}R^{i}P_{i}\\
\beta^{i}  &  =R^{i}%
\end{align}

\end{proposition}

\begin{proof}
We have $R^{i}=R^{j}g_{j}^{i}.$ Using (\ref{d.one.1}) this gives:%
\begin{align*}
R^{i}  &  =R^{j}g_{j}^{i}\\
&  =R^{j}\left(  \frac{\varepsilon\left(  P\right)  }{P^{2}}P_{j}P^{i}%
+h_{j}^{i}\right) \\
&  =\frac{\varepsilon\left(  P\right)  }{P^{2}}(R^{j}P_{j})P^{i}+h_{j}%
^{i}R^{j}.%
\end{align*}

\end{proof}

\begin{proposition}
A general second order tensor $Y_{ij}$ is 1+3 decomposed wrt the vector
$P^{i}$ by means of the identity:%
\begin{equation}
Y_{ij}=\alpha P_{i}P_{j}+\varepsilon\left(  P\right)  \beta_{k}h_{j}^{k}%
P_{i}+\varepsilon\left(  P\right)  \gamma_{k}h_{i}^{k}P_{j}+\epsilon_{ij}
\label{d.one.3}%
\end{equation}
where%
\[
\alpha=\frac{1}{P^{4}}Y_{ij}P^{i}P^{j}\,\,,\,\,\beta_{i}=\frac{1}{P^{2}}%
Y_{ij}P^{j}\,\,,\,\,\gamma_{i}=\frac{1}{P^{2}}Y_{ji}P^{j}\,\,,\,\,\epsilon
_{ij}=Y_{kr}h_{i}^{k}h_{j}^{r}%
\]

\begin{proof}
We write $Y_{ij}=\allowbreak Y_{kr}g_{i}^{k}g_{j}^{r}$ and using
(\ref{d.one.1}) we have:%

\begin{align*}
Y_{ij}  &  =Y_{kr}g_{i}^{k}g_{j}^{r}\\
&  =Y_{kr}(\left(  \frac{1}{P^{2}}\varepsilon\left(  P\right)  P^{k}%
P_{i}+h_{i}^{k}\right)  \left(  \frac{1}{P^{2}}\varepsilon\left(  P\right)
P^{r}P_{j}+h_{j}^{r}\right) \\
&  =\left(  \frac{1}{P^{4}}Y_{kr}P^{k}P^{r}\right)  P_{i}P_{j}+\frac{1}{P^{2}%
}\varepsilon\left(  P\right)  \left(  Y_{kr}P^{k}\right)  P_{i}h_{j}^{r}%
+\frac{1}{P^{2}}\varepsilon\left(  P\right)  \left(  Y_{kr}P^{r}\right)
h_{i}^{k}P_{j}+Y_{kr}h_{i}^{k}h_{j}^{r}\\
&  =\alpha P_{i}P_{j}+\varepsilon\left(  P\right)  \beta_{k}h_{j}^{k}%
P_{i}+\varepsilon\left(  P\right)  \gamma_{k}h_{i}^{k}P_{j}+\epsilon_{ij}.
\end{align*}

\end{proof}
\end{proposition}

We note that a symmetric  tensor of rank 2  is specified (and equivalently specifies) in terms
of five different quantities: one scalar, two vector and one projected second
rank tensor.

\subsection{1+3 decomposition wrt a timelike unit vector field}

In the following we consider $P^{i}$ to be a normalized timelike vector field
$\left(  \varepsilon\left(  P\right)  =-1\right)  $ (e.g. a 4-velocity) and we
denote it $u^{i}$. Then equations (\ref{d.one.1}), (\ref{d.one.2}),
(\ref{d.one.3}) give the 1+3 decomposition wrt $u^{i}$ ($u^{i}u_{i}=-1):$
\begin{align}
h_{ij}\left(  u\right)   &  =g_{ij}+u_{i}u_{j}\label{d.one.4}\\
R^{i}  &  =-\left(  R^{j}u_{j}\right)  u^{i}+R^{j}h_{j}^{i}\left(  u\right)
\label{d.one.5}\\
Y_{ij}  &  =\left(  Y_{kr}u^{k}u^{r}\right)  u_{i}u_{j}-\left(  Y_{kr}%
u^{k}\right)  u_{i}h_{j}^{r}-\left(  Y_{kr}u^{r}\right)  h_{i}^{k}u_{j}%
+Y_{kr}h_{i}^{k}h_{j}^{r}. \label{d.one.6}%
\end{align}

For a \textbf{symmetric }tensor of type (0,2) decomposition (\ref{d.one.6}) is
written as follows:%
\begin{equation}
Y_{ab}=(Y^{rs}u_{r}u_{s})u_{a}u_{b}-Y^{rs}u_{r}h_{sb}u_{a}-Y^{rs}u_{s}%
h_{ra}u_{b}+\frac{1}{3}(Y^{rs}h_{rs})h_{ab}+(h_{a}^{r}h_{b}^{s}-\frac{1}%
{3}h_{ab}h^{rs})Y_{rs} \label{N.2}%
\end{equation}
that is we brake further the symmetric part in a trace and a traceless part.

We consider now various applications of the 1+3 decomposition wrt the vector
field $u^{a}.$

\subsubsection{The kinematic variables of the four-velocity}

We perform the 1+3 decomposition of the tensor $u_{i;j}$. By definition, from
(\ref{d.one.4}) we have the following identity/decomposition:%

\begin{equation}
u_{i;j}=\left(  u_{k;r}u^{k}u^{r}\right)  u_{i}u_{j}-\left(  u_{k;r}%
u^{k}\right)  u_{i}h_{j}^{r}-\left(  u_{k;r}u^{r}\right)  h_{i}^{k}%
u_{j}+u_{k;r}h_{i}^{k}h_{j}^{r}. \label{d.one.7}%
\end{equation}
Moreover,$\ u_{i;j}u^{i}=\frac{1}{2}\left(  u_{i}u^{i}\right)  _{;j}=0$ and
$u_{i;j}u^{i}=\dot{u}_{j}=h_{j}^{r}\dot{u}_{r}.$ Therefore, identity
(\ref{d.one.7}) is simplified as follows
\begin{equation}
u_{i;j}=-\dot{u}_{i}u_{j}+u_{k;r}h_{i}^{k}h_{j}^{r}. \label{d.one.8}%
\end{equation}

We continue by decomposing the space-like part $u_{k;r}h_{i}^{k}h_{j}^{r}$ in
an antisymmetric and a symmetric part as follows:
\begin{align}
u_{k;r}h_{i}^{k}h_{j}^{r}  &  =\omega_{ij}+\theta_{ij}\label{d.one.9}\\
\omega_{ij}  &  =u_{k;r}h_{[i}^{k}h_{j]}^{r}\label{d.one.10}\\
\theta ij  &  =u_{k;r}h_{(i}^{k}h_{j)}^{r}. \label{d.one.11}%
\end{align}
The symmetric part can be decomposed covariantly further to a trace and a
traceless part, that is we write:
\begin{equation}
\theta_{ij}=\sigma_{ij}+\frac{1}{3}\theta h_{ij} \label{d.one.12}%
\end{equation}
where%
\begin{align*}
\theta &  =\theta_{i}^{i}=h^{ij}u_{i;j}\\
\sigma_{ij}  &  =\theta_{ij}-\frac{1}{3}\theta h_{ij}=\left[  h_{(i}^{r}%
h_{r)}^{k}-\frac{1}{3}h^{rk}h_{ij}\right]  u_{r;k}.
\end{align*}

The $\ $term \ $\omega_{ij}$ is called the vorticity tensor of $u^{i},$
$\sigma_{ij}$ the shear tensor of $u^{i},$ $\theta$ the expansion of $u^{i}$
and $\dot{u}^{i}$ the four-acceleration of the timelike vector field $u^{i}$.
These are the kinematic variables of $u^{i}$, considering $u^{i}$ to be the
4-velocity of a relativistic fluid.

We infer by their definition that the kinematic variables satisfy the
properties:%
\begin{align*}
\omega_{ij}  &  =-\omega_{ji}\,,\,\omega_{ij}u^{j}=0\\
\sigma_{ij}  &  =\sigma_{ji}\,,\,\sigma_{i}^{i}\,=\,\sigma_{ij}u^{j}=0\\
\omega_{ij}  &  =h_{i}^{k}h_{j}^{r}\omega_{kr}\\
\sigma_{ij}  &  =h_{i}^{k}h_{j}^{r}\sigma_{kr}.
\end{align*}

The geometric meaning of each kinematic term is obtained from the study of the
integral curves of $u^{a}.$ We shall not comment further on that at this point.

We interpret the physical meaning each quantity in terms of relative motion.
In this respect we say that $\theta$ is expansion (isotropic strain),
$\sigma_{ab}$ is shear (anisotropic strain), $\omega_{ab}$ relative rotation
and $\overset{.}{u}_{a}$denotes the four-acceleration. These quantities (i.e.
$\sigma_{ab},\omega_{ab},\overset{.}{u}_{a},\theta$) \ are the fundamental
physical quantities of relativistic (and Newtonian with the necessary
adjustments) kinematics.

Two scalars of special interests are the following:%
\begin{align}
\sigma^{2}  &  =\frac{1}{2}\sigma_{ab}\sigma^{ab}\label{d.one.12.a}\\
\omega^{2}  &  =\frac{1}{2}\omega_{ab}\omega^{ab} \label{d.one.12.b}%
\end{align}

In order to demonstrate the geometric importance of the kinematic quantities
we refer the following proposition.

\begin{proposition}
Necessary and sufficient conditions for the vector field $u^{i}$ to be a
Killing vector is:
\[
\sigma_{ij}=0\,,\,\theta=0\,,\dot{u}^{i}=0
\]

\end{proposition}

\begin{proof}
We have:%
\begin{equation}
L_{u}g_{ij}=2u_{\left(  i;j\right)  }=2\sigma_{ij}+\frac{2}{3}\theta
h_{ij}-2\dot{u}_{(i}u_{j)} \label{Ldd.5}%
\end{equation}
In order $u^{i}$ to be a Killing vector it should $\ $satisfy the
condition\ $L_{u}g_{ij}=0.$ Combining these two relations we find:%
\[
\sigma_{ij}=0\,,\,\theta=0\,,\dot{u}^{i}=0
\]

\end{proof}

\section{The propagation and the constraint equations}

As we remarked above the vector field of the four-velocity decomposes both the
geometric objects of the theory as well as the covariant equations among them.
The Riemannian structure of spacetime has two major geometric identities: The
Ricci identity and the Bianchi identities. The 1+3 decomposition of these
identities produces the necessary equations which must be satisfied by all the
kinematic quantities involved.

The 1+3 decomposition of the Ricci identity leads to two sets of equations
each set containing nine equations. The first set contains the derivatives of
the kinematic quantities along $u^{a}$ which we call the propagation
equations. The second set of equations is called the constraint equations
\cite{ellis}. The two sets of equations have as follows:

\textbf{Propagation equations}\footnote{Note that the propagation equation of
$\omega_{ab}$ (\ref{CPE.22}) contains only kinematic terms and it is
independent of $R_{ab}$ hence of the dynamical variables to be introduced
later on.}:
\begin{equation}
h_{a}^{b}\dot{\omega}^{a}=(\sigma_{d}^{a}-\frac{2}{3}\theta h_{d}^{a}%
)\omega^{d}+\frac{1}{2}\eta^{abcd}u_{b}\dot{u}_{c;d}\mbox{ \ \
(three
equations)} \label{CPE.22}%
\end{equation}%
\begin{equation}
\dot{\theta}+\frac{1}{3}\theta^{2}+2(\sigma^{2}-\omega^{2})=-R_{ab}u^{a}%
u^{b}+\dot{u}^{a}{}_{;a}\text{ \ (one equation) } \label{CPE.23}%
\end{equation}%
\begin{align}
-E_{st}+\frac{1}{2}\left(  h_{s}^{a}h_{t}^{b}-\frac{1}{3}h_{st}h^{ab}\right)
R_{ab}\text{ }  &  =h_{s}^{a}h_{t}^{b}\left[  \dot{\sigma}_{ab}-\dot
{u}_{(a;b)}\right]  +\sigma_{sc}\sigma_{\text{ \thinspace}t}^{c}+\mbox{(five
equations)}\nonumber\\
&  +\frac{2}{3}\sigma_{st}\theta+\omega_{s}\omega_{t}-\dot{u}_{s}\dot{u}%
_{t}-\frac{1}{3}\left(  2\sigma^{2}+\omega^{2}-\dot{u}_{;b}^{b}\right)
h_{st}. \label{CPE.24}%
\end{align}

\textbf{Constraint equations:}%

\begin{equation}
h_{b}^{a}\omega_{;a}^{b}=\dot{u}^{a}\omega_{a}\mbox{ \ (one
equation)} \label{CPE.21}%
\end{equation}%
\begin{equation}
h_{s}^{c}\left[  \frac{2}{3}\theta,_{c}-h^{ab}\sigma_{ca;b}-\eta_{camn}%
u^{a}\left(  \omega^{m;n}+2\omega^{m}\dot{u}^{n}\right)  \right]  =-h_{s}%
^{c}R_{cd}u^{d\text{ \ }}\mbox{(three equations)} \label{CPE.26}%
\end{equation}%
\begin{equation}
-h_{(s}^{a}h_{l)}^{b}\left[  \sigma_{b}^{\text{ \ }c;d}+\omega_{b}^{\text{
\ }c;d}\right]  \eta_{arcd}u^{r}+2\dot{u}_{(s}\omega_{l)}=H_{sl}%
\mbox{ \ (five
equations)} \label{CPE.27}%
\end{equation}

\section{The 1+3 decomposition of the Bianchi identity}

The 1+3 decomposition of the Bianchi identity (\ref{BI.02}) involves the Weyl tensor.

It is well known that the Riemann tensor can be decomposed in the following
irreducible parts as follows:%
\begin{equation}
R_{abcd}=C_{abcd}+\frac{1}{2}\left(  g_{ac}R_{bd}+g_{bd}R_{ac}-g_{ad}%
R_{bc}-g_{bc}R_{ad}\right)  -\frac{1}{6}Rg_{abcd}\label{CPE.22a}%
\end{equation}
where:%
\begin{equation}
g_{abcd}=g_{ac}g_{bd}-g_{ad}g_{bc}\label{CPE.22c}%
\end{equation}
and $C_{abcd}$ \ is the Weyl tensor. The  Weyl tensor has the same symmetries of indices
as the curvature tensor but all its traces vanish. Equation (\ref{CPE.22a}) is
a mathematical identity. The Weyl tensor is decomposed further in terms of the
electric part $E_{ab}$ and the magnetic part $H_{ab}$ wrt the vector field
$u^{a}$ as follows:%
\begin{equation}
C_{abcd}=(g_{abrs}g_{cdmt}-\eta_{abrs}\eta_{cdmt})u^{r}u^{m}E^{st}%
-(\eta_{abrs}g_{cdmt}+g_{abrs}\eta_{cdmt})u^{r}u^{m}H^{st}.\label{CPE.22b}%
\end{equation}
where
\begin{equation}
E_{ac}=C_{abcd}u^{b}u^{d},\text{ \ }H_{ac}=\frac{1}{2}\eta_{am.}^{\text{
\ \ \ \ }kl}C_{klbt}u^{m}u^{t}.\label{CPE.22d}%
\end{equation}
The tensors $E_{ac},H_{ac}$ are symmetric, traceless and satisfy the property:%
\begin{equation}
E_{ac}u^{c}=H_{ac}u^{c}=0.\label{CPE.22e}%
\end{equation}

If in (\ref{BI.02}) we substitute the Weyl tensor $C^{abcd}$ in terms of the
electric and the magnetic part $E^{ab},H^{ab}$ we obtain four identities
which have a form similar to Maxwell equations for the electric and the
magnetic field:

$\nabla E:$%
\begin{align}
h_{a}^{t}h_{s}^{d}E_{\mbox{\phantom {...}
};d}^{as}-\eta^{tbpq}u_{b}\sigma_{p}^{\mbox{ \ }d}H_{qd}+3H_{\mbox{
\ }s}^{t}\omega^{s}  &  =\frac{1}{3}h^{tb}\mu_{;b}-\frac{1}{2}h_{c}^{t}%
\pi_{~\ ;b}^{cb}-\frac{3}{2}\omega_{.b}^{t}q^{b}+\label{BI.05}\\
&  +\frac{1}{2}\sigma_{.b}^{t}q^{b}+\frac{1}{2}\pi_{.b}^{t}\dot{u}^{b}%
-\frac{1}{3}\theta q^{t}\nonumber
\end{align}

$\nabla H:$%
\begin{align}
h_{a}^{t}h_{s}^{d}H_{\mbox{\phantom {...}};d}^{as}+\eta^{tbpq}u_{b}\sigma
_{.p}^{d}E_{qd}-3E_{.s}^{t}\omega^{s}  &  =\left(  \mu+p\right)  \omega
^{t}+\frac{1}{2}\eta^{tbsf}u_{b}q_{s;f}+\nonumber\\
&  +\frac{1}{2}\eta^{tbsf}u_{b}\pi_{sc}\left(  \omega_{f}^{.c}+\sigma_{f}%
^{.c}\right)  \label{BI.06}%
\end{align}

$\dot{E}_{ab}:$%
\begin{align}
\left[
\begin{array}
[c]{c}%
h_{a}^{m}h_{c}^{t}\dot{E}^{ac}+h_{a}^{(m}\eta^{t)rsd}u_{r}H_{.s;d}^{a}%
-2H_{q}^{.(t}\eta^{m)bpq}u_{b}\dot{u}_{p}+\\
+h^{mt}\sigma^{ab}E_{ab}+\theta E^{mt}-3E_{s}^{(m}\sigma^{t)s}-E_{s}%
^{(m}\omega^{t)s}%
\end{array}
\right]   &  =-\frac{1}{2}\left(  \mu+p\right)  \sigma^{tm}-\dot{u}^{(t}%
q^{m)}-\frac{1}{2}h_{a}^{t}h_{c}^{m}q^{\left(  a;c\right)  }\nonumber\\
&  -\frac{1}{2}h_{a}^{t}h_{c}^{m}\dot{\pi}^{ac}-\frac{1}{2}\pi^{b(m}\left(
\omega_{b}^{.t)}+\sigma_{b}^{.t)}\right)  -\frac{1}{6}\pi^{tm}\theta
\nonumber\\
&  + \frac{1}{6}h^{mt}\left(  q_{.;a}^{a}+\dot{u}_{a}q^{a}+\pi^{ab}\sigma
_{ab}\right)  \label{BI.07}%
\end{align}

$\dot{H}_{ab}:$%
\begin{align}
\left[
\begin{array}
[c]{c}%
h^{ma}h^{tc}\dot{H}_{ac}-h_{a}^{(m}\eta^{t)rsd}E_{.s;d}^{a}+2E_{q}^{(t}%
\eta^{m)bpq}u_{b}\dot{u}_{p}+\\
+h^{mt}\sigma^{ab}H_{ab}+\theta H^{mt}-2H_{s}^{.(m}\sigma^{t)s}-H_{s}%
^{.(m}\omega^{t)s}%
\end{array}
\right]   &  =\frac{1}{2}\sigma_{c}^{(t}\eta^{m)bcf}u_{b}q_{f}-\frac{1}%
{2}h_{c}^{(t}\eta^{m)bef}u_{b}\pi_{.e;f}^{c}\nonumber\\
&  +\frac{1}{2}\left(  h^{mt}\omega_{c}q^{c}-3\omega^{(m}q^{t)}\right)
\label{BI.08}%
\end{align}

The contracted Bianchi identity $G_{..;b}^{ab}=0$ is contained in the above identities.

\section{Propagation of the kinematic quantities along a collineation
vector}

The propagation equations provide the derivative of the kinematic quantities
along the four-velocity $u^{a}$ of the observers due to the internal structure
of spacetime. However when a collineation is assumed by a given model then
this is an external geometric constraint which must also be satisfied. This
makes necessary the knowledge of the change of the kinematic quantities along
the vector field generating the collineation. Because the collineations are
defined in terms of the Lie derivative (Lie transport) we are interested in
the quantities $L_{\xi}\{u^{a},\omega_{ab},\sigma_{ab},\theta,\dot{u}^{a}\}.$
We give the following general result.

\begin{proposition}
Let $X^{a}$ be a non-null vector field with index $\varepsilon(X)=\pm1$ i.e.
$X^{a}X_{a}=\varepsilon(X)X^{2}$ $(X>0)$ \ and let $\xi^{a}$ be an arbitrary
vector field with collineation parameters(see \cite{Tsamparlis 1992})
$\psi,H_{ab}$ that is:%
\begin{equation}
\xi_{(a;b)}=\psi g_{ab}+H_{ab}=\frac{1}{2}L_{\xi}g_{ab},\qquad L_{\xi}%
g^{ab}=-2\psi g^{ab}-2H^{ab}. \label{Ldd.0}%
\end{equation}
Then $L_{\xi}X^{a}$ is 1+3 decomposed wrt the four-velocity $u^{a}$ as
follows:%
\begin{equation}
L_{\xi}X^{a}=\left[  (\ln X)_{,b}\xi^{b}-\psi-\frac{\varepsilon(X)}{X^{2}%
}H_{cb}X^{c}X^{b}\right]  X^{a}+V^{a}(X) \label{Ldd.4a}%
\end{equation}
where $V^{a}(X)=h_{b}^{a}(X)(L_{\xi}X^{b})\,\ $is a vector field normal to
$X^{a}.$ For the covariant quantity $L_{\xi}X_{a}$ holds:%
\begin{equation}
L_{\xi}X_{a}=\left[  (\ln X)_{,b}\xi^{b}+\psi+\frac{\varepsilon(X)}{X^{2}%
}H_{cb}X^{c}X^{b}\right]  X_{a}+\hat{V}_{a}(X) \label{Ldd.4e}%
\end{equation}
in which%
\begin{equation}
\hat{V}_{a}(X)=2h_{a}^{b}(X)H_{bc}X^{c}+V_{a}\left(  X\right)  .
\label{Ldd.4f.11}%
\end{equation}
Note that in general $\hat{V}_{a}(X)\neq V_{a}(X)$
\end{proposition}

In case where~$X^{a}$ is normalized, i.e. $X^{2}=1$ formulae (\ref{Ldd.4a})
and (\ref{Ldd.4e}) reduce to:%
\begin{align}
L_{\xi}X^{a}  &  =-\left[  \psi+\varepsilon(X)H_{cb}X^{c}X^{b}\right]
X^{a}+V^{a}(X)\label{Ldd.4f}\\
L_{\xi}X_{a}  &  =\left[  \psi+\varepsilon(X)H_{cb}X^{c}X^{b}\right]
X_{a}+\hat{V}_{a}(X). \label{Ldd.4fa}%
\end{align}
In the special case $X^{a}=u^{a}$ it follows%
\begin{align}
L_{\xi}u^{a}  &  =-(\psi-H_{cd}X^{c}X^{d})u^{a}+V^{a}(u)\label{Ldd.4f.b}\\
L_{\xi}u_{a}  &  =(\psi-H_{cd}X^{c}X^{d})u_{a}+\hat{V}_{a}(u).
\label{Ldd.4f.c}%
\end{align}

Some important results are collected in the following proposition

\begin{proposition}
(i)\quad$\hat{V}_{a}(X)=V^{a}(X)$ \ iff $h^{ab}(X)H_{bc}X^{c}=0$ iff
$H_{bc}X^{c}=aX_{b}$ that is, $X^{a}$ is an eigenvector of $H_{bc}.$

(ii) \quad$\hat{V}^{a}(X)=V^{a}(X)$ for all non-null $X^{a}$ iff $\xi^{a}$ is
at most a CKV.

(iii) For $X^{a}=u^{a}$ we have $\hat{V}^{a}(u)=V^{a}(u)$ iff either $u^{a}$
is an eigenvector of $H_{bc}$ or $\xi^{a}$ is at most a CKV. Furthermore in
that later case holds:%
\begin{align}
L_{\xi}u^{a}  &  =-\psi u^{a}+V^{a}(u)\label{Ldd.4f.10}\\
L_{\xi}u_{a}  &  =\psi u_{a}+V_{a}(u). \label{Ldd.4f.12}%
\end{align}

\end{proposition}

These results are easily established and we omit the proof.

\subsection{The case $X^{2}=1$}

In the following we assume that $X^{a}$ is normalized, that is $X^{a}%
X_{a}=\varepsilon\left(  X\right)  $, where$~\varepsilon(X)=\pm1$,~ and we
compute the quantity $V_{a}(u)$ in terms of the \textquotedblleft
kinematic\textquotedblright\ quantities of $X^{a}.$

We find:%
\begin{equation}
L_{\xi}X_{a}=\left[  2\omega(X)_{ab}+\varepsilon(X)\overset{\ast}{(X_{a}}%
X_{b}-\overset{\ast}{X_{b}}X_{a})\right]  \xi^{b}+(X_{b}\xi^{b})_{;a}
\label{Ldd.4e.1}%
\end{equation}
where a star \textquotedblleft*\textquotedblright\ over a symbol indicates
derivation along $X^{a}$ i.e. $\overset{\ast}{X_{a}}=X_{a;b}X^{b}.$

We also write $2\omega(X)_{ab}=$ $h_{a}^{c}(X)h_{b}^{d}(X)(X_{c;d}-X_{d;c}).$
We have according to the general 1+3 decomposition formula:
\begin{equation}
X_{a;b}=h(X)_{a}^{c}h(X)_{b}^{d}X_{a;b}+\varepsilon(X)\overset{\ast}{X_{a}%
}X_{b}. \label{Ldd.4f.13}%
\end{equation}

We continue with the 1+3 decomposition of $L_{\xi}X_{a}$ along $X^{a}$. For
the parallel component we find\footnote{We remark that ${\overset{\ast}{\xi}}^{b}$
does not in general coincides with ${\overset{\ast}{\xi}}_{\,b}$}:%
\begin{equation}
X^{a}L_{\xi}X_{a}=-(\overset{\ast}{X}_{b}\xi^{b})+(X_{b}\xi^{b})^{\ast}%
=X_{b}\overset{\ast}{\xi}^{b} \label{Ldd.4e.2}%
\end{equation}
This implies the formula\footnote{Compare with (\ref{Ldd.4fa}). Also note that
$L_{\xi}X_{a}=\varepsilon(X)(X^{b}L_{\xi}X_{b})X_{a}+h_{a}^{b}L_{\xi}X_{b}$.}:%
\begin{equation}
L_{\xi}X_{a}=\varepsilon(X)\left[  -(\overset{\ast}{X}_{b}\xi^{b})+(X_{b}%
\xi^{b})^{\ast}\right]  +\hat{V}_{a}(X). \label{Ldd.4e.2a}%
\end{equation}

By comparing (\ref{Ldd.4e.1}) and (\ref{Ldd.4e.2a}) we find:%
\[
\hat{V}_{a}(X)=2\omega(X)_{ab}\xi^{b}+\varepsilon(X)\overset{\ast}{X}%
_{a}(X_{b}\xi^{b})+(X_{b}\xi^{b})_{;a}-\varepsilon(X)(X_{b}\xi^{b})^{\ast
}X_{a}.
\]
Noting that:%
\[
(X_{b}\xi^{b})_{;a}-\varepsilon(X)(X_{b}\xi^{b})^{\ast}X_{a}=h_{a}^{b}%
(X_{c}\xi^{c})_{;b}%
\]
we obtain the final result:%
\begin{equation}
\hat{V}_{a}(X)=2\omega(X)_{ab}\xi^{b}+\varepsilon(X)\overset{\ast}{X_{a}%
}(X_{b}\xi^{b})+h_{a}^{b}(X_{c}\xi^{c})_{;b}. \label{Ldd.4e.6}%
\end{equation}

It is possible to express the parallel component of the $L_{\xi}X_{a}$ in
terms of the collineation components. Indeed by Lie differentiation of:%
\[
g^{ab}X_{a}X_{b}=\varepsilon(X)
\]
we find:
\begin{align}
(L_{\xi}g^{ab})X_{a}X_{b}+2g^{ab}X_{a}L_{\xi}X_{b}  &  =0\Rightarrow
\nonumber\\
(-2\psi g^{ab}-2H^{ab})X_{a}X_{b}+2g^{ab}X_{a}L_{\xi}X_{b}  &  =0\Rightarrow
\nonumber\\
X^{b}L_{\xi}X_{b}  &  =\varepsilon(X)\psi+H^{ab}X_{a}X_{b}. \label{Ldd.4e.3}%
\end{align}
From (\ref{Ldd.4e.2}) and (\ref{Ldd.4e.3}) follows:
\begin{equation}
\varepsilon(X)\psi+H^{ab}X_{a}X_{b}=-(\overset{\ast}{X}_{b}\xi^{b})+(X_{b}%
\xi^{b})^{\ast}=X_{b}{\overset{\ast}{\xi}}^{b}. \label{Ldd.4e.4}%
\end{equation}
This relation implies the identity/decomposition:
\begin{equation}
L_{\xi}X_{a}=\left[  \psi+\varepsilon(X)H^{cd}X_{c}X_{d}\right]  X_{a}%
+2\omega(X)_{ab}\xi^{b}+\varepsilon(X)\overset{\ast}{X}_{a}(X_{b}\xi
^{b})+h_{a}^{b}(X_{c}\xi^{c})_{;b}. \label{Ldd.4e.5}%
\end{equation}

Next, we compute the Lie derivative $L_{\xi}X^{a}$. We find:%
\begin{equation}
L_{\xi}X^{a}=\left[  -\psi+\varepsilon(X)H^{cd}X_{c}X_{d}\right]
X^{a}-2H^{ab}X_{b}+2\omega(X)_{.\;b}^{a}\xi^{b}+\varepsilon(X)\overset{\ast
}{X^{a}}(X_{b}\xi^{b})+h^{ab}(X_{c}\xi^{c})_{;b}. \label{Ldd.4e.7}%
\end{equation}
But ($X^{2}=1)$:%
\[
H^{ab}X_{b}=\varepsilon(X)(H^{cd}X_{c}X_{d})X^{a}+h_{c}^{a}H^{cd}X_{d}%
\]
therefore:%
\begin{equation}
L_{\xi}X^{a}=\left[  -\psi-\varepsilon(X)H^{cd}X_{c}X_{d}\right]  X^{a}%
-2h_{c}^{a}H^{cd}X_{d}+2\omega(X)_{.\;b}^{a}\xi^{b}+\varepsilon(X)\overset
{\ast}{X^{a}}(X_{b}\xi^{b})+h^{ab}(X_{c}\xi^{c})_{;b}. \label{Ldd.4e.8}%
\end{equation}
Finally, by comparing with (\ref{Ldd.4fa}) we find the following constraint
equation%
\begin{equation}
V_{a}(X)=-2h_{c}^{a}H^{cd}X_{d}+2\omega(X)_{.\;b}^{a}\xi^{b}+\varepsilon
(X)\overset{\ast}{X^{a}}(X_{b}\xi^{b})+h^{ab}(X_{c}\xi^{c})_{;b}.
\label{Ldd.4e.7b}%
\end{equation}

\subsection{The special case $X^{a}=u^{a}$}

An important special case is $X^{a}=u^{a},$ that is, $X^{a}$ is unit and
timelike (the four-velocity). We write $L_{\xi}u_{a}$ and $L_{\xi}u^{a}~$given
by formulae (\ref{Ldd.4f.b}) and (\ref{Ldd.4f.c}) in terms of the kinematic
quantities of $u^{a}.$

From (\ref{Ldd.4f.b}) follows:
\begin{equation}
L_{\xi}u_{a}=\left[  \psi-H^{cd}u_{c}u_{d}\right]  u_{a}+2\omega_{ab}\xi
^{b}-\dot{u}_{a}\,(u_{b}\xi^{b})+h_{a}^{b}(u_{c}\xi^{c})_{;b} \label{K1.7}%
\end{equation}
and from (\ref{Ldd.4f.c}):%
\begin{equation}
L_{\xi}u^{a}=-\left[  \psi-H^{cd}u_{c}u_{d}\right]  u^{a}-2h_{b}^{a}%
H^{bc}u_{c}+2\omega_{.\;b}^{a}\xi^{b}-\dot{u}^{a}(u_{b}\xi^{b})+h^{ab}%
(u_{c}\xi^{c})_{;b}. \label{K1.7a}%
\end{equation}

These imply that:%
\begin{align}
\hat{V}_{a}(u)  &  =2\omega_{ab}\xi^{b}-\dot{u}_{a}\,(u_{b}\xi^{b})+h_{a}%
^{b}(u_{c}\xi^{c})_{;b}\label{K1.7aa}\\
V^{a}(u)  &  =-2h_{b}^{a}H^{bc}u_{c}+2\omega_{.\;b}^{a}\xi^{b}-\dot{u}%
^{a}(u_{b}\xi^{b})+h^{ab}(u_{c}\xi^{c})_{;b}. \label{K1.7ab}%
\end{align}
From (\ref{K1.7}) we draw the following conclusions:

1. If $u_{a}\xi^{a}$ is an acceleration potential (that is the acceleration
potential is of the form $a=-(u_{b}\xi^{b})^{-1})$ then for any collineation
$\xi^{a}$:
\begin{equation}
L_{\xi}u_{a}=\left[  \psi-H^{cd}u_{c}u_{d}\right]  u_{a}+2\omega_{ab}\xi^{b}.
\label{K1.8}%
\end{equation}

One special case of this is when $u_{a}\xi^{a}=0$ i.e. the symmetry vector is
normal to the flow vector $u^{a}.$(See \cite{Herrera et all (1984)}).

2. If $\xi^{a}$ is a KV we have:%
\begin{equation}
L_{\xi}u_{a}=2\omega_{ab}\xi^{b}-\dot{u}_{a}(u_{b}\xi^{b})+h_{a}^{b}(u_{c}%
\xi^{c})_{;b} \label{K1.14}%
\end{equation}
from which we conclude:%
\begin{equation}
L_{\xi}u_{a}=0\Longleftrightarrow2\omega_{ab}\xi^{b}-\dot{u}_{a}(u_{b}\xi
^{b})+h_{a}^{b}(u_{c}\xi^{c})_{;b}=0. \label{K1.15}%
\end{equation}

A\ CKV\ is inherited\footnote{The symmetry is inherited if $L_{\xi}%
u^{a}=\lambda u^{a}$.} by the four-velocity field $u^{a}$ if $u_{b}\xi^{b}=0$
and the fluid is either irrotational (i.e. $\omega_{\alpha b}=0)$ or the
vorticity vector $\omega^{a}\shortparallel\xi^{a}.$ (See \cite{Herrera et all
(1984)})

3. For a CKV\ the inheritance of the symmetry by $u^{a}$ (i.e.$L_{\xi}%
u_{a}=\psi u_{a}$ or $L_{\xi}u^{a}=-\psi u^{a})$ is equivalent to the identity
$h_{a}^{b}L_{\xi}dx^{a}=0$ (See \cite{Maartens Mason Tsamparlis 1985}).
Obviously for a KV/HKV/SCKV\ this result remains true.

We note that the requirement of surface forming of $\xi^{a}$ with a unit
timelike Killing vector (a static spacetime) results in the kinematic
constraint (\ref{K1.15}) which is by no means trivial. For example when
$\xi^{a}u_{a}=0$ then $\omega_{a} \xi^{a}=0$ i.e. $\omega_{a} // \xi^{a}.$

\section{The 1+3 decomposition of the Lie derivative $L_{\xi}X_{a;b}$ wrt
$X^{a}$}

This decomposition is useful because allows us to compute the Lie derivative
of the kinematic quantities along the collineation vector $\xi^{a}.$ We have
the following identity for any pair of vector fields\footnote{This identity
gives the commutation of the Lie and the covariant derivative in a Riemannian
space. It can be found in \cite{Yano 1956}.} $X^{a},\xi^{a}:$%
\begin{equation}
L_{\xi}X_{a;b}=(L_{\xi}X_{a})_{;b}-(L_{\xi}\Gamma_{ab}^{c})X_{c} \label{LDX.1}%
\end{equation}
and also the identity:%
\begin{equation}
L_{\xi}\Gamma_{ab}^{c}=\frac{1}{2}g^{cd}\left[  (L_{\xi}g_{da});_{b}+(L_{\xi
}g_{db});_{a}-(L_{\xi}g_{ab});_{d}\right]  . \label{LDX.2}%
\end{equation}

Assuming $\xi^{a}$ to be a collineation, the second identity gives:%
\[
L_{\xi}\Gamma_{ab}^{c}=g^{cd}\left[  \psi,_{b}g_{da}+\psi,_{a}g_{db}-\psi
,_{d}g_{ab}+\text{ \ \ }H_{da;b}+H_{db;a}-H_{ab;d}\right]
\]
therefore:%
\begin{equation}
L_{\xi}X_{a;b}=(L_{\xi}X_{a})_{;b}-\left[  \psi,_{b}g_{da}+\psi,_{a}%
g_{db}-\psi,_{d}g_{ab}+H_{da;b}+H_{db;a}-H_{ab;d}\right]  X^{d}.
\label{LDX.2.1}%
\end{equation}

A different expression is found as follows:%
\begin{align}
L_{\xi}\Gamma_{ab}^{c}  &  =\frac{1}{2}g^{cd}\left[  (L_{\xi}g_{da}%
)_{;b}+(L_{\xi}g_{db})_{;a}-(L_{\xi}g_{ab})_{;d}\right] \nonumber\\
&  =g^{cd}\left[  (\xi_{(d;a);b}+(\xi_{(d;b)})_{a}-(\xi_{(a;b)})_{;d}\right]
\nonumber\\
&  =g^{cd}\left[  \xi_{d;(ab)}+\xi_{a;[db]}+\xi_{b;[da]}\right] \nonumber\\
&  =g^{cd}\left[  \xi_{d;(ab)}+\frac{1}{2}R_{tadb}\xi^{t}+\frac{1}{2}%
R_{tbda}\xi^{t}\right] \nonumber\\
&  =g^{cd}\left[  \xi_{d;(ab)}+R_{t(a|d|b)}\xi^{t}\right]  . \label{LDX.2.a}%
\end{align}
Therefore,
\begin{equation}
L_{\xi}X_{a;b}=(L_{\xi}X_{a})_{;b}-g^{cd}\left[  \xi_{d;(ab)}+R_{t(a|d|b)}%
\xi^{t}\right]  X_{c} \label{LDX.2.b}%
\end{equation}
where we have defined the Riemann tensor with the Ricci identity:%
\begin{equation}
2\xi_{a;[bc]}=R_{tabc}\xi^{t}. \label{LDX.2.b.1}%
\end{equation}

\subsection{The case $X^{2}=1$}

We consider the case $X^{2}=1$ and find from (\ref{Ldd.4f}):%
\begin{align}
(L_{\xi}X_{a})_{;b}  &  =\left[  \psi+\varepsilon(X)H_{cd}X^{c}X^{d}%
]X_{a}+\hat{V}_{a}(X)\right]  ;_{b}\nonumber\\
&  =\left[  \psi+\varepsilon(X)H_{cd}X^{c}X^{d}\right]  _{;b}X_{a}+\left[
\psi+\varepsilon(X)H_{cd}X^{c}X^{d}\right]  X_{a;b}+\hat{V}_{a}(X)_{;b}
\label{LDX.2.c}%
\end{align}
where:%
\begin{equation}
\hat{V}_{a}(X)=2\omega(X)_{ab}\xi^{b}+\varepsilon(X)\overset{\ast}{X_{a}%
}(X_{b}\xi^{b})+h_{a}^{b}(X_{c}\xi^{c})_{;b}. \label{LDX.2.d}%
\end{equation}
To save writing we set:%
\begin{equation}
K(X)=\psi+\varepsilon(X)H_{cd}X^{c}X^{d} \label{LDX.2.e}%
\end{equation}
and have:%
\[
(L_{\xi}X_{a})_{;b}=K(X)_{;b}X_{a}+K(X)X_{a;b}+\hat{V}_{a}(X)_{;b}.
\]
Then using (\ref{LDX.2.c}) we get:%
\begin{align}
L_{\xi}X_{a;b}  &  =\left[  K(X)_{;b}X_{a}+K(X)X_{a;b}+\hat{V}_{a}%
(X)_{;b}\right]  +\label{LDX.2.f}\\
&  -\left[  \psi,_{b}g_{da}+\psi,_{a}g_{db}-\psi,_{d}g_{ab}+H_{da;b}%
+H_{db;a}-H_{ab;d}\right]  X^{d}\nonumber
\end{align}
and also:
\begin{equation}
L_{\xi}X_{a;b}=\left[  K(X)_{;b}X_{a}+K(X)X_{a;b}+\hat{V}_{a}(X)_{;b}\right]
-\left[  \xi_{d;(ab)}+R_{t(a|d|b)}\xi^{t}\right]  X^{d}. \label{LDX.2.g}%
\end{equation}

\subsection{The case $X^{a}$ is a unit timelike vector field, $i.e.$ $X^{a}=u^{a}$}

When $X^{a}=u^{a}$ from (\ref{LDX.2.f}) we find:%
\begin{align}
L_{\xi}u_{a;b}  &  =K(u)_{;b}u_{a}+K(u)u_{a;b}+\hat{V}_{a}(u)_{;b}%
+\label{LDX.4}\\
&  -\left[  \psi,_{b}g_{da}+\psi,_{a}g_{db}-\psi,_{d}g_{ab}+H_{da;b}%
+H_{db;a}-H_{ab;d}\right]  u^{d}.
\end{align}
where%
\begin{equation}
K=\psi-H_{cd}u^{c}u^{d},\text{ }\hat{V}_{a}=\hat{V}_{a}(u). \label{LDX.4.a}%
\end{equation}

We consider the 1+3 decomposition $u_{a;b}=\sigma_{ab}+\omega_{ab}+\frac{1}%
{2}\theta h_{ab}-\dot{u}_{a}u_{b}$ and the lhs of (\ref{LDX.4}) becomes:%
\[
L_{\xi}\sigma_{ab}+L_{\xi}\omega_{ab}+\frac{1}{2}L_{\xi}\theta h_{ab}+\frac
{1}{2}\theta L_{\xi}h_{ab}-(L_{\xi}\dot{u}_{a})u_{b}-\dot{u}_{a}L_{\xi}u_{b}%
\]
therefore we can compute the quantities $L_{\xi}\sigma_{ab},L_{\xi}\omega
_{ab},L_{\xi}\theta,L_{\xi}\dot{u}_{a}$ in terms of the collineation
parameters $\psi,H_{ab}$ by taking the irreducible parts of equation
(\ref{LDX.4}). Before we do that we need to calculate the Lie derivative of
the projection tensor $h_{ab}.$

\subsection{Calculation of $L_{\xi}h_{ab}$}

For the Lie derivative of the projection tensor we find the results:%
\begin{align}
L_{\xi}h_{ab}(X)  &  =L_{\xi}\left(  g_{ab}-\frac{\varepsilon(X)}{X^{2}}%
X_{a}X_{b}\right)  =2\psi h_{ab}\left(  X\right)  +2H_{ab}-2\frac{1}{X^{4}%
}\left(  H_{cd}X^{c}X^{d}\right)  X_{(a}X_{b)}+\label{LDHX.1}\\
&  ~\ \ -2\frac{\varepsilon(X)}{X^{2}}X_{(a}\hat{V}_{b)}(X)-2\frac
{\varepsilon(X)}{X^{2}}\left[  (\ln X)_{,c}\xi^{c}\right]  X_{(a}%
X_{b)}\nonumber
\end{align}
\begin{align}
L_{\xi}h_{b}^{a}\left(  X\right)   &  =L_{\xi}\left(  g_{b}^{a}-\frac
{\varepsilon(X)}{X^{2}}X^{a}X_{b}\right) \nonumber\\
&  =-\frac{\varepsilon(X)}{X^{2}}\left[  2(\ln X)_{,c}\xi^{c}X^{a}X_{b}%
+V^{a}(X)X_{b}+X^{a}\hat{V}_{b}(X)\right]  . \label{LDHX.2}%
\end{align}
In the case where the vector field is unit $\left(  X^{2}=1\right)  $
equations (\ref{LDHX.1}) and (\ref{LDHX.2}) reduce to:%
\begin{equation}
L_{\xi}h_{ab}(X)=2\psi h_{ab}\left(  X\right)  +2H_{ab}-2\left(  H_{cd}%
X^{c}X^{d}\right)  X_{(a}X_{b)}-2\varepsilon\left(  X\right)  X_{(a)}\hat
{V}_{b)}(X) \label{LDHX.3}%
\end{equation}%
\begin{equation}
L_{\xi}h_{b}^{a}\left(  X\right)  =-\varepsilon\left(  X\right)
[V^{a}(X)X_{b}+X^{a}\hat{V}_{b}(X)]. \label{LDHX.4}%
\end{equation}
In the special case $X^{a}$ is the 4-velocity $u^{a}$ the above relations become:%

\begin{equation}
L_{\xi}h_{ab}=2\psi h_{ab}+2H_{ab}-2\left(  H_{cd}u^{c}u^{d}\right)
u_{(a}u_{b)}+2u_{(a}\hat{V}_{b)} \label{LDHX.5}%
\end{equation}%
\begin{equation}
L_{\xi}h_{b}^{a}=u^{a}\hat{V}_{b}+V^{a}u_{b} \label{LDHX.6}%
\end{equation}

Replacing $\hat{V}_{b},V^{a}$ \ from (\ref{K1.7aa}), (\ref{K1.7ab}) we find:%
\begin{align}
L_{\xi}h_{ab}  &  =2\psi h_{ab}+2H_{ab}-2\left(  H_{cd}u^{c}u^{d}\right)
u_{(a}u_{b)}+2u_{(a}[2\omega_{b)c}\xi^{c}-\dot{u}_{b)}\,(u_{c}\xi^{c}%
)+h_{b)}^{d}(u_{c}\xi^{c})_{;d}]\label{LDHX.6a}\\
L_{\xi}h_{b}^{a}  &  =u^{a}[2\omega_{bc}\xi^{c}-\dot{u}_{b}\,(u_{c}\xi
^{c})+h_{b}^{d}(u_{c}\xi^{c})_{;d}]+[-2h_{d}^{a}H^{dc}u_{c}+2\omega_{.\;c}%
^{a}\xi^{c}-\dot{u}^{a}(u_{c}\xi^{c})+h^{ad}(u_{c}\xi^{c})_{;d}]u_{b}.
\label{LDHX.6b}%
\end{align}
Concerning $h^{ab}$ we find:%
\begin{align}
L_{\xi}h^{ab}  &  =-g^{ac}g^{bd}L_{\xi}h_{cd}\nonumber\\
&  =-2\psi h^{ab}-2H^{ab}-g^{ac}g^{bd}\left[  -2\left(  H_{mn}u^{m}%
u^{n}\right)  u_{c}u_{d}+2u_{(c}[2\omega_{d)m}\xi^{m}-\dot{u}_{d)}\,(u_{n}%
\xi^{n})+h_{d)}^{m}(u_{n}\xi^{n})_{;m}]\right]  \label{LDHX.6e}%
\end{align}

\begin{proposition}
If $\xi^{a}$ is a spacelike CKV the following formulae are true:%
\begin{align}
L_{\xi}h_{ab}  &  =2\psi h_{ab}+4u_{(a}\omega_{b)c}\xi^{c},\label{LDHX.6c}\\
L_{\xi}h_{b}^{a}  &  =2(u^{a}\omega_{bc}+\omega_{.\;c}^{a}u_{b})\xi
^{c},\label{LDHX.6d}\\
L_{\xi}h^{ab}  &  =-2\psi h^{ab}-2H^{ab}-4g^{ac}g^{bd}u_{(c}\omega_{d)m}%
\xi^{m}. \label{LDHX.6f}%
\end{align}

\end{proposition}

\subsection{Calculation of $L_{\xi}\omega_{ab},~~L_{\xi}\sigma_{ab},~~L_{\xi
}\theta,~~L_{\xi}\dot{u}_{a}$}

We compute now the quantities $L_{\xi}\sigma_{ab},L_{\xi}\omega_{ab},L_{\xi
}\theta,L_{\xi}\dot{u}_{a}$ in terms of the collineation parameters
$\psi,H_{ab}.$ We have\footnote{Recall that
\begin{equation}
L_{\xi}u_{a}=\left[  \psi-H^{cd}u_{c}u_{d}\right]  u_{a}+2\omega_{ab}\xi
^{b}-\dot{u}_{a}\,(u_{b}\xi^{b})+h_{a}^{b}(u_{c}\xi^{c})_{;b}%
\end{equation}
\par
and from (\ref{Ldd.4e.7}):%
\begin{equation}
L_{\xi}u^{a}=-\left[  \psi+H^{cd}u_{c}u_{d}\right]  u^{a}-2h_{b}^{a}%
H^{bc}u_{c}+2\omega_{.\;b}^{a}\xi^{b}-\dot{u}^{a}(u_{b}\xi^{b})+h^{ab}%
(u_{c}\xi^{c})_{;b}. \label{K1.7ae}%
\end{equation}
\par
where%
\begin{align}
\hat{V}_{a}(u)  &  =2\omega_{ab}\xi^{b}-\dot{u}_{a}\,(u_{b}\xi^{b})+h_{a}%
^{b}(u_{c}\xi^{c})_{;b}\\
V^{a}(u)  &  =-2h_{b}^{a}H^{bc}u_{c}+2\omega_{.\;b}^{a}\xi^{b}-\dot{u}%
^{a}(u_{b}\xi^{b})+h^{ab}(u_{c}\xi^{c})_{;b}.
\end{align}
} $(\hat{V}=\hat{V}(u),V=V(u)):$%
\begin{align}
L_{\xi}\left(  h_{a}^{c}h_{b}^{d}u_{c;d}\right)   &  =h_{a}^{c}h_{b}%
^{d}\left(  L_{\xi}u_{c;d}\right)  +h_{a}^{c}(L_{\xi}h_{b}^{d})u_{c;d}%
+h_{b}^{d}(L_{\xi}h_{a}^{c})u_{c;d}\nonumber\\
&  =h_{a}^{c}h_{b}^{d}\left[  \left[  \psi-H_{mn}u^{m}u^{n}\right]  _{;d}%
u_{c}+\left[  \psi-H_{mn}u^{m}u^{n}\right]  u_{d;c}+\hat{V}_{c;d}\right]
\nonumber\\
&  -h_{a}^{c}h_{b}^{d}\left[  \psi,_{c}g_{ed}+\psi,_{d}g_{ec}-\psi,_{e}%
g_{cd}+H_{ec;d}+H_{ed;c}-H_{cd;e}\right]  u^{e}\nonumber\\
&  +h_{a}^{c}\left[  u^{d}\hat{V}_{b}+V^{d}u_{b}\right]  u_{c;d}+h_{b}%
^{d}\left[  u^{c}\hat{V}_{a}+V^{c}u_{a}\right]  u_{c;d}\nonumber\\
&  =h_{a}^{c}h_{b}^{d}\left(  \psi-H_{mn}u^{m}u^{n}\right)  u_{c;d}+h_{a}%
^{c}h_{b}^{d}\hat{V}_{c;d}\label{LDHX.7}\\
&  +h_{a}^{c}u_{b}u_{c;d}V^{d}+h_{b}^{c}u_{a}u_{d;c}V^{d}+\dot{u}_{a}\hat
{V}_{b}\nonumber\\
&  +\dot{\psi}h_{ab}-h_{a}^{c}h_{b}^{d}\left[  H_{ec;d}+H_{ed;c}%
-H_{cd;e}\right]  u^{e}.\nonumber
\end{align}
We consider the symmetric and the antisymmetric parts of $u_{a;b}.\,$\ For the
antisymmetric part we find:%
\begin{align}
L_{\xi}\omega_{ab}  &  =h_{a}^{c}h_{b}^{d}L_{\xi}u_{[c;d]}\Rightarrow
\nonumber\\
L_{\xi}\omega_{ab}  &  =\left(  \psi-H_{mn}u^{m}u^{n}\right)  \omega
_{ab}+h_{a}^{c}h_{b}^{d}\hat{V}_{[c;d]}+\dot{u}_{[a}\hat{V}_{b]}-2u_{[a}%
\omega_{b]c}\hat{V}^{c} \label{LDHX.8}%
\end{align}

For the symmetric part $\theta_{ab}=u_{(a;b)}$ we find:%
\begin{align}
L_{\xi}\theta_{ab}  &  =L_{\xi}\left(  h_{a}^{c}h_{b}^{d}u_{(c;d)}\right)
\Rightarrow\nonumber\\
L_{\xi}\theta_{ab}  &  =\left(  \psi-H_{mn}u^{m}u^{n}\right)  \theta
_{ab}+h_{a}^{c}h_{b}^{d}\hat{V}_{(c;d)}+2u_{(a}\theta_{b)c}V^{c}+\dot{u}%
_{(a}\hat{V}_{b)}\nonumber\\
&  +\dot{\psi}h_{ab}-h_{(a}^{c}h_{b)}^{d}\left[  H_{ec;d}+H_{ed;c}%
-H_{cd;e}\right]  u^{e}. \label{LDHX.9a}%
\end{align}

We decompose $\theta_{ab}$ in trace and traceless part as follows
\begin{equation}
\theta_{ab}=\sigma_{ab}+\frac{1}{3}h_{ab}\theta.
\end{equation}

For the trace $\theta~$we find:%
\begin{align}
&  L_{\xi}\theta=-\left(  \psi+H_{mn}u^{m}u^{n}\right)  \theta-2H^{ab}%
\theta_{ab}+h^{cd}\hat{V}_{(c;d)}+2\omega_{ab}\dot{u}^{a}\xi^{b}-(\dot{u}%
_{a}\dot{u}^{a})\,(u_{b}\xi^{b})\nonumber\\
&  +\dot{u}^{a}(u_{c}\xi^{c})_{;a}+3\dot{\psi}-h^{cd}[H_{ec;d}+H_{ed;c}%
-H_{cd;e}]u^{e}, \label{LDHX.9c}%
\end{align}
The term$~2H^{ab}\theta_{ab}$ becomes
\[
2H^{ab}\theta_{ab}=2H^{ab}\sigma_{ab}+\frac{2}{3}H^{ab}h_{ab}\theta
=2H^{ab}\sigma_{ab}+\frac{2}{3}H^{mn}u_{m}u_{n}\theta.
\]
Therefore,
\begin{align}
L_{\xi}\theta &  =-\left(  \psi+\frac{5}{3}H_{mn}u^{m}u^{n}\right)
\theta-2H^{ab}\sigma_{ab}+h^{cd}\hat{V}_{(c;d)}+2\omega_{ab}\dot{u}^{a}\xi
^{b}\nonumber\\
&  -(\dot{u}_{a}\dot{u}^{a})\,(u_{b}\xi^{b})+\dot{u}^{a}(u_{c}\xi^{c}%
)_{;a}+3\dot{\psi}-h^{cd}[H_{ec;d}+H_{ed;c}-H_{cd;e}]u^{e}. \label{LDHX.9d}%
\end{align}

Concerning the four-acceleration~$\dot{u}_{a}$ it follows%
\begin{equation}
L_{\xi}\dot{u}_{a}=\dot{K}(u)\text{ }u_{a}+\hat{V}_{a}(u)_{;b}u^{b}%
+V^{b}(u)u_{a;b}+\psi,_{a}-\left[  H_{da;b}+H_{db;a}-H_{ab;d}\right]
u^{d}u^{b}%
\end{equation}
where $K(u)=\psi-H_{mn}u^{m}u^{n}.$

Finally for the shear $\sigma_{ab.}$ we have%
\begin{align}
L_{\xi}\sigma_{ab}  &  =K(u)\sigma_{ab}-\frac{1}{3}\left(  \frac{2}{3}%
H_{mn}u^{m}u^{n}\right)  \theta h_{ab}+h_{a}^{c}h_{b}^{d}\hat{V}%
_{(c;d)}+2u_{(a}\theta_{b)c}V^{c}+\dot{u}_{(a}\hat{V}_{b)}\nonumber\\
&  -h_{(a}^{c}h_{b)}^{d}\left[  H_{ec;d}+H_{ed;c}-H_{cd;e}\right]  u^{e}%
-\frac{1}{3}\left[  -2H_{ab}-2\left(  H_{cd}u^{c}u^{d}\right)  u_{(a}%
u_{b)}+2u_{(a}\hat{V}_{b)}\right]  \theta\nonumber\\
&  -\frac{1}{3}\left[
\begin{array}
[c]{c}%
-2H^{mn}\sigma_{mn}+h^{mn}\hat{V}_{(m;n)}+2\omega_{mn}\dot{u}^{m}\xi^{n}%
-(\dot{u}_{m}\dot{u}^{m})\,(u_{n}\xi^{n})\\
\\
+\dot{u}^{m}(u_{n}\xi^{n})_{;m}-h^{mn}[H_{em;n}+H_{en;m}-H_{mn;e}]u^{e}%
\end{array}
\right]  h_{ab} \label{LDHX.12}%
\end{align}

\section{The matter - Dynamics}

\label{The matter - Dynamics}

General Relativity studies the physics in spacetime in the same way it is done
in Special Relativity, that is, each physical field is described by a
symmetric tensor $T_{ab}$ which is called the energy momentum tensor of the
specific field. If the field is described by a Lagrangian $L$ then its energy
momentum tensor in a model theory with metric tensor $g_{ab}$ is computed from
the relation%
\begin{equation}
T_{ab}=\frac{\delta L}{\delta g_{ab}}.%
\end{equation}
In case in a given situation there are various fields interacting with the
gravitational field then the system of the fields is described by the sum of
all the energy momentum tensors of all the fields involved, while the
interaction with the gravitational field is described by the "conservation
equation\ "
\begin{equation}
_{total}T_{\phantom{...};b}^{ab}=0. \label{Key.19}%
\end{equation}
Of course in this approach each field still has its own field equations which
have to be satisfied in the model spacetime considered. For example a scalar
field $\phi$ given by the Action Integral%
\[
S=\int L\left(  g_{\mu\nu},\phi,\phi_{;k}\right)  d^{4}x
\]
which leads to the Klein Gordon equation
\begin{equation}
D_{i}\left(  \frac{\partial L}{\partial\phi_{;i}}\right)  -\frac{\partial
L}{\partial\phi}=0.
\end{equation}
In GR the gravitational field is assumed to be created by all types of matter
in spacetime{\LARGE .} The energy momentum tensor of matter is the total
energy momentum tensor $_{total}T^{ab}.$ In GR\ it is also assumed that the
totality of physical fields (including the gravitational field)\ formulate the
geometry (that is the metric $g_{ab})$ of spacetime via Einstein field
equations:%
\begin{equation}
G_{ab}=kT_{ab} \label{Key.20}%
\end{equation}
where $G_{ab}$ is the Einstein tensor (of the particular $g_{ab}$ which is
unknown!). The conservation equation (\ref{Key.19}) leads to the geometric
condition $G_{;b}^{ab}=0$ \ which is the second Bianchi identity, hence no new
constraints are introduced. The constant $k$ is a universal constant because
it relates two entities of a different nature (i.e. geometry and the physical
fields) and it is called the Einstein gravitational constant. \ By choosing
appropriate units (called physical units) one sets without loss of generality
$k=1.$ Furthermore, by demanding that GR agrees with the Newtonian theory in
the solar system one relates $k$ with the Newtonian gravitational constant $G$
with the formula $k=\frac{8\pi G}{c^{2}}$. \

In this picture so far there is no observational status and we have a kind of
theory which is universal for all observers and all types and combinations of
matter. This is the scenario of General Relativity \cite{Saridakis Tsamparlis
1991}. We note that Einstein field equations (\ref{Key.20}) are independent of
the observer and they relate the geometry of the background space (which is
not specified!) with the matter content of the universe (which is also not
specified). Therefore they cannot be solved because one knows neither the
metric, therefore one cannot compute the $G_{ab}$ nor the mater i.e. the
energy momentum tensor $T_{ab}$. In other words equation (\ref{Key.20}) is
used only to frame the general set up (i.e. the scenario) of the theory and
cannot give any information until a class of observers and the matter content
of spacetime are specified. This is the reason that GR\ is a theory of
Physics which generates gravitational models instead of a theory giving a
unique model, as it is the case with the Newtonian gravitational theory.

Before we consider the observational aspects of GR, that is, definite
gravitational models as seen by a specified class of observers, we must
address the more fundamental question, which is\footnote{The question then is:
Which class of observers?
\par
The ideal choice would be the intrinsic observers but only `God' knows who are
they. What it is left to us is to chose an arbitrary class of observers by
means of a timelike unit vector field and adding some more simplifying
assumptions write down (\ref{Key.20}) in terms of some set of differential
equations whose solutions will give us an indication of spacetime as seen by
these observers. The experiment and the observations will show us how close to
the intrinsic observers we are and how close to `reality' our simplifying
assumptions are. In conclusion we cannot solve the field equations but we can
use them as the vehicle to make up scientific pictures of the gravitational
field and get some answers which we hope they will be as close as possible to
our measurements (=reality!). The ultimate truth is hidden within the
intrinsic observers with whom we have no touch or communication.
}

 \emph{To what degree the matter content and the geometry of spacetime are
interconnected as a result of the field equations (\ref{Key.20})?}

The answer to this question has two parts. The first is due to the inherent
identities of the Riemannian geometry of spacetime and the second is the result
of the assumed collineations of spacetime. Concerning the first we see that
due to the contracted Bianchi identity (\ref{BI.03}) \ Einstein field
equations (\ref{Key.20})\ imply that $T_{ab}$ is symmetric\footnote{Or at
least only the symmetric part of $T_{ab}$ enters the field equations and
determines the gravitational field.} and that that the conservation equations
(\ref{Key.19}) are satisfied as a constraint / identity independently of the
observers which are used to build a gravitational model.

Concerning the identities from the collineations assumed by a specific model
we note that $L_{\xi}G_{ab}$ can be expressed in terms of the generic
collineation $L_{\xi}g_{ab}$ and subsequently in terns of the collineation
parameters $\psi,H_{ab}.$ The general expression in terms of the trace and the
traceless part is:%
\begin{align}
L_{\xi}G_{ab}  &  =-\left(  n-2\right)  A_{ab}+2K_{ab}-\square H_{ab}%
-RH_{ab}+\nonumber\\
&  +\frac{1}{n}g_{ab}\left[  \left(  n-1\right)  \left(  n-2\right)
\square\psi-\left(  n-2\right)  H_{.;cd}^{cd}+nR_{cd}H^{cd}\right]
\label{Key.22a}%
\end{align}
where:%
\begin{equation}
A_{ab}=\psi_{;ab}-\frac{1}{n}g_{ab}\square\psi~\text{and }K_{ab}=H_{.\left(
a;b\right)  d}^{d}-\frac{1}{n}g_{ab}H_{.;cd}^{cd} \label{Key.22}%
\end{equation}

Then Einstein field equations (\ref{Key.20}) relate $L_{\xi}T_{ab}$ with the
collineation parameters $\psi,H_{ab}$ which implies that the `symmetries' of
$T_{ab}$ as a two index tensor are related to the symmetries of the metric
$g_{ab}.$ In this manner when we make a symmetry assumption in spacetime we
simultaneously impose a constraint in the form of a symmetry to the energy
momentum tensor, therefore we restrict the possible forms of matter this model
spacetime can support. In a sense by an external assumption at the level of
symmetry we constraint both the geometry and the matter content of spacetime.

To bring all the above at a level which can be used in practice, that is
observation, we must introduce a class of observers $u^{a}.$ The choice of
these observers is our choice and there is no guarantee whatsoever that they
are, or they will be, the intrinsic observers of spacetime. However this is not
to discourage us because it is the most we can do. The introduction of
observers immediately introduces a kinematics in spacetime and all the
considerations of the previous sections apply.

\section{The role of observers in GR}

As it has been remarked Einstein field equations provide the canvas on which
the various gravitational models of GR are created. The construction of a
gravitational model requires two necessary elements/actions:

a. Additional geometric assumptions which are made beyond the internal
geometric identities of the Riemannian structure of spacetime and

b. The choice of a class of observers, that is a timelike unit four vector
field $u^{a}$ which defines the four-velocity of the observers.

Note that in this scenario the energy momentum tensor $T_{ab}$ of matter and
the other fields are not effected. What it changes is the 1+3 decomposition of
$T_{ab}$ in irreducible parts which encounter the observation by the specific
observers of the physical variables corresponding to each irreducible part.

\subsection{The 1+3 decomposition of the energy momentum tensor: The dynamic
variables}

We apply the general formula (\ref{d.one.3}) in the case $Y_{ab}=T_{ab}$ where
$T_{ab}=T_{ba}$ is the energy momentum tensor. This defines the irreducible
parts (tensors):%
\begin{align}
\mu &  =T_{ab}u^{a}u^{b}\label{TE.1}\\
p  &  =\frac{1}{3}h^{ab}T_{ab}\label{TE.2}\\
q^{a}  &  =h^{ab}T_{bc}u^{c}\label{TE.3}\\
\pi_{ab}  &  =(h_{a}^{r}h_{b}^{s}-\frac{1}{3}h_{ab}h^{rs})T_{rs} \label{TE.4}%
\end{align}
and we have the covariant decomposition/identity:
\begin{equation}
T_{ab}=\mu u_{a}u_{b}+ph_{ab}+2q_{(a}u_{b)}+\pi_{ab}. \label{TE.5}%
\end{equation}

We note that in this decomposition $T_{ab}$ is described by two scalar fields
$(\mu,p)$ \ one spacelike vector $(q^{a},q_{a}u^{a}=0)$ \ and a traceless
symmetric 2-tensor $(\pi_{ab},g_{ab}\pi^{ab}=0).$

These quantities we call the physical variables and assume that they represent
the mass density, the isotropic pressure, the heat flux and the traceless
stress tensor respectively \textbf{as measured by the observers} $u^{a}$. Due
to this decomposition and the assumed physical interpretation we consider the
\ types of `gravitating fluids' given in the Table \ref{table01}.%

\begin{table}[tbp] \centering
\caption{Types of Energy - Momentum Tensors
}%
\begin{tabular}
[c]{cccccc}\hline\hline
$\mathbf{\mu}$ & $\mathbf{p}$ & $\mathbf{q}^{a}$ & $\mathbf{\pi}^{ab}$ &
$\mathbf{T}_{ab}$ & \textbf{Type of fluid}\\\hline
$0$ & $0$ & $0$ & $0$ & $0$ & Empty space\\
$\neq0$ & $0$ & $0$ & $0$ & $T_{ab}=\mu u_{a}u_{b}$ & Dust\\
$\neq0$ & $\neq0$ & $0$ & $0$ & $T_{ab}=\mu u_{a}u_{b}+ph_{ab}$ & Perfect
fluid\\
$\neq0$ & $\neq0$ & $\neq0$ & $0$ & $T_{ab}=\mu u_{a}u_{b}+ph_{ab}%
+2q_{(a}u_{b)}$ & Isotropic non-perfect fluid\\
$\neq0$ & $\neq0$ & $0$ & $\neq0$ & $T_{ab}=\mu u_{a}u_{b}+ph_{ab}+\pi_{ab}.$
& Anisotropic fluid without heat flux\\
$\neq0$ & $\neq0$ & $\neq0$ & $\neq0$ & $T_{ab}=\mu u_{a}u_{b}+ph_{ab}%
+2q_{(a}u_{b)}+\pi_{ab}.$ & General anisotropic fluid\\\hline\hline
\end{tabular}
\label{table01}%
\end{table}%

\subsection{The 1+3 decomposition of the conservation equations}

The conservation equations follow from the contracted Bianchi identity and
Einstein field equations$.$ By 1+3 decomposing $T_{;b}^{ab}=0$ wrt $u^{a}$ we
find the conservation equations as two sets of equations, one set resulting
form the projection along the direction $u^{a}$ \ and one normal to $u^{a}.$
We replace $T^{ab}$ form (\ref{TE.5}) \ in terms of the physical variables and
find the two conservation equations\footnote{We write as usually the derivative
along $u^{a}$ by a dot, e.g. $\dot{A}^{a}=A_{;b}^{a}u^{b}.$}:%
\begin{equation}
\dot{\mu}+(\mu+p)\theta+q_{;b}^{b}+q^{a}\dot{u}_{a}+\pi^{ab}\sigma_{ab}=0.
\label{TE.6}%
\end{equation}%
\begin{equation}
(\mu+p)\dot{u}_{a}+h_{a}^{c}(p_{,c}+\dot{q}_{c}+\pi_{\text{ }c\text{ };d}%
^{d})+q^{c}\left(  \omega_{ac}+\sigma_{ac}+\frac{4}{3}\theta h_{ac}\right)
=0. \label{TE.7}%
\end{equation}

\section{The physical role of the propagation and the constraint equations}

The propagation and the constraint equations have been written in terms of the
kinematic variables and the Ricci tensor. However with the introduction of the
physical variables by the 1+3 decomposition of the energy momentum tensor we
can use Einstein equations and replace $R_{ab}$ in terms of the physical
variables. Then we have the complete physical role of the propagation and the
constraint equations.

Einstein field equations provide%
\begin{equation}
R_{ab}=T_{ab}-\frac{1}{2}g_{ab}T+\Lambda g_{ab}. \label{PE.7}%
\end{equation}
where $T_{ab}=\mu u_{a}u_{b}+ph_{ab}+2q_{(a}u_{b)}+\pi_{ab}$ \ by (\ref{TE.5})
and the trace:%
\begin{equation}
T=T^{a}{}_{a}=-\mu+3p. \label{PE.8}%
\end{equation}
Replacing in (\ref{PE.7}) we find $R_{ab}$ in terms of the physical variables:%
\begin{equation}
R_{ab}=\mu u_{a}u_{b}+ph_{ab}+2q_{(a}u_{b)}+\pi_{ab}+\frac{1}{2}g_{ab}%
(\mu-3p+2\Lambda) \label{PE.9}%
\end{equation}
or:%
\begin{equation}
R_{ab}=(\mu+p)u_{a}u_{b}+\frac{1}{2}(\mu-p+2\Lambda)g_{ab}+2q_{(a}u_{b)}%
+\pi_{ab}. \label{PE.9a}%
\end{equation}
From (\ref{PE.9}) \ follows:%
\begin{align}
R_{ab}u^{a}u^{b}  &  =\mu-\frac{1}{2}(\mu-3p)-\Lambda=\frac{1}{2}\mu+\frac
{3}{2}p-\Lambda\label{PE.10}\\
R_{ab}u^{a}  &  =-\mu u_{b}-\frac{1}{2}q_{b}+\frac{1}{2}u_{b}(\mu
-3p+2\Lambda)\label{PE.11}\\
R  &  =R^{a}{}_{a}=\mu-3p+4\Lambda\label{PE.12}%
\end{align}
where $R=g^{ab}R_{ab}$ is the curvature scalar.

We\ consider now the constraint and propagation equations and we divide them
in two sets. One set contains the ones which do not contain the Curvature
tensor and the Ricci tensor, therefore they are independent of the field
equations (and consequently the matter content of the universe) and they are
purely kinematic equations. The second set contains the equations which depend
on the Ricci tensor and the field equations and are the dynamical propagation
and constraint equations.

\subsection{The dynamical propagation and constraint equations}

The dynamical propagation equations\footnote{The remaining propagation
equation (\ref{CPE.22}) giving the propagation of the vorticity, does not
involve the Ricci tensor hence it is a kinematical equation.} are equation
(\ref{CPE.23}) and equation \ (\ref{CPE.24}). Replacing the Ricci tensor from
(\ref{PE.9}) we find for the first:%
\begin{equation}
\dot{\theta}+\frac{1}{3}\theta^{2}+2(\sigma^{2}-\omega^{2})-\dot{u}^{a}{}%
_{;a}=-\frac{1}{2}(\mu+3p)+\Lambda. \label{CPE.23a}%
\end{equation}
This is known as the\textbf{\ Raychaudhuri equation}.

Concerning the propagation equation \ \ (\ref{CPE.24}) \ we have:%
\begin{align*}
&  h_{s}^{a}h_{t}^{b}\left[  \dot{\sigma}_{ab}-\dot{u}_{(a;b)}\right]
+\sigma_{(s|c|}\sigma_{\text{ \thinspace}t)}^{c}+\frac{2}{3}\sigma_{st}%
\theta+\omega_{s}\omega_{t}-\dot{u}_{s}\dot{u}_{t}-\frac{1}{3}\left(
2\sigma^{2}+\omega^{2}-\dot{u}_{;b}^{b}\right)  h_{st}\\
&  =-E_{st}+\frac{1}{2}\left(  h_{s}^{a}h_{t}^{b}-\frac{1}{3}h_{st}%
h^{ab}\right)  R_{ab}\\
&  =-E_{st}+\frac{1}{2}\left(  h_{s}^{a}h_{t}^{b}-\frac{1}{3}h_{st}%
h^{ab}\right)  \left(  \mu u_{a}u_{b}+ph_{ab}+2q_{(a}u_{b)}+\pi_{ab}+\frac
{1}{2}g_{ab}(\mu-3p+2\Lambda)\right) \\
&  =-E_{st}+\frac{1}{2}\left(  h_{s}^{a}h_{t}^{b}-\frac{1}{3}h_{st}%
h^{ab}\right)  \pi_{ab}=-E_{st}+\frac{1}{2}\pi_{st}.
\end{align*}
Therefore \ in terms of the physical variables the propagation equation
(\ref{CPE.24}) reads:%
\begin{equation}
h_{s}^{a}h_{t}^{b}\left[  \dot{\sigma}_{ab}-\dot{u}_{(a;b)}\right]
+\sigma_{(s|c|}\sigma_{\text{ \thinspace}t)}^{c}+\frac{2}{3}\sigma_{st}%
\theta+\omega_{s}\omega_{t}-\dot{u}_{s}\dot{u}_{t}-\frac{1}{3}\left(
2\sigma^{2}+\omega^{2}-\dot{u}_{;b}^{b}\right)  h_{st}=-E_{st}+\frac{1}{2}%
\pi_{st}. \label{CPE.23b}%
\end{equation}
This equation gives the propagation of shear.

Moreover, the only dynamic constraint equation is equation (\ref{CPE.26}). We
find:%
\begin{equation}
h_{s}^{c}\left[  \frac{2}{3}\theta,_{c}-h^{ab}\sigma_{ac;b}-\eta_{camn}%
u^{a}\left(  \omega^{m;n}+2\omega^{m}\dot{u}^{n}\right)  \right]  =q_{b}.
\label{CPE.26a}%
\end{equation}

\section{The dynamic role of collineations}

As we remarked in section \ref{The matter - Dynamics} the collineation
parameters restrict the possible forms of $T_{ab.}$ Indeed Einstein field
equations and (\ref{Key.22a}) give:%
\begin{align}
L_{\xi}T_{ab}  &  =-\left(  n-2\right)  A_{ab}+2K_{ab}-\square H_{ab}%
-RH_{ab}+\nonumber\\
&  +\frac{1}{n}g_{ab}\left[  \left(  n-1\right)  \left(  n-2\right)
\square\psi-\left(  n-2\right)  H_{.;cd}^{cd}+nR_{cd}H^{cd}\right]
\label{Key.27}%
\end{align}
where the tensors $A_{ab},K_{ab}$ are given in (\ref{Key.22}). From this
equation it is possible to give the effect of a collineation directly to the
physical parameters of the observers $u^{a}.$ Indeed if we replace in
(\ref{Key.27}) $T_{ab}$ in $L_{\xi}T_{ab}$ using (\ref{TE.5}) and take the
irreducible parts, we express the Lie derivative along the collineation vector
of the dynamical variables in terms of the collineation parameters. This form
of the field equations will relate directly the physical and the geometric
variables thus enabling one to draw direct conclusions of the symmetry
assumptions at the level of dynamics. Furthermore the approach is completely
general and the results hold for all types of matter, all observers and all collineations.

From (\ref{PE.9a}) we decompose the $L_{\xi}R_{ab}$ in irreducible parts with
respect to the 1+3 decomposition defined by $u^{a}:$%
\begin{align}
L_{\xi}R_{ab}  &  =\left[  \frac{1}{2}\left(  \mu+3p\right)  ^{\circ
}-2\mathring{q}_{c}u^{c}+\mathring{\pi}_{cd}u^{c}u^{d}-\left(  \mu
+3p-2\Lambda\right)  \left(  \dot{\xi}_{c}u^{c}\right)  -2q_{c}\dot{\xi}%
^{c}\right]  u_{a}u_{b}+\nonumber\\
& \nonumber\\
& \nonumber\\
&  -2\left[
\begin{array}
[c]{c}%
-\left(  \mu+p\right)  \mathring{u}_{d}h_{c}^{d}-\mathring{q}_{d}h_{c}%
^{d}+\mathring{\pi}_{de}u^{d}h_{c}^{e}-\frac{1}{2}\left(  \mu+3p-2\Lambda
\right)  u_{c}\xi_{;d}^{c}h_{e}^{d}+\\
\\
-q_{c}\xi_{;d}^{c}h_{e}^{d}+\frac{1}{2}\left(  \mu-p+2\Lambda\right)  \dot
{\xi}_{d}h_{e}^{d}+\left(  q_{d}h_{e}^{d}\right)  \left(  \dot{\xi}^{c}%
u_{c}\right)  -h_{e}^{d}\pi_{cd}\dot{\xi}^{c}%
\end{array}
\right]  u_{(a}h_{e)}^{c}+\nonumber\\
& \nonumber\\
&  +2\left[  \frac{3}{2}\left(  \mu-p\right)  ^{\circ}+2q_{c}\mathring{u}%
^{c}+\mathring{\pi}_{cd}h^{cd}+\left(  \mu-p+2\Lambda\right)  (4\psi+\dot{\xi
}^{c}u_{c})+2u_{c}\xi_{;d}^{c}q^{d}+2\pi_{cd}H^{cd}\right]  h_{ab}\nonumber\\
& \nonumber\\
&  +2\left[
\begin{array}
[c]{c}%
\frac{1}{2}\left(  \mu-p\right)  ^{\circ}h_{ef}+2q_{(c}\mathring{u}_{d)}%
h_{e}^{c}h_{f}^{d}+\mathring{\pi}_{cd}h_{e}^{c}h_{f}^{d}+\\
\\
+\left(  \mu-p+2\Lambda\right)  \xi_{d;r}h_{e}^{d}h_{f}^{r}+2q_{r}u_{c}%
\xi_{;d}^{c}h_{e}^{d}h_{f}^{r}+2\pi_{cr}\xi_{;d}^{c}h_{e}^{d}h_{f}^{r}%
\end{array}
\right]  \left[  h_{(a}^{e}h_{b)}^{f}-\frac{1}{3}h_{ab}h^{ef}\right]
\label{Key.28}%
\end{align}
where $\circ$ means covariant derivation along $\xi^{a}$ e.g. $\mathring{\mu
}=\mu_{;a}\xi^{a}.$Also from (\ref{Key.13}) we have $($for $n=4!)$:%
\begin{equation}
L_{\xi}R_{ab}=-2\psi_{;ab}-g_{ab}\square\psi+2H_{.\left(  a;b\right)  d}%
^{d}-\square H_{ab}. \label{Key28.a}%
\end{equation}
We 1+3 decompose this expression and equate with (\ref{Key.28}) from which we
find the field equations in the form of Lie derivative of physical variables
in terms of kinematic variables and the collineation parameters.

To do that we 1+3 decompose first $\psi_{;ab}.$ We write:%
\begin{equation}
\psi_{;ab}=\lambda_{\psi}u_{a}u_{b}+p_{\psi}h_{ab}+2q_{\psi(a}u_{b)}+\pi_{\psi
ab} \label{Key.29}%
\end{equation}
where:%
\begin{equation}
\mu_{\psi}=\psi_{;ab}u^{a}u^{b},\;p_{\psi}=\frac{1}{3}\psi_{;ab}%
h^{ab},\;q_{\psi a}=-\psi_{;bc}h_{a}^{b}u^{c},\;\pi_{\psi ab}=(h_{a}^{r}%
h_{b}^{s}-\frac{1}{3}h_{ab}h^{rs})\psi_{;rs}. \label{Key.30}%
\end{equation}

We also have:%

\begin{equation}
\square\psi=\psi_{;ab}g^{ab}=-\mu_{\psi}+3p_{\psi}. \label{Key.31}%
\end{equation}
It is also possible to 1+3 decompose the collineation tensor $H_{ab}$ as
follows:%
\begin{equation}
H_{ab}=\mu_{H}u_{a}u_{b}+p_{H}h_{ab}-2q_{H(a}u_{b)}+\pi_{Hab} \label{Key.32}%
\end{equation}
where:\bigskip%
\begin{equation}
\mu_{H}=H_{ab}u^{a}u^{b},\;p_{H}=\frac{1}{3}H_{ab}h^{ab},\;q_{Ha}=-H_{bc}%
h_{a}^{b}u^{c},\;\pi_{Hab}=(h_{a}^{r}h_{b}^{s}-\frac{1}{3}h_{ab}h^{rs})H_{rs}.
\label{Key.33}%
\end{equation}
We compute:%
\begin{align*}
-2\psi_{;ab}-g_{ab}\square\psi &  =-2\left(  \mu_{\psi}u_{a}u_{b}+p_{\psi
}h_{ab}+2q_{\psi(a}u_{b)}+\pi_{\psi ab}\right)  +\left(  \mu_{\psi}-3p_{\psi
}\right)  (-u_{a}u_{b}+h_{ab})\\
&  =-3(\mu_{\psi}-p_{\psi})u_{a}u_{b}+(\mu_{\psi}-5p_{\psi})h_{ab}-4q_{\psi
(a}u_{b)}-2\pi_{\psi ab}.
\end{align*}
Concerning the second term in the rhs of (\ref{Key28.a}) we define:%
\begin{equation}
K_{ab}=2H_{.\left(  a;b\right)  d}^{d}-\square H_{ab} \label{Key.34}%
\end{equation}
and we write:%
\begin{equation}
K_{ab}=\mu_{K}u_{a}u_{b}+p_{K}h_{ab}-2q_{K(a}u_{b)}+\pi_{Kab} \label{Key.35}%
\end{equation}
where:%
\begin{equation}
\mu_{K}=K_{ab}u^{a}u^{b},\;p_{K}=\frac{1}{3}K_{ab}h^{ab},\;q_{Ka}=-K_{bc}%
h_{a}^{b}u^{c},\;\pi_{Kab}=(h_{a}^{r}h_{b}^{s}-\frac{1}{3}h_{ab}h^{rs})K_{rs}.
\label{Key.36}%
\end{equation}
Then we have:%
\begin{equation}
L_{\xi}R_{ab}=(-3\mu_{\psi}+3p_{\psi}+\mu_{K})u_{a}u_{b}+(\mu_{\psi}-5p_{\psi
}+p_{K})h_{ab}+2(-2q_{\psi(a}+q_{K(a})u_{b)}-2\pi_{\psi ab}+\pi_{Kab}.
\label{Key.37}%
\end{equation}

Comparing with (\ref{Key.28}) we write the field equations in the following
form:%
\begin{align}
\frac{1}{2}\left(  \mu+3p\right)  ^{\circ}-2\mathring{q}_{c}u^{c}%
+\mathring{\pi}_{cd}u^{c}u^{d}-\left(  \mu+3p-2\Lambda\right)  \left(
\dot{\xi}_{c}u^{c}\right)  -2q_{c}\dot{\xi}^{c}  &  =-3\mu_{\psi}+3p_{\psi
}+\mu_{K}\label{Key.38}\\
& \nonumber
\end{align}%
\begin{align}
\left[
\begin{array}
[c]{c}%
-\left(  \mu+p\right)  \mathring{u}_{d}h_{e}^{d}-\mathring{q}_{d}h_{e}%
^{d}+\mathring{\pi}_{dc}u^{d}h_{e}^{c}-\frac{1}{2}\left(  \mu+3p-2\Lambda
\right)  u_{c}\xi_{;d}^{c}h_{e}^{d}+\\
\\
-q_{c}\xi_{;d}^{c}h_{e}^{d}+\frac{1}{2}\left(  \mu-p+2\Lambda\right)  \dot
{\xi}_{d}h_{e}^{d}+\left(  q_{d}h_{e}^{d}\right)  \left(  \dot{\xi}^{c}%
u_{c}\right)  -h_{e}^{d}\pi_{cd}\dot{\xi}^{c}%
\end{array}
\right]   &  =-(-2q_{\psi e}+q_{Ke})\label{Key.39}\\
& \nonumber
\end{align}%
\begin{align}
&  \frac{3}{2}\left(  \mu-p\right)  ^{\circ}+2q_{c}\mathring{u}^{c}%
+\mathring{\pi}_{cd}h^{cd}+\left(  \mu-p+2\Lambda\right)  (4\psi+\dot{\xi}%
^{c}u_{c})+2u_{c}\xi_{;d}^{c}q^{d}+2\pi_{cd}H^{cd}\nonumber\\
&  =\frac{3}{2}(\mu_{\psi}-5p_{\psi}+p_{K}) \label{Key.40}%
\end{align}

\begin{align}
2\left[
\begin{array}
[c]{c}%
\frac{1}{2}\left(  \mu-p\right)  ^{\circ}h_{ef}+2q_{(c}\mathring{u}_{d)}%
h_{e}^{c}h_{f}^{d}+\mathring{\pi}_{cd}h_{e}^{c}h_{f}^{d}+\\
\\
+\left(  \mu-p+2\Lambda\right)  \xi_{d;r}h_{e}^{d}h_{f}^{r}+2q_{r}u_{c}%
\xi_{;d}^{c}h_{e}^{d}h_{f}^{r}+2\pi_{cr}\xi_{;d}^{c}h_{e}^{d}h_{f}^{r}%
\end{array}
\right]   &  & \left[  h_{(a}^{e}h_{b)}^{f}-\frac{1}{3}h_{ab}h^{ef}\right]
\nonumber\\
&  & =-2\pi_{\psi ab}+\pi_{Kab}. \label{Key.41}%
\end{align}

It is possible to simplify these equations.

We note the relations:%
\begin{align*}
\dot{\xi}^{c}u_{c}  &  =\xi_{c;d}u^{c}u^{d}=-\psi+H_{cd}u^{c}u^{d}\\
q^{d}\left(  u_{c}\xi_{;d}^{c}+\dot{\xi}_{d}\right)   &  =q^{d}(u^{c}\xi
_{c;d}+\xi_{d;c}u^{c})=2q^{d}u^{c}H_{cd}%
\end{align*}
Then equation (\ref{Key.38}) is written as follows:%
\begin{align*}
LHS  &  =\frac{1}{2}\left\{  \left(  \mu+3p\right)  ^{\circ}-4\mathring{q}%
_{c}u^{c}+2\mathring{\pi}_{cd}u^{c}u^{d}-2\left(  \mu+3p-2\Lambda\right)
\left(  \dot{\xi}_{c}u^{c}\right)  -4q_{c}\dot{\xi}^{c}\right\} \\
&  =\frac{1}{2}\left\{  \left(  \mu+3p\right)  ^{\circ}+4q_{c}\mathring{u}%
^{c}+2\mathring{\pi}_{cd}u^{c}u^{d}-2\left(  \mu+3p-2\Lambda\right)  \left(
-\psi+H_{cd}u^{c}u^{d}\right)  -4q_{c}\dot{\xi}^{c}\right\} \\
&  =\frac{1}{2}\left\{  \left(  \mu+3p\right)  ^{\circ}+4q_{c}(\mathring
{u}^{c}-\dot{\xi}^{c})+2\mathring{\pi}_{cd}u^{c}u^{d}-2\left(  \mu
+3p-2\Lambda\right)  \left(  -\psi+H_{cd}u^{c}u^{d}\right)  \right\} \\
&  =\frac{1}{2}\left\{  \left(  \mu+3p\right)  ^{\circ}+4q_{c}L_{\xi}%
u^{c}+2\mathring{\pi}_{cd}u^{c}u^{d}-2\left(  \mu+3p-2\Lambda\right)  \left(
-\psi+H_{cd}u^{c}u^{d}\right)  \right\}
\end{align*}
hence (\ref{Key.38}) becomes:%
\begin{equation}
\left(  \mu+3p\right)  ^{\circ}+4q_{c}L_{\xi}u^{c}+2\mathring{\pi}_{cd}%
u^{c}u^{d}-2\left(  \mu+3p-2\Lambda\right)  \left(  -\psi+H_{cd}u^{c}%
u^{d}\right)  =2(-3\mu_{\psi}+3p_{\psi}+\mu_{K}) \label{Key.38a}%
\end{equation}

Similarly, for expression (\ref{Key.40}) we have:%
\begin{align*}
&  LHS\\
&  =2\left[  \frac{3}{2}\left(  \mu-p\right)  ^{\circ}+2q_{c}\mathring{u}%
^{c}+\mathring{\pi}_{cd}h^{cd}+\left(  \mu-p+2\Lambda\right)  (4\psi+\dot{\xi
}^{c}u_{c})+2u_{c}\xi_{;d}^{c}q^{d}+2\pi_{cd}H^{cd}\right] \\
&  =2\left[  \frac{3}{2}\left(  \mu-p\right)  ^{\circ}+2q_{c}\mathring{u}%
^{c}+\mathring{\pi}_{cd}u^{c}u^{d}+\left(  \mu-p+2\Lambda\right)
(3\psi+H_{cd}u^{c}u^{d})+2u^{c}q^{d}(\psi g_{cd}+H_{cd}-\xi_{d;c})+2\pi
_{cd}H^{cd}\right] \\
&  =2\left[  \frac{3}{2}\left(  \mu-p\right)  ^{\circ}+2q_{c}(\mathring{u}%
^{c}-\dot{\xi}^{c})+\mathring{\pi}_{cd}u^{c}u^{d}+\left(  \mu-p+2\Lambda
\right)  (3\psi+H_{cd}u^{c}u^{d})+2u^{c}q^{d}H_{cd}+2\pi_{cd}H^{cd}\right] \\
&  =3\left(  \mu-p\right)  ^{\circ}+4q_{c}L_{\xi}u^{c}+2\mathring{\pi}%
_{cd}u^{c}u^{d}+2\left(  \mu-p+2\Lambda\right)  (3\psi+H_{cd}u^{c}%
u^{d})+42u^{c}q^{d}H_{cd}+4\pi_{cd}H^{cd},
\end{align*}
that is,
\begin{align}
&  3\left(  \mu-p\right)  ^{\circ}+4q_{c}L_{\xi}u^{c}+2\mathring{\pi}%
_{cd}u^{c}u^{d}+2\left(  \mu-p+2\Lambda\right)  (3\psi+H_{cd}u^{c}%
u^{d})+4H_{cd}u^{c}q^{d}+4\pi_{cd}H^{cd}\nonumber\\
&  =3(\mu_{\psi}-5p_{\psi}+p_{K}). \label{Key.40a}%
\end{align}
Adding the two new equations we get:%
\begin{align*}
LHS  &  =\left(  \mu+3p\right)  ^{\circ}+3\left(  \mu-p\right)  ^{\circ
}+8q_{c}L_{\xi}u^{c}+4\mathring{\pi}_{cd}u^{c}u^{d}\\
&  +2\left(  3\mu-3p+6\Lambda+\mu+3p-2\Lambda\right)  \psi\\
&  +2\left(  \mu-p+2\Lambda-\mu-3p+2\Lambda\right)  H_{cd}u^{c}u^{d}\\
&  +4H_{cd}u^{c}q^{d}+4\pi_{cd}H^{cd}\\
&  =4\mathring{\mu}+8q_{c}L_{\xi}u^{c}+4\mathring{\pi}_{cd}u^{c}u^{d}%
+8(\mu+\Lambda)\psi\\
&  +8(-p+\Lambda)H_{cd}u^{c}u^{d}+4H_{cd}u^{c}q^{d}+4\pi_{cd}H^{cd}%
\end{align*}%
\[
RHS=3(\mu_{\psi}-5p_{\psi}+p_{K})+2(-3\mu_{\psi}+3p_{\psi}+\mu_{K}%
)=-3\mu_{\psi}-9p_{\psi}+3p_{K}+2\mu_{K},
\]
that is,%
\begin{align*}
&  4\mathring{\mu}+8q_{c}L_{\xi}u^{c}+4\mathring{\pi}_{cd}u^{c}u^{d}%
+8(\mu+\Lambda)\psi\\
&  +8(-p+\Lambda)H_{cd}u^{c}u^{d}+4H_{cd}u^{c}q^{d}+4\pi_{cd}H^{cd}\\
&  =-3\mu_{\psi}-9p_{\psi}+3p_{K}+2\mu_{K}%
\end{align*}
or%
\begin{align}
\mathring{\mu}  &  =\frac{1}{4}\left(  -3\mu_{\psi}-9p_{\psi}+3p_{K}+2\mu
_{K}\right)  -2q_{c}L_{\xi}u^{c}-\mathring{\pi}_{cd}u^{c}u^{d}-2(\mu
+\Lambda)\psi\nonumber\\
&  +2(p-\Lambda)H_{cd}u^{c}u^{d}-H_{cd}u^{c}q^{d}-\pi_{cd}H^{cd}.
\label{Key.42}%
\end{align}

The last equation expresses the derivative of $\mu$ along the collineation
vector $\xi^{a}.$ To find $\mathring{p}$ we use (\ref{Key.38a}) and replace
$\mathring{\mu}.$ We have:%
\begin{align*}
&  3\mathring{p}+\frac{1}{4}\left(  -3\mu_{\psi}-9p_{\psi}+3p_{K}+2\mu
_{K}\right)  -2(\mu+\Lambda)\psi\\
&  +2(p-\Lambda)H_{cd}u^{c}u^{d}-H_{cd}u^{c}q^{d}-\pi_{cd}H^{cd}\\
&  +2q_{c}L_{\xi}u^{c}+\mathring{\pi}_{cd}u^{c}u^{d}-2\left(  \mu
+3p-2\Lambda\right)  \left(  -\psi+H_{cd}u^{c}u^{d}\right) \\
&  =2(-3\mu_{\psi}+3p_{\psi}+\mu_{K})
\end{align*}

\begin{align*}
&  3\mathring{p}-H_{cd}u^{c}q^{d}-\pi_{cd}H^{cd}\\
&  +2q_{c}L_{\xi}u^{c}+\mathring{\pi}_{cd}u^{c}u^{d}+\frac{1}{4}\left(
-3\mu_{\psi}-9p_{\psi}+3p_{K}+2\mu_{K}\right) \\
&  +2\left(  \mu+3p-2\Lambda-\mu-\Lambda\right)  \psi+2(p-\Lambda
-\mu-3p+2\Lambda)H_{cd}u^{c}u^{d}\\
&  =2(-3\mu_{\psi}+3p_{\psi}+\mu_{K})
\end{align*}%
\begin{align*}
&  3\mathring{p}-H_{cd}u^{c}q^{d}-\pi_{cd}H^{cd}\\
&  +2q_{c}L_{\xi}u^{c}+\mathring{\pi}_{cd}u^{c}u^{d}\\
&  +6(p-\Lambda)\psi-2(2p+\mu-\Lambda)H_{cd}u^{c}u^{d}\\
&  =2(-3\mu_{\psi}+3p_{\psi}+\mu_{K})-\frac{1}{4}\left(  -3\mu_{\psi}%
-9p_{\psi}+3p_{K}+2\mu_{K}\right)
\end{align*}%
\begin{align*}
&  3\mathring{p}-H_{cd}u^{c}q^{d}-\pi_{cd}H^{cd}+2q_{c}L_{\xi}u^{c}%
+\mathring{\pi}_{cd}u^{c}u^{d}\\
&  +6(p-\Lambda)\psi-2(2p+\mu-\Lambda)H_{cd}u^{c}u^{d}\\
&  =\frac{3}{4}\left(  7\mu_{\psi}+11p_{\psi}+2\mu_{K}-p_{K}\right)
\end{align*}%
\begin{align}
3\mathring{p}  &  =H_{cd}u^{c}q^{d}+\pi_{cd}H^{cd}-2q_{c}L_{\xi}%
u^{c}-\mathring{\pi}_{cd}u^{c}u^{d}\nonumber\\
&  -6(p-\Lambda)\psi+2(2p+\mu-\Lambda)H_{cd}u^{c}u^{d}\nonumber\\
&  +\frac{1}{4}\left(  -3\mu_{\psi}-9p_{\psi}+3p_{K}+2\mu_{K}\right)  .
\label{Key.43}%
\end{align}

We concentrate now on (\ref{Key.39}). We have
\begin{align*}
LHS  &  =-\left(  \mu+p\right)  \mathring{u}_{d}h_{e}^{d}-\mathring{q}%
_{d}h_{e}^{d}+\mathring{\pi}_{dc}u^{d}h_{e}^{c}-\frac{1}{2}\left(
\mu+3p-2\Lambda\right)  u_{c}\xi_{;d}^{c}h_{e}^{d}\\
&  -q_{c}\xi_{;d}^{c}h_{e}^{d}+\frac{1}{2}\left(  \mu-p+2\Lambda\right)
\dot{\xi}_{d}h_{e}^{d}+\left(  q_{d}h_{e}^{d}\right)  \left(  \dot{\xi}%
^{c}u_{c}\right)  -h_{e}^{d}\pi_{cd}\dot{\xi}^{c}%
\end{align*}
The term:%
\begin{align*}
&  -\frac{1}{2}\left(  \mu+3p-2\Lambda\right)  u_{c}\xi_{;d}^{c}h_{e}%
^{d}+\frac{1}{2}\left(  \mu-p+2\Lambda\right)  \dot{\xi}_{d}h_{e}^{d}\\
&  =-\frac{1}{2}\left(  \mu+3p-2\Lambda\right)  H_{cd}u^{c}h_{e}^{d}+\frac
{1}{2}\left(  \mu+3p-2\Lambda\right)  \dot{\xi}_{d}h_{e}^{d}+\frac{1}%
{2}\left(  \mu-p+2\Lambda\right)  \dot{\xi}_{d}h_{e}^{d}\\
&  =-\frac{1}{2}\left(  \mu+3p-2\Lambda\right)  H_{cd}u^{c}h_{e}^{d}%
+(\mu+p)\dot{\xi}_{d}h_{e}^{d}%
\end{align*}
hence:%
\begin{align*}
LHS  &  =-(\mu+p)h_{ed}L_{\xi}u^{d}-\left[  \mathring{q}_{d}h_{e}^{d}+q_{c}%
\xi_{;d}^{c}h_{e}^{d}\right] \\
&  +\mathring{\pi}_{dc}u^{d}h_{e}^{c}-h_{e}^{d}\pi_{cd}\dot{\xi}^{c}.
\end{align*}

The term:%
\begin{align*}
\mathring{q}_{d}h_{e}^{d}+q_{c}\xi_{;d}^{c}h_{e}^{d}  &  =\mathring{q}%
_{d}h_{e}^{d}+q_{c}(\psi\delta_{d}^{c}+H_{d}^{c}-\xi_{d}^{..;c})h_{e}^{d}\\
&  =(\mathring{q}_{d}-\xi_{d}^{..;c}q_{c})h_{e}^{d}+\psi q_{e}+H_{cd}%
q^{c}h_{e}^{d}\\
&  =h_{e}^{d}L_{\xi}q_{d}+\psi q_{e}+H_{cd}q^{c}h_{e}^{d}%
\end{align*}
Also the Lie derivative:%
\[
L_{\xi}\pi_{cd}=\pi_{cd;f}\xi^{f}+\pi_{fd}\xi_{;c}^{f}+\pi_{cf}\xi_{;d}%
^{f}=\mathring{\pi}_{cd}+\pi_{fd}\xi_{;c}^{f}+\pi_{cf}\xi_{;d}^{f}%
\]
hence:%
\[
u^{d}h_{e}^{c}L_{\xi}\pi_{cd}=\mathring{\pi}_{cd}u^{d}h_{e}^{c}+\pi_{cf}%
\dot{\xi}^{f}h_{e}^{c}%
\]

Therefore:%
\[
LHS=-(\mu+p)h_{ed}L_{\xi}u^{d}-h_{e}^{d}L_{\xi}q_{d}+u^{d}h_{e}^{c}L_{\xi}%
\pi_{cd}-\psi q_{e}-H_{cd}q^{c}h_{e}^{d}%
\]
and equation (\ref{Key.39}) becomes:%
\[
-(\mu+p)h_{ed}L_{\xi}u^{d}-h_{de}L_{\xi}q^{d}+u^{d}h_{e}^{c}L_{\xi}\pi
_{cd}-\psi q_{e}-H_{cd}q^{c}h_{e}^{d}=-(-2q_{\psi e}+q_{Ke})
\]
or:%
\begin{equation}
-h_{de}L_{\xi}q^{d}+u^{d}h_{e}^{c}L_{\xi}\pi_{cd}=(\mu+p)h_{ed}L_{\xi}%
u^{d}+\psi q_{e}+H_{cd}q^{c}h_{e}^{d}-(-2q_{\psi e}+q_{Ke}) \label{Key.39a}%
\end{equation}

We continue with equation (\ref{Key.41}).

We have:%
\begin{align*}
&  LHS =2\left[  h_{(a}^{e}h_{b)}^{f}-\frac{1}{3}h_{ab}h^{ef}\right] \\
&  \left[  \frac{1}{2}\left(  \mu-p\right)  ^{\circ}h_{ef}+2q_{(c}\mathring
{u}_{d)}h_{e}^{c}h_{f}^{d}+\mathring{\pi}_{cd}h_{e}^{c}h_{f}^{d}+\left(
\mu-p+2\Lambda\right)  (\psi g_{dr}+H_{dr})h_{e}^{d}h_{f}^{r}+2q_{r}u_{c}%
\xi_{;d}^{c}h_{e}^{d}h_{f}^{r}+2\pi_{cr}\xi_{;d}^{c}h_{e}^{d}h_{f}^{r}\right]
\end{align*}

The terms:%
\[
2\left[  h_{(a}^{e}h_{b)}^{f}-\frac{1}{3}h_{ab}h^{ef}\right]  \frac{1}%
{2}\left(  \mu-p\right)  ^{\circ}h_{ef}=0
\]%
\begin{align*}
&  \left[  h_{(a}^{e}h_{b)}^{f}-\frac{1}{3}h_{ab}h^{ef}\right]  \left(
\mu-p+2\Lambda\right)  \xi_{d;r}h_{e}^{d}h_{f}^{r}\\
&  =\left[  h_{(a}^{e}h_{b)}^{f}-\frac{1}{3}h_{ab}h^{ef}\right]  \left(
\mu-p+2\Lambda\right)  (\psi g_{dr}+H_{dr})h_{e}^{d}h_{f}^{r}\\
&  =\left(  \mu-p+2\Lambda\right)  \left[  h_{(a}^{e}h_{b)}^{f}-\frac{1}%
{3}h_{ab}h^{ef}\right]  H_{dr}h_{e}^{d}h_{f}^{r}\\
&  =\left(  \mu-p+2\Lambda\right)  \left[  h_{(a}^{e}h_{b)}^{f}-\frac{1}%
{3}h_{ab}h^{ef}\right]  H_{ef}%
\end{align*}
i.e. the traceless part of the projection of the traceless tensor $H_{ab}.$
The term:%
\begin{align*}
&  \left[  h_{(a}^{e}h_{b)}^{f}-\frac{1}{3}h_{ab}h^{ef}\right]  2q_{r}u^{c}%
\xi_{c;d}h_{e}^{d}h_{f}^{r}\\
&  =\left[  h_{(a}^{e}h_{b)}^{f}-\frac{1}{3}h_{ab}h^{ef}\right]  2q_{r}%
u^{c}(\psi g_{cd}+H_{cd})h_{e}^{d}h_{f}^{r}\\
&  =\left[  h_{(a}^{e}h_{b)}^{f}-\frac{1}{3}h_{ab}h^{ef}\right]  2q_{r}%
u^{c}H_{cd}h_{e}^{d}h_{f}^{r}\\
&  =\left[  h_{(a}^{e}h_{b)}^{f}-\frac{1}{3}h_{ab}h^{ef}\right]  2q_{f}%
u^{c}H_{cd}h_{e}^{d}\\
&  =2q_{(a}h_{b)}^{d}H_{cd}u^{c}-\frac{2}{3}\left(  H_{cd}u^{c}q^{d}\right)
h_{ab}.
\end{align*}
The term:%
\begin{align*}
&  \left[  h_{(a}^{e}h_{b)}^{f}-\frac{1}{3}h_{ab}h^{ef}\right]  2\pi_{cr}%
\xi_{;d}^{c}h_{e}^{d}h_{f}^{r}\\
&  =2\left[  h_{(a}^{e}h_{b)}^{f}-\frac{1}{3}h_{ab}h^{ef}\right]  \pi_{.r}%
^{c}(\psi g_{cd}+H_{cd})h_{e}^{d}h_{f}^{r}\\
&  =2\left[  h_{(a}^{e}h_{b)}^{f}-\frac{1}{3}h_{ab}h^{ef}\right]  (\psi
\pi_{ef}+\pi_{.f}^{c}H_{cd}h_{e}^{d})\\
&  =2\left[  h_{(a}^{e}h_{b)}^{f}-\frac{1}{3}h_{ab}h^{ef}\right]  \psi\pi
_{ef}+2\left[  h_{(a}^{e}h_{b)}^{f}-\frac{1}{3}h_{ab}h^{ef}\right]  \pi
_{.f}^{c}H_{cd}h_{e}^{d}\\
&  =2\left[  h_{(a}^{e}h_{b)}^{f}-\frac{1}{3}h_{ab}h^{ef}\right]  \pi_{.f}%
^{c}H_{cd}h_{e}^{d}\\
&  =2h_{(a}^{e}h_{b)}^{f}\pi_{.f}^{c}H_{cd}h_{e}^{d}-\frac{2}{3}h_{ab}%
(\pi^{cd}H_{cd})
\end{align*}
Therefore we have:%
\begin{align}
&  2q_{(c}\mathring{u}_{d)}h_{e}^{c}h_{f}^{d}+\mathring{\pi}_{cd}h_{e}%
^{c}h_{f}^{d}+\left(  \mu-p+2\Lambda\right)  \left[  h_{(a}^{e}h_{b)}%
^{f}-\frac{1}{3}h_{ab}h^{ef}\right]  H_{ef}\nonumber\\
&  +2q_{(a}h_{b)}^{d}H_{cd}u^{c}-\frac{2}{3}\left(  H_{cd}u^{c}q^{d}\right)
h_{ab}\nonumber\\
&  +2h_{(a}^{e}h_{b)}^{f}\pi_{.f}^{c}H_{cd}h_{e}^{d}-\frac{2}{3}h_{ab}%
(\pi^{cd}H_{cd})\nonumber\\
&  =\frac{1}{2}\left[  \pi_{\psi ab}+\pi_{Kab}\right]  \label{Key.41a}%
\end{align}

The gravitational field equations are the (\ref{Key.38a}),(\ref{Key.39a}%
),(\ref{Key.40a}),(\ref{Key.41a}). With these field equations we have
completed the scenario for the generic General Relativistic model and we pass
to the well known gravitational model Bianchi I\ model.

\section{Example: The Bianchi I\ model}

In the following we assume the following conventions:

Greek indices take the space values $1,2,3$ and Latin indices the space-time
values $0,1,2,3$. We define the sign of the curvature tensor from the identity
$A_{;bc}^{a}-A_{;cb}^{a}=R_{bcd}^{a}A^{d}$ or $A_{a;bc}-A_{a;cb}=R_{dabc}%
A^{d}.$ In terms of the connection coefficients $R_{bcd}^{a}=\Gamma_{db}%
^{a},_{c}-\Gamma_{cb}^{a},_{d}+\Gamma_{cf}^{a}\Gamma_{db}^{f}-\Gamma_{df}%
^{a}\Gamma_{cb}^{f}.$

\subsection{Geometric assumptions defining the model}

A diagonal Bianchi I space-time is a spatially homogeneous space-time which
admits an Abelian group of isometries $G_{3}$, acting on spacelike
hypersurfaces, generated by the spacelike KVs $\mathbf{\xi}_{1}=\partial
_{x},\mathbf{\xi}_{2}=\partial_{y},\mathbf{\xi}_{3}=\partial_{z}$ and a
timelike gradient KV the $u^{a}=\frac{\partial}{\partial t}$ which is normal
to the homogeneous 3-d hypersurfaces.

In synchronous coordinates $\{t,x,y,z\}$ the above assumptions imply that the
metric of this spacetime is:%

\begin{equation}
ds^{2}=-dt^{2}+A_{\mu}^{2}(t)(dx^{\mu})^{2} \label{Key.45}%
\end{equation}
where the metric functions $A_{1}(t),A_{2}(t),A_{3}(t)$ are functions of the
time coordinate only. When two of the functions $A_{\mu}(t)$ are equal (e.g.
$A_{2}=A_{3}$) the Bianchi I space-time reduces to the important class of
plane symmetric space-times (a special class of the Locally Rotational
Symmetric space-times \cite{Ellis1},\cite{Stewart-Ellis} which admit a $G_{4}$
group of isometries acting multiply transitively on the spacelike
hypersurfaces of homogeneity generated by the vectors $\mathbf{\xi}%
_{1},\mathbf{\xi}_{2},\mathbf{\xi}_{3}$ and $\mathbf{\xi}_{4}=x^{2}%
\partial_{3}-x^{3}\partial_{2})$.

For economy of writing in the following we write $A_{\mu}$ instead of $A_{\mu
}^{2}(t).$ Furthermore we shall be interested only in \emph{proper diagonal}
Bianchi I space-times (which in the following shall be referred simply as
Bianchi I\ space-times), hence all metric functions are assumed to be
different and the dimension of the group of isometries acting on the spacelike
hypersurfaces is three.

The implications of the additional geometric assumptions (i.e. the symmetries we assumed) are:

1. The computation of the Ricci tensor%

\begin{equation}
R_{tt}=-{\frac{\ddot{A}_{1}A_{2}A_{3}+\ddot{A}_{2}A_{1}A_{3}+\ddot{A}_{3}%
A_{1}A_{2}}{A_{1}A_{2}A_{3}}} \label{Key.46}%
\end{equation}

\begin{equation}
R_{xx}={\frac{A_{1}\left(  \ddot{A}_{1}A_{2}A_{3}+\dot{A}_{1}\dot{A}_{2}%
A_{3}+\dot{A}_{1}\dot{A}_{3}A_{2}\right)  }{A_{2}A_{3}}} \label{Key.47}%
\end{equation}

\begin{equation}
R_{yy}={\frac{A_{2}\left(  \ddot{A}_{2}A_{1}A_{3}+\dot{A}_{1}\dot{A}_{2}%
A_{3}+\dot{A}_{2}\dot{A}_{3}A_{1}\right)  }{A_{1}A_{3}}} \label{Key.48}%
\end{equation}

\begin{equation}
R_{zz}={\frac{A_{3}\left(  \ddot{A}_{3}A_{1}A_{2}+\dot{A}_{1}\dot{A}_{3}%
A_{2}+\dot{A}_{2}\dot{A}_{3}A_{1}\right)  }{A_{1}A_{2}}} \label{Key.49}%
\end{equation}

From the Ricci tensor we compute the Einstein tensor and then use Einstein
field equations to write the energy momentum tensor in terms of the metric
functions $A_{1}(t),A_{2}(t),A_{3}(t)$ and their derivatives. However that
does not mean that we are able to discuss anything about the physical
variables (energy density, isotropic pressure etc.) because in order to do
that we need to have observers. The expression of $T_{ab}$ we find for
$T_{ab}$ is the matter content of this spacetime, the same for \textbf{all}
observers in the Bianchi I\ background.

2. The computation of the Weyl tensor

The Weyl tensor defined in (\ref{CPE.22b}) is important because it is involved
in the second Bianchi identity. More specifically from this tensor one
computes the electric and the magnetic parts $E^{ab},H^{ab}$ which enter into
the 1+3 decomposition of the second Bianchi identity given by equations
(\ref{BI.05}) - (\ref{BI.08}). We compute:%

\begin{equation}
E_{xx}=\ {\frac{A_{1}\left(  -2\,\ddot{A}_{1}A_{2}A_{3}+\dot{A}_{1}\dot{A}%
_{2}A_{3}+\dot{A}_{1}\dot{A}_{3}A_{2}+\ddot{A}_{2}A_{1}A_{3}+\ddot{A}_{3}%
A_{1}A_{2}-2\,\dot{A}_{2}\dot{A}_{3}A_{1}\right)  }{6A_{2}A_{3}}}
\label{Key.49a}%
\end{equation}

\begin{equation}
E_{yy}={\frac{A_{2}\left(  -2\,\ddot{A}_{2}A_{1}A_{3}+\dot{A}_{1}\dot{A}%
_{2}A_{3}+\dot{A}_{2}\dot{A}_{3}A_{1}+\ddot{A}_{1}A_{2}A_{3}+\ddot{A}_{3}%
A_{1}A_{2}-2\,\dot{A}_{1}\dot{A}_{3}A_{2}\right)  }{6A_{1}A_{3}}}
\label{Key.49b}%
\end{equation}

\begin{equation}
E_{zz}=-\,{\frac{A_{3}\left(  2\,\ddot{A}_{3}A_{1}A_{2}-\dot{A}_{1}\dot{A}%
_{3}A_{2}-\dot{A}_{2}\dot{A}_{3}A_{1}-\ddot{A}_{1}A_{2}A_{3}-\ddot{A}_{2}%
A_{1}A_{3}+2\,\dot{A}_{1}\dot{A}_{2}A_{3}\right)  }{6A_{1}A_{2}}}
\label{Key.49c}%
\end{equation}
and the magnetic part \thinspace$H_{ab}=0.$

\textbf{Observers.}

The choice of observers is open and independent of the choice of the assumed
symmetries (i.e. the model) spacetime. However as it has been noted the
kinematic quantities they define must satisfy the propagation, the constraint
equations and the Bianchi second identity.

We choose the observers (this is one choice, any other would do provided it
satisfies the aforementioned identities) to be the ones defined by the time
coordinate $t$ i.e. we take $u^{a}=\delta_{0}^{a}$ in the synchronous
coordinate system.

The implications of this choice of observers are:

1. Kinematics

We 1+3 decompose $u_{a;b}$ and find the kinematic variables:%
\begin{align}
\theta &  =\left[  \ln(A_{1}A_{2}A_{3})\right]  ^{\cdot}\label{Key.50}\\
\omega_{ab}  &  =0,\dot{u}^{a}=0\label{Key.51}\\
\sigma_{ab}  &  =\frac{1}{3}diag\left(  0,A_{1}^{2}\left[  \ln\left(
\frac{A_{1}^{2}}{A_{2}A_{3}}\right)  \right]  ^{\cdot},A_{2}^{2}\left[
\ln\left(  \frac{A_{2}^{2}}{A_{3}A_{1}}\right)  \right]  ^{\cdot},A_{3}%
^{2}\left[  \ln\left(  \frac{A_{3}^{2}}{A_{1}A_{2}}\right)  \right]  ^{\cdot
}\right)  \label{Key.52}%
\end{align}

a. The propagation equations give\footnote{Note that $\sigma^{2}=\frac{1}%
{2}\sigma_{ab}\sigma^{ab}.$ For the case we are considering we calculate:%
\[
\sigma^{2}=\frac{1}{3}\left[
{\displaystyle\sum\limits_{I=1}^{3}}
\left(  \frac{\dot{A}_{I}}{A_{I}}\right)  ^{2}-%
{\displaystyle\sum\limits_{I\neq J=1}^{3}}
\frac{\dot{A}_{I}}{A_{I}}\frac{\dot{A}_{J}}{A_{J}}\right]
\]
}:%
\begin{equation}
\dot{\theta}+\frac{1}{3}\theta^{2}+2\sigma^{2}=\frac{\ddot{A}_{1}A_{2}%
A_{3}+\ddot{A}_{2}A_{1}A_{3}+\ddot{A}_{3}A_{1}A_{2}}{A_{1}A_{2}A_{3}}
\label{Key.54}%
\end{equation}%
\begin{equation}
-E_{st}=h_{s}^{a}h_{t}^{b}\dot{\sigma}_{ab}+\sigma_{sc}\sigma_{\text{
\thinspace}t}^{c}+\frac{2}{3}\sigma_{st}\theta-\frac{2}{3}2\sigma^{2}h_{st}
\label{Key.55}%
\end{equation}

b. The constraint equations give:%
\begin{equation}
\frac{2}{3}h_{s}^{c}\theta,_{c}=h_{s}^{c}h^{ab}\sigma_{ac;b}\text{(three
equations)} \label{Key.56}%
\end{equation}%
\begin{equation}
0=-h_{(s}^{a}h_{t)}^{b}\sigma_{b}^{\text{ \ }c;d}\eta_{arcd}u^{r}\text{(five
equations)} \label{Key.57}%
\end{equation}

Equations (\ref{Key.54}) - (\ref{Key.57}) must be satisfied identically by the
kinematic quantities. It is easy to show that this is true for equations
(\ref{Key.54}) and (\ref{Key.55}). Equation (\ref{Key.56}) is trivially
satisfied because $\theta,\sigma_{ab}$ are functions of $t$ only. Equation
(\ref{Key.57}) is also trivially satisfied because it contains derivatives of
the components of $\sigma_{ab}$ along the space coordinates only (due to the
term $\eta_{arcd}u^{r})$. We conclude that the propagation and the constraint
equations do not give any new conditions on the metric functions
$A_{1}(t),A_{2}(t),A_{3}(t).$

c. The Bianchi identities give the derivatives of $E_{st},H_{st}$ therefore
they do not add new constraints on the metric functions. They are only
compatibility conditions.

d. The propagation of the kinematic quantities along the symmetry vectors.

The fact that $\xi_{\mu}^{a}$ ($\mu=1,2,3)$ are Killing vectors (hence
$\psi=0,H_{ab}=0)$ provides:%
\begin{equation}
V^{a}(u)=\hat{V}_{a}(u)=0 \label{Key.58}%
\end{equation}
From (\ref{LDHX.9d}) and (\ref{LDHX.12}) we have taking into account the above
results:%
\[
L_{\xi}\theta=0\Rightarrow\theta=\theta(t)
\]

\begin{equation}
L_{\xi}\sigma_{ab}=0\Rightarrow\sigma_{ab}=\sigma_{ab}(t).\nonumber
\end{equation}

These equations give nothing new because we have already computed
$\theta,\sigma_{ab}$ and have found that they are functions of $t$ only.

\textbf{Dynamics}

a. We compute the physical parameters for the chosen observers. We find:%

\begin{equation}
\mu={\frac{\dot{A}_{1}\dot{A}_{2}A_{3}+\dot{A}_{1}\dot{A}_{3}A_{2}+\dot{A}%
_{2}\dot{A}_{3}A_{1}}{A_{1}A_{2}A_{3}}=\frac{\dot{A}_{1}\dot{A}_{2}}%
{A_{1}A_{2}}+}\frac{\dot{A}_{1}\dot{A}_{3}}{A_{1}A_{3}}+\frac{\dot{A}_{2}%
\dot{A}_{3}}{A_{2}A_{3}} \label{Key.59}%
\end{equation}

\begin{equation}
p=-{\frac{2\,\ddot{A}_{2}A_{1}A_{3}+2\,\ddot{A}_{3}A_{1}A_{2}+\dot{A}_{2}%
\dot{A}_{3}A_{1}+2\,\ddot{A}_{1}A_{2}A_{3}+\dot{A}_{1}\dot{A}_{3}A_{2}+\dot
{A}_{1}\dot{A}_{2}A_{3}}{3A_{1}A_{2}A_{3}}} \label{Key.60}%
\end{equation}

\begin{equation}
\pi_{xx}=-{\frac{A_{1}\left(  -2\,\ddot{A}_{1}A_{2}A_{3}+\ddot{A}_{2}%
A_{1}A_{3}+\ddot{A}_{3}A_{1}A_{2}+2\,\dot{A}_{2}\dot{A}_{3}A_{1}-\dot{A}%
_{3}\dot{A}_{1}A_{2}-\dot{A}_{1}\dot{A}_{2}A_{3}\right)  }{3A_{2}A_{3}}}
\label{Key.61}%
\end{equation}

\begin{equation}
\pi_{yy}=-{\frac{A_{2}\left(  -2\,\ddot{A}_{2}A_{1}A_{3}+\ddot{A}_{3}%
A_{1}A_{2}+\ddot{A}_{1}A_{2}A_{3}+2\,\dot{A}_{1}\dot{A}_{3}A_{2}-\dot{A}%
_{2}\dot{A}_{3}A_{1}-\dot{A}_{1}\dot{A}_{2}A_{3}\right)  }{3A_{1}A_{3}}}
\label{Key.62}%
\end{equation}

\begin{equation}
\pi_{zz}=-{\frac{A_{3}\left(  -2\,\ddot{A}_{3}A_{1}A_{2}+\ddot{A}_{2}%
A_{1}A_{3}+\ddot{A}_{1}A_{2}A_{3}+2\,\dot{A}_{1}\dot{A}_{2}A_{3}-\dot{A}%
_{2}\dot{A}_{3}A_{1}-\dot{A}_{1}\dot{A}_{3}A_{2}\right)  }{3A_{1}A_{2}}}
\label{Key.63}%
\end{equation}

The momentum transfer vector $q^{a}=0$\footnote{This is expected from the
symmetries of the metric and the non-degeneracy of the $T_{ab.}$ Note that
both $R_{ab}$ and $T_{ab}$ are of the same form as the metric. This is to be
expected because they can be considered as metrics (a metric is a symmetric
tensor of type (0,2)\ and nothing more or less) and they admit the same KVs
with the metric.}.

We note that the second equation can be written:%
\begin{equation}
\frac{\ddot{A}_{1}}{A_{1}}+\frac{\ddot{A}_{2}}{A_{2}}+\frac{\ddot{A}_{3}%
}{A_{3}}=-\frac{1}{2}(\mu+3p). \label{Key.63a}%
\end{equation}
We introduce the notation:%
\begin{align*}
I_{2}  &  =\frac{\ddot{A}_{1}}{A_{1}}+\frac{\ddot{A}_{2}}{A_{2}}+\frac
{\ddot{A}_{3}}{A_{3}}=-\frac{1}{2}(\mu+3p)\\
I_{1}  &  ={\frac{\dot{A}_{1}\dot{A}_{2}}{A_{1}A_{2}}+}\frac{\dot{A}_{1}%
\dot{A}_{3}}{A_{1}A_{3}}+\frac{\dot{A}_{2}\dot{A}_{3}}{A_{2}A_{3}}=\mu.
\end{align*}
and we have:%
\begin{align}
\mu &  =I_{1}\label{Key.85}\\
p  &  =-\frac{1}{3}\left(  2I_{2}+I_{1}\right) \label{Key.86}\\
\pi_{xx}  &  =\frac{A_{1}^{2}}{3}\left(  3\frac{\ddot{A}_{1}}{A_{1}}%
-3\frac{\dot{A}_{2}\dot{A}_{3}}{A_{2}A_{3}}-I_{2}+I_{1}\right)  \label{Key.87}%
\\
\pi_{yy}  &  =\frac{A_{2}^{2}}{3}\left(  3\frac{\ddot{A}_{2}}{A_{2}}%
-3\frac{\dot{A}_{1}\dot{A}_{3}}{A_{1}A_{3}}-I_{2}+I_{1}\right)  \label{Key.88}%
\\
\pi_{zz}  &  =\frac{A_{3}^{2}}{3}\left(  3\frac{\ddot{A}_{3}}{A_{3}}%
-3\frac{\dot{A}_{1}\dot{A}_{2}}{A_{1}A_{2}}-I_{2}+I_{1}\right)  \label{Key.89}%
\end{align}
b. Conservation equations.

From (\ref{TE.6}) and (\ref{TE.7}) we have for Bianchi I\ spacetime and the
observers we selected the conservation equations:
\begin{equation}
\dot{\mu}+(\mu+p)\theta+\pi^{ab}\sigma_{ab}=0 \label{Key.64}%
\end{equation}%
\begin{equation}
h_{a}^{c}(p_{,c}+\pi_{\text{ }c\text{ };d}^{d})=0. \label{Key.65}%
\end{equation}
Equation (\ref{Key.65}) is trivially satisfied because all quantities are
functions of $t.$ Equation (\ref{Key.64}) is also trivially satisfied if we
replace the expressions of $\mu,p,\theta,\pi_{ab}$ from the corresponding expressions.

We see that there are no field equations to solve! Indeed we have solved them
in terms of three arbitrary functions which are the metric functions
$A_{1}(t),A_{2}(t),A_{3}(t)!$ Therefore for the observers $u^{a}=\delta
_{0}^{a}$ we have solved the problem completely\footnote{For another class of
observers we could have more constraint equations which would have to be
solved. Obviously in this case all dynamic and kinematic variables will (in
general)\ be different.}.

We are free to select special solutions from the three parameter family of
solutions we have found by imposing extra additional requirements. In the
following section we make \emph{one such} requirement and consider those
Bianchi I spacetime which \underline{\emph{for the observers we have chosen}}
give rise to a special type of mater which we call string fluid. Needless to
say that one could consider other requirements and select other types of
matter for the \emph{same spacetime} and the \emph{same observers}. Every
specification / condition on the metric functions $A_{1}(t),A_{2}(t),A_{3}(t)$
will produce a model Bianchi I\ spacetime (physical or not).

\section{The string fluid}

The connection between strings and vortices is well known \cite{Ray (1978)},
\cite{Lund F Regge T (1976)}, \cite{Letelier 1980}, \cite{Letelier 1981}. In
particular a geometric or Nambu string is a two-dimensional timelike surface
in spacetime. Letelier \cite{Letelier 1983} has considered a fluid represented
by a combination of geometric strings with particles attached to them so that
both have the same four velocity. He called such a fluid a string fluid and he
studied the gravitational field it produces in given spacetime backgrounds. In
a series of papers various authors \cite{Yavuz Yilmaz (1997)}, \cite{Yilmaz
2001}, \cite{Baysal Camci et all 2002}, \cite{Baysal Yilmaz 2002}, \cite{U
Camci 2002}, \cite{Sharif M Sheikh U 2005} \cite{M. Tsamparlis 2006} have
considered various types of collineations for a string fluid and derived the
conditions which must be satisfied in order the string fluid to admit a given collineation.

In this work we consider the string fluid with particles attached to the
strings which for some observers $u^{a}$ is described by the energy momentum
tensor\footnote{This expression is found form equation (2.28a) of
\cite{Letelier 1980} if we set $\rho=-\sigma$and $\pi\rightarrow-\pi.$}
\cite{Letelier 1980},\cite{Letelier 1981}:
\begin{equation}
T_{ab}=\rho(u_{a}u_{b}-n_{a}n_{b})+qp_{ab} \label{Key.66}%
\end{equation}
where:

- $\rho=\rho_{p}+\rho_{s}$ is the sum of the mass density of the strings
$(\rho_{s})$ and the mass density of the particles $(\rho_{p}),$

- $u^{a}$ is the common four velocity $(u^{a}u_{a}=-1)$ of the string and the
attached particle

- $n^{a}$ is a unit spacelike vector ($n^{a}n_{a}=1)$ normal to $u^{a}$
$(u^{a}n_{a}=0),$ which specifies the direction of \ the string (and the
direction of anisotropy of the string fluid)

- $q$ is a parameter contributing to the dynamic and kinematic properties of
the string

- $p_{ab}=h_{ab}-n_{a}n_{b}$ is the screen projection operator defined by the
vectors \ $u^{a},n^{a}$

By rewriting the energy momentum tensor as:
\begin{equation}
T_{ab}=\rho u_{a}u_{b}+\frac{1}{3}(2q-\rho)h_{ab}+(q+\rho)(\frac{1}{3}%
h_{ab}-n_{a}n_{b}) \label{Key.67}%
\end{equation}
(or otherwise) we compute its 1+3 decomposition. It follows that for a string
fluid:%
\begin{equation}
\mu=\rho,\;\;p=\frac{1}{3}(2q-\rho),\;\;q^{a}=0,\;\;\pi_{ab}=(q+\rho)(\frac
{1}{3}h_{ab}-n_{a}n_{b}). \label{Key.68}%
\end{equation}

We conclude that a string fluid is an anisotropic fluid with vanishing heat
flux. Furthermore we note that $n^{a}$ is an eigenvector of the anisotropic
stress tensor $\pi_{ab}$ with eigenvalue $-\frac{2}{3}(q+\rho).$ We assume
$q+\rho\neq0$ otherwise the string fluid reduces to a perfect fluid with
energy momentum tensor $T_{ab}=qg_{ab}$. This fluid has the unphysical
equation of state $\mu+p=0$.

From the above we note that the structure of the energy momentum tensor is
compatible with the general expressions (\ref{Key.59}) - (\ref{Key.63})
therefore the model of a string fluid we considered it is possible to be
described in a Bianchi I spacetime background with the comoving observers
$u^{a}=\delta_{0}^{a}$ we considered in the previous section. This is
equivalent to say that the additional symmetry assumptions we did is compatible
in the Bianchi I\ model spacetime  with the form of the energy momentum
tensor (\ref{Key.66}). This does not mean that there are not different sets of
additional assumptions which are compatible with matter of the form
(\ref{Key.66}). Bianchi I\ is just one.

To find the specific metric functions $A_{1}(t),A_{2}(t),A_{3}(t)$ which
correspond (or select) the string fluid model we equate the dynamical
parameters of (\ref{Key.59}) - (\ref{Key.63}) with those of (\ref{Key.68}). We
assume the string direction to be along the $x-axis$ i.e. we take ($A_{1}%
\neq0)$:%
\begin{equation}
n^{a}=\frac{1}{\sqrt{A_{1}}}\delta_{1}^{a}. \label{Key.69}%
\end{equation}
From the mass density and the pressure we find:%
\begin{equation}
\mu+3p=2q\Rightarrow I_{2}=-q \label{Key.71a}%
\end{equation}%
\[
\pi_{ab}=(q+\rho)(\frac{1}{3}h_{ab}-A_{1}^{2}\delta_{a}^{1}\delta_{b}%
^{1})=(\rho-I_{2})(\frac{1}{3}h_{ab}-A_{1}^{2}\delta_{a}^{1}\delta_{b}^{1}).
\]
In the coordinates we use:%
\begin{equation}
h_{ab}=diag(0,A_{1}^{2},A_{2}^{2}.A_{3}^{2}) \label{Key.70}%
\end{equation}
replacing in the expression of $\pi_{ab}$ we find:%
\begin{equation}
\pi_{ab}=\frac{1}{3}(\rho-I_{2})diag(0,-2A_{1}^{2},A_{2}^{2},A_{3}^{2}).
\label{Key.80}%
\end{equation}
Equating the two expressions of $\pi_{ab}$ (\ref{Key.87})- (\ref{Key.89}) and
(\ref{Key.80}) we find the field equations:%

\begin{align*}
\frac{\ddot{A}_{1}}{A_{1}}-\frac{\dot{A}_{2}\dot{A}_{3}}{A_{2}A_{3}}  &
=-\rho+I_{2}=-(\rho+q)\\
\frac{\ddot{A}_{2}}{A_{2}}-\frac{\dot{A}_{1}\dot{A}_{3}}{A_{1}A_{3}}  &  =0\\
\frac{\ddot{A}_{3}}{A_{3}}-\frac{\dot{A}_{1}\dot{A}_{2}}{A_{1}A_{2}}  &  =0
\end{align*}

The last three equations are dependent (one follows form the other two).
Eventually we have the following system of four simultaneous equations for the
five unknowns $A_{1}(t),A_{2}(t),A_{3}(t)$,$q,\rho:$%

\begin{equation}
{\frac{\dot{A}_{1}\dot{A}_{2}}{A_{1}A_{2}}+}\frac{\dot{A}_{1}\dot{A}_{3}%
}{A_{1}A_{3}}+\frac{\dot{A}_{2}\dot{A}_{3}}{A_{2}A_{3}}=\rho\label{Key.72}%
\end{equation}%
\begin{equation}
\frac{\ddot{A}_{1}}{A_{1}}+\frac{\ddot{A}_{2}}{A_{2}}+\frac{\ddot{A}_{3}%
}{A_{3}}=-q \label{Key.73}%
\end{equation}%
\begin{equation}
\frac{\ddot{A}_{2}}{A_{2}}-\frac{\dot{A}_{1}\dot{A}_{3}}{A_{1}A_{3}}=0
\label{Key.77}%
\end{equation}%
\begin{equation}
\frac{\ddot{A}_{3}}{A_{3}}-\frac{\dot{A}_{1}\dot{A}_{2}}{A_{1}A_{2}}=0.
\label{Key.78}%
\end{equation}
We have still the freedom to specify one more condition. This
condition could be an equation of state (in the board sense).

Note that the kinematic variables are not effected. Therefore when we take the
equation of state and we determine the functions $A_{1}(t),A_{2}(t),A_{3}(t)$
then we can compute the kinematic variables $\theta,\sigma_{ab}$ and draw
conclusions on the kinematics of the string fluid.

\subsection{An alternative energy momentum tensor}

Instead of the energy momentum tensor (\ref{Key.66}) Lattelier \cite{Letelier
1983} considered the energy momentum tensor:%
\begin{equation}
T_{ab}=\rho u_{a}u_{b}-\lambda n_{a}n_{b}. \label{Key.79}%
\end{equation}

This tensor reduces to the one we have considered if we set $q=0$ and
$\rho=\lambda.$ In general it is different to the one we have discussed above.
However our analysis applies the same. Considering the same observers and the
same direction $n^{a}$ we have the following 1+3 decomposition of $T_{ab}:$%
\begin{equation}
T_{ab}=\rho u_{a}u_{b}-\frac{1}{3}\lambda h_{ab}+\lambda(\frac{1}{3}%
h_{ab}-n_{a}n_{b}) \label{Key.90}%
\end{equation}
from which follows:%
\begin{align}
\mu &  =\rho,p=-\frac{1}{3}\lambda,q_{a}=0\label{Key.91}\\
\pi_{ab}  &  =\lambda(\frac{1}{3}h_{ab}-n_{a}n_{b})=\frac{\lambda}%
{3}(0,-2A_{1}^{2},A_{2}^{2},A_{3}^{2}). \label{Key.92}%
\end{align}
The:%
\[
I_{2}=-\frac{1}{2}(\mu+3p)=-\frac{1}{2}(\rho-\lambda)
\]
Equating the two expressions of $\pi_{ab}$ (\ref{Key.87})- (\ref{Key.89}) and
(\ref{Key.92}) we find the field equations:%

\begin{align*}
\frac{\ddot{A}_{1}}{A_{1}}-\frac{\dot{A}_{2}\dot{A}_{3}}{A_{2}A_{3}}  &
=-\lambda+I_{2}=-\frac{1}{6}(5\lambda+3\rho)\\
\frac{\ddot{A}_{2}}{A_{2}}-\frac{\dot{A}_{1}\dot{A}_{3}}{A_{1}A_{3}}  &
=I_{2}=\frac{1}{2}(\frac{\lambda}{3}-\rho)\\
\frac{\ddot{A}_{3}}{A_{3}}-\frac{\dot{A}_{1}\dot{A}_{2}}{A_{1}A_{2}}  &
=\frac{\ddot{A}_{2}}{A_{2}}-\frac{\dot{A}_{1}\dot{A}_{3}}{A_{1}A_{3}}.
\end{align*}

The last three equations are dependent (one follows form the other two).
Eventually we have again the following system of four simultaneous equations
for the five unknowns $A_{1}(t),A_{2}(t),A_{3}(t)$,$\lambda(t),\rho(t):$%

\begin{equation}
{\frac{\dot{A}_{1}\dot{A}_{2}}{A_{1}A_{2}}+}\frac{\dot{A}_{1}\dot{A}_{3}%
}{A_{1}A_{3}}+\frac{\dot{A}_{2}\dot{A}_{3}}{A_{2}A_{3}}=\rho\label{Key.93}%
\end{equation}%
\begin{equation}
\frac{\ddot{A}_{1}}{A_{1}}+\frac{\ddot{A}_{2}}{A_{2}}+\frac{\ddot{A}_{3}%
}{A_{3}}=\frac{1}{2}(\lambda-\rho) \label{Key.94}%
\end{equation}%
\begin{equation}
\frac{\ddot{A}_{2}}{A_{2}}-\frac{\dot{A}_{1}\dot{A}_{3}}{A_{1}A_{3}}=\frac
{1}{2}(\lambda-\rho) \label{Key.95}%
\end{equation}%
\begin{equation}
\frac{\ddot{A}_{3}}{A_{3}}-\frac{\dot{A}_{1}\dot{A}_{2}}{A_{1}A_{2}}=\frac
{1}{2}(\lambda-\rho) \label{Key.96}%
\end{equation}
We note that in this choice one has again the freedom to consider an extra
condition / equation of state. The solution of the field equations will be the
starting point of making Physics with this model.

\section{The 1+1+2 decomposition wrt a double congruence}

It is possible in a gravitational model one has in addition to the four velocity $u^{a}$
an additional  non-null vector field $n^{a}$ which is not parallel to $u^{a},$ as it is the
case for example with the string fluid. The existence of two characteristic
vector fields in spacetime introduces the concept of the double congruence
which leads to the finer 1+1+2 decomposition of tensor fields in spacetime.
This decomposition introduces new "kinematic" variables for the vector field
$n^{a}$ and new dynamical fields from the 1+1+2 decomposition of the energy
momentum tensor $T_{ab}.$

Consider two unit vector fields $u^{i}$ and $n^{i}$, with signatures
$u^{i}u_{i}=\varepsilon\left(  u\right)  $ , $n^{i}n_{i}=\varepsilon\left(
n\right)  $ so that $u^{i}n_{i}=\phi$ where $\phi\neq\pm1$\ (i.e. $u^{i}$ and
$n^{i}$ \ are not parallel) and define the \textbf{screen projection tensor}
$P_{ij}\left(  u,n\right)  $ by the formula:
\begin{equation}
P_{ij}\left(  u,n\right)  =g_{ij}+\frac{1}{\Delta}[\varepsilon(n)u_{i}%
u_{j}+\varepsilon\left(  u\right)  n_{i}n_{j}-\phi(u_{i}n_{j}+n_{i}u_{j})],
\label{d.two.1}%
\end{equation}
where $\Delta=\phi^{2}-\varepsilon\left(  u\right)  \varepsilon\left(
n\right)  .$

It is easy to show the properties:
\begin{align*}
P_{ij}  &  =P_{ji},P_{ij}P_{k}^{j}=P_{ik\,},,\,P_{i}^{i}=2\text{ \ }\\
&  \text{(}P_{ij}\text{ is a symmetric metric in the 2-space spanned by
(}u^{i},n^{i}))\\
P_{ij}u^{i}  &  =P_{ij}n^{i}=0\,\ \\
&  \text{(}P_{ij}\text{ projects normal to both (}u^{i},n^{i})\text{ that is
the \textquotedblleft screen space\textquotedblright%
\ spanned\footnote{Necessary an sufficient condition for the vector fields
$(u^{i},n^{i})$ to span a 2-dimensional space is that $L_{u}n^{i}%
=Au^{i}+Bn^{i}$ \ where $A,B$ \ are scalars. } by (}u^{i},n^{i}).
\end{align*}

When the vector fields $u^{i},n^{i}$ are timelike ($\varepsilon\left(
u\right)  =-1)$ and spacelike ($\varepsilon\left(  n\right)  =+1)$
respectively, the above formula becomes:%
\begin{equation}
P_{ij}\left(  u,n\right)  =g_{ij}-\frac{1}{1+\phi^{2}}[-u_{i}u_{j}+n_{i}%
n_{j}+\phi(u_{i}n_{j}+n_{i}u_{j})].\label{d.two.1a}%
\end{equation}
Furthermore, when $u^{i},n^{i}$ are perpendicular, then $\phi=0,$ \ and the
formula for $P_{ij}\left(  u,n\right)  $ reduces further to%
\begin{equation}
P_{ij}\left(  u,n\right)  =g_{ij}+u_{i}u_{j}-n_{i}n_{j}=h_{ij}-n_{i}%
n_{j}\label{d.two.1aa}%
\end{equation}
where $h_{ij}=g_{ij}+u_{i}u_{j}~$is the projection tensor for $u^{i}.$ The two
unit vector fields $\{u^{i},n^{i}\}$ define a double congruence which with the
use of the projection tensor $P_{ij}\left(  u,n\right)  $ define the 1+1+2
decomposition of geometric objects in spacetime. We start with the following

\begin{proposition}
A vector field (null or not) $R^{i}$ is 1+1+2 decomposed wrt the double
congruence $\{u^{i},n^{i}\}$ by means of the following identity:%
\begin{equation}
R^{i}=-\frac{1}{\Delta}\left[  (\varepsilon\left(  n\right)  R_{u}-\phi
R_{n})u^{i}+(\varepsilon\left(  u\right)  R_{n}-\phi R_{u})n^{i}\right]
+P_{j}^{i}R^{j}\label{d.two.3}%
\end{equation}
where $R_{u}=R^{i}u_{i}$ and $R_{n}=R^{i}n_{i}.$

\begin{proof}
Assume that:%
\[
R^{i}=Au^{i}+Bn^{i}+P_{j}^{i}R^{j}.%
\]
Contract with $u^{i},n^{i}$ \ to get the system of equations:%
\begin{align*}
\varepsilon\left(  u\right)  A+\phi B  &  =R_{u}\\
\phi A+\varepsilon\left(  n\right)  B  &  =R_{n}%
\end{align*}
where $R_{u}=R^{i}u_{i}$ and $R_{n}=R^{i}n_{i}.$ The determinant of the system
is \ $\varepsilon\left(  u\right)  \varepsilon\left(  n\right)  -\phi
^{2}=-\Delta.$ The solution is:%
\begin{align*}
A  &  =-\frac{1}{\Delta}(\varepsilon\left(  n\right)  R_{u}-\phi R_{n})\\
B  &  =-\frac{1}{\Delta}(\varepsilon\left(  u\right)  R_{n}-\phi R_{u}).
\end{align*}

\end{proof}
\end{proposition}

\begin{proposition}
The second rank tensor field $Y_{ij}$ is 1+1+2 decomposed wrt the double
congruence $\{u^{i},n^{i}\}$ by means of the identity:%
\begin{equation}
Y_{ij}=\frac{1}{\Delta}\left[  \alpha u_{i}u_{j}+\beta n_{i}n_{j}+\gamma
u_{i}n_{j}+\delta n_{i}u_{j}-\theta_{k}P_{i}^{k}u_{j}+\kappa_{k}P_{i}^{k}%
n_{j}-\rho_{s}P_{j}^{s}u_{i}+\nu_{k}n_{i}P_{j}^{k}\right]  +P_{i}^{k}P_{j}%
^{r}Y_{kr}\label{d.two.4}%
\end{equation}
where
\begin{align*}
\alpha &  =\frac{1}{\Delta}\left[  Y_{uu}-\phi\varepsilon(n)Y_{un}%
-\phi\varepsilon(n)Y_{nu}+\phi^{2}Y_{nn}\right]  \\
\beta &  =\frac{1}{\Delta}\left[  Y_{nn}-\phi\varepsilon(u)Y_{nu}%
-\phi\varepsilon(u)Y_{un}+\phi^{2}Y_{uu}\right]  \\
\gamma &  =\frac{1}{\Delta}\left[  \left(  \varepsilon(n)\varepsilon\left(
u\right)  Y_{un}-\phi\varepsilon(n)\right)  Y_{uu}-\phi\varepsilon
(u)Y_{nn}+\phi^{2}Y_{nu}\right]  \\
\delta &  =\frac{1}{\Delta}\left[  \varepsilon(n)\varepsilon\left(  u\right)
Y_{nu}-\phi\varepsilon(u)Y_{nn}-\phi\varepsilon(n)Y_{uu}+\phi^{2}%
Y_{un}\right]  \\
\theta_{k} &  =[\varepsilon\left(  n\right)  Y_{kr}u^{r}-\phi n^{s}Y_{rs}]\\
\kappa_{r} &  =[-\varepsilon\left(  u\right)  Y_{rs}n^{s}+\phi u^{s}Y_{rs}]\\
\rho_{s} &  =\left[  \varepsilon(n)Y_{rs}u^{r}-\phi Y_{rs}n^{r}\right]  \\
\nu_{r} &  =\left[  -\varepsilon(u)Y_{rs}n^{r}+\phi Y_{rs}u^{r}\right]  +\\
\Delta &  =\phi^{2}-\varepsilon\left(  s\right)  \varepsilon\left(  n\right)
\end{align*}

\end{proposition}

\begin{corollary}
Let \thinspace$u^{i}$ be a unit time-like vector $(\varepsilon(u)=-1)$
and $n^{i}$ a unit space-like vector $\left(  \varepsilon(n)=1\right)  $
normal to $u^{i},$ that is and $u^{i}n_{i}=\phi=0$. Then equations
(\ref{d.two.1}),(\ref{d.two.3}) and ( \ref{d.two.4}) become:
\begin{align}
P_{ij}\left(  u,n\right)   &  =g_{ij}+u_{i}u_{j}-n_{i}n_{j}\label{d.two.5}%
\\[0.03in]
R^{i} &  =-\left(  R^{j}u_{j}\right)  u^{i}+\left(  R^{j}n_{j}\right)
n^{i}+R^{j}P_{j}^{i}\label{d.two.6}\\
Y_{ij} &  =\left(  Y_{kr}u^{k}u^{r}\right)  u_{i}u_{j}+\left(  Y_{kr}%
n^{k}n^{r}\right)  n_{i}n_{j}\text{ }\label{d.two.7}\\
&  -\left(  Y_{kr}u^{k}n^{r}\right)  u_{i}n_{j}-\left(  Y_{kr}n^{k}%
u^{r}\right)  n_{i}u_{j}\nonumber\\
&  -\left(  Y_{sr}u^{s}\right)  P_{i}^{r}u_{j}-\left(  Y_{sr}u^{r}\right)
P_{j}^{s}u_{i}+\left(  Y_{sr}n^{s}\right)  P_{i}^{r}n_{j}+\left(  Y_{sr}%
n^{r}\right)  P_{j}^{s}n_{i}+Y_{kr}P_{i}^{k}P_{j}^{r}\nonumber
\end{align}

\end{corollary}

\bigskip In matrix form the decomposition of $Y_{ij}$ is written as
follows\footnote{The matrix form of the decomposition is useful in the
computation of the various irreducible parts by means of algebraic computing
programs. }:%
\begin{equation}
Y_{ij}\rightarrow\left(
\begin{tabular}
[c]{lll}%
$\left(  Y_{kr}u^{k}u^{r}\right)  $ & $-\left(  Y_{kr}u^{k}n^{r}\right)  $ &
$-\left(  Y_{sr}u^{r}P_{i}^{r}\right)  $\\
$-\left(  Y_{kr}n^{k}u^{r}\right)  $ & $\left(  Y_{kr}n^{k}n^{r}\right)  $ &
$\left(  Y_{sr}n^{r}\right)  P_{j}^{s}n_{i}$\\
$-\left(  Y_{sr}u^{r}\right)  P_{j}^{s}u_{i}$ & $\left(  Y_{sr}n^{s}\right)
P_{i}^{r}n_{j}$ & $Y_{kr}P_{i}^{k}P_{j}^{r}$%
\end{tabular}
\ \right)  .\label{d.two.7a}%
\end{equation}

By means of this decomposition we brake $Y_{ij}$ in irreducible parts which is
easier to study (because they are simpler).

\subsection{Applications}

We apply the previous general decomposition of an arbitrary second rank tensor
to various cases which interest the physical applications. In the following we
assume that the pair of vectors defining the 1+1+2 decomposition are the
vectors $u^{a},n^{a}$ \ which are timelike \ ($\varepsilon(u)=-1)$ and
spacelike ($n^{a}=1)$ respectively and in addition they are normal to each
other so that $\phi=0.$ Under these assumptions the decomposition of the
symmetric tensor becomes (\ref{d.two.7}) or (\ref{d.two.7a}) and the
coefficients of equation (\ref{d.two.4}) \ read:%
\begin{align*}
A_{00}  &  =Y_{11}=Y_{kr}u^{k}u^{r}\\
A_{01}  &  =-Y_{12}=-Y_{kr}u^{k}n^{r}\\
A_{10}  &  =-Y_{21}=-Y_{kr}n^{k}u^{r}\\
A_{11}  &  =Y_{22}=Y_{kr}n^{k}n^{r}\\
\widetilde{A_{i}}  &  =-Y_{1t}P_{i}^{t}=-Y_{sr}u^{s}P_{i}^{r}\\
\widetilde{B_{i}}  &  =Y_{2t}P_{i}^{t}=Y_{sr}n^{s}P_{i}^{r}\\
\widetilde{C_{i}}  &  =-Y_{t1}P_{i}^{t}=-Y_{sr}u^{r}P_{i}^{r}\\
\widetilde{D_{i}}  &  =Y_{t2}P_{i}^{t}=Y_{sr}n^{r}P_{j}^{s}\\
\widetilde{C_{ij}}  &  =P_{[i}^{k}P_{j]}^{r}Y_{k;r}\\
\widetilde{D_{ij}}  &  =\left[  P_{(i}^{k}P_{j)}^{r}-\frac{1}{2}P_{ij}%
P^{kr}\right]  Y_{k;r}\\
\widetilde{E}  &  =P^{kr}Y_{k;r}%
\end{align*}

\subsubsection{The 1+1+2 decomposition of the tensor $u_{a;b}$}

We find the 1+1+2 decomposition of the derivative $u_{i;j}$ where $u_{i}$ is
the timelike unit vector defining the decomposition. We compute:%
\begin{align*}
A_{00}  &  :u_{i;j}u^{i}u^{j}=0\\
A_{01}  &  :u_{i;j}u^{i}n^{b}=0\\
A_{10}  &  :u_{i;j}n^{i}u^{j}=\dot{u}^{i}n_{i}\\
A_{11}  &  :u_{i;j}n^{i}n^{j}=\overset{\ast}{u}^{i}n_{i}\\
A_{i}  &  :P_{i}^{j}u_{j;k}u^{k}=P_{i}^{j}\dot{u}_{j}\\
\widetilde{A_{i}}  &  :-u_{j;k}u^{j}P_{i}^{k}=0\\
\widetilde{B_{i}}  &  :u_{i;k}n^{i}P_{j}^{k}n^{i}=P_{i}^{k}\left(  \omega
_{jk}n^{j}+\sigma_{jk}n^{j}\right)  =P_{i}^{k}\left(  -\omega_{kr}n^{r}%
+\sigma_{kr}n^{r}\right) \\
\widetilde{C_{i}}  &  :-u_{j;k}u^{k}P_{i}^{j}=-P_{i}^{j}\dot{u}_{j}\\
\widetilde{D_{i}}  &  :u_{j;k}n^{k}P_{i}^{j}=P_{i}^{j}\overset{\ast}{u}_{j}\\
\mathcal{R}_{ij}  &  :P_{[i}^{k}P_{j]}^{r}u_{k;r}=P_{[i}^{k}P_{j]}^{r}%
\omega_{kr}\\
\mathcal{S}_{ij}  &  :\left[  P_{(i}^{k}P_{j)}^{r}-\frac{1}{2}P_{ij}%
P^{kr}\right]  u_{k;r}=\left[  P_{(i}^{k}P_{j)}^{r}-\frac{1}{2}P_{ij}%
P^{kr}\right]  \sigma_{kr}\\
\mathcal{E}  &  :P^{kr}u_{k;r}=-\sigma_{kr}n^{k}n^{r}+\frac{2}{3}\theta
\end{align*}
Therefore the 1+1+2 decomposition of $u_{i;j}$ is expressed as
\begin{align}
u_{i;j}  &  =-\left(  \dot{u}^{k}n_{k}\right)  n_{i}u_{j}+\left(  \sigma
_{rk}n^{r}n^{k}+\frac{1}{3}\theta\right)  n_{i}n_{j}-P_{i}^{k}\dot{u}_{k}%
u_{j}+P_{i}^{k}\overset{\ast}{u}_{k}n_{j}+P_{j}^{k}\left(  -\omega_{kr}%
n^{r}+\sigma_{kr}n^{r}\right)  n_{i}+\nonumber\\
&  +P_{[i}^{k}P_{j]}^{r}\omega_{kr}+\left(  P_{(i}^{k}P_{j)}^{r}-\frac{1}%
{2}P_{ij}P^{kr}\right)  \sigma_{kr}+\frac{1}{2}P_{ij}\left(  -\sigma_{kr}%
n^{k}n^{r}+\frac{2}{3}\theta\right)  \text{ \ } \label{du.1}%
\end{align}
or in a matrix form%
\begin{equation}
u_{i;j}\rightarrow%
\begin{pmatrix}
0 & 0 & 0\\
-\dot{u}^{k}n_{k} & \left(  \sigma_{rk}n^{r}n^{k}+\frac{1}{3}\theta\right)  &
P_{j}^{k}\left(  -\omega_{kr}n^{r}+\sigma_{kr}n^{r}\right) \\
-P_{i}^{k}\dot{u}_{k} & P_{i}^{j}\left(  \omega_{jk}n^{k}+\sigma_{jk}%
n^{k}\right)  & P_{[i}^{k}P_{j]}^{r}\omega_{kr}+\left(  P_{(i}^{k}P_{j)}%
^{r}-\frac{1}{2}P_{ij}P^{kr}\right)  \sigma_{kr}+\frac{1}{2}P_{ij}\left(
-\sigma_{kr}n^{k}n^{r}+\frac{2}{3}\theta\right)
\end{pmatrix}
\label{du.2}%
\end{equation}

\subsubsection{The 1+1+2 decomposition of the shear tensor $\sigma_{ab}$ and
the vorticity tensor $\omega_{ab}$}

For the stress and the vorticity tensors for the 1+1+2 decomposition we find
the result:%
\begin{align}
\sigma_{ij}  &  =\left(  \sigma_{kr}n^{k}n^{r}\right)  \left(  n_{i}%
n_{j}-\frac{1}{2}P_{ij}\right)  +2P_{(i}^{k}n_{j)}\sigma_{kr}n^{r}+\left(
P_{(i}^{k}P_{j)}^{r}-\frac{1}{2}P_{ij}P^{kr}\right)  \sigma_{kr}\label{du.3}\\
\omega_{ij}  &  =2P_{[i}^{k}n_{j]}\omega_{kr}n^{r}+\omega_{kr}P_{[i}^{k}%
P_{j]}^{r} \label{du.4}%
\end{align}
In matrix form the above expressions are presented as follows%
\begin{equation}
\sigma_{ij}\rightarrow%
\begin{pmatrix}
0 & 0 & 0\\
0 & \sigma_{kr}n^{k}n^{r} & P_{i}^{k}\sigma_{kr}n^{r}\\
0 & P_{j}^{k}\sigma_{kr}n^{r} & \left(  P_{(i}^{k}P_{j)}^{r}-\frac{1}{2}%
P_{ij}P^{kr}\right)  \sigma_{kr}-\frac{1}{2}P_{ij}\left(  \sigma_{kr}%
n^{k}n^{r}\right)  \text{ \ \ \ \ }%
\end{pmatrix}
\label{du.5}%
\end{equation}

\begin{equation}
\omega_{ij}\rightarrow%
\begin{pmatrix}
0 & 0 & 0\\
0 & 0 & -P_{i}^{k}\omega_{kr}n^{r}\text{ \ \ }\\
& -P_{j}^{k}\omega_{kr}n^{r} & \omega_{kr}P_{[i}^{k}P_{j]}^{r}%
\end{pmatrix}
. \label{du.6}%
\end{equation}

\subsubsection{The 1+1+2 decomposition of the tensor $n_{a;b}$}

We consider now the 1+1+2 decomposition of the derivative $n_{a;b}$ of the
unit space like vector $n_{a}$ defining the double congruence. We find:
\begin{align*}
A_{00}  &  :n_{i;j}u^{i}u^{j}=\dot{n}^{i}u_{i}=-\dot{u}^{i}n_{i}\\
A_{01}  &  :n_{i;j}u^{i}n^{j}=\overset{\ast}{n}^{i}u_{i}=-\overset{\ast}%
{u}^{i}n_{i}=-\left(  \sigma_{rk}n^{r}n^{k}+\frac{1}{3}\theta\right) \\
A_{10}  &  :n_{i;j}n^{i}u^{j}=0\\
A_{11}  &  :n_{i;j}n^{i}n^{j}=0\\
A_{i}  &  :P_{i}^{j}n_{j;k}u^{k}=P_{i}^{j}\dot{n}_{j}\\
\widetilde{A_{i}}  &  :P_{i}^{k}n_{j;k}u^{j}=-P_{i}^{k}u_{j;k}n^{j}=-P_{i}%
^{k}\left(  \omega_{jk}n^{j}+\sigma_{jk}n^{j}\right) \\
B_{i}  &  :P_{i}^{j}n_{j;k}n^{k}=P_{i}^{j}\overset{\ast}{n}_{j}\\
\widetilde{B_{i}}  &  :P_{j}^{k}n_{i;k}n^{i}=0\\
\widetilde{C_{ij}}  &  :P_{[i}^{k}P_{j]}^{r}n_{r;k}=\mathit{R}_{ij}\\
\widetilde{D_{ij}}  &  :\left[  P_{(i}^{k}P_{j)}^{r}-\frac{1}{2}P_{ij}%
P^{kr}\right]  n_{k;r}=\mathit{J}_{ij}\\
\widetilde{E}  &  :\frac{1}{2}P_{ij}P^{kr}n_{k;r}=\frac{1}{2}P_{ij}\mathit{E}%
\end{align*}
where we have set:%

\begin{equation}
P_{i}^{k}P_{j}^{r}n_{j;r}=\mathit{S}_{ij}+\mathit{R}_{ij}+\frac{1}%
{2}\mathit{E}P_{ij} \label{d.two.9}%
\end{equation}
where%
\begin{align}
\mathit{S}_{ij}  &  =P_{i}^{k}P_{j}^{r}n_{(k;r)}-\frac{1}{2}EP_{ij}%
\label{d.two.10}\\
\mathit{E}  &  =P^{ij}n_{i;j}\label{d.two.11}\\
\mathit{R}_{ij}  &  =P_{i}^{k}P_{j}^{r}n_{[j;r]} \label{d.two.12}%
\end{align}

We call $\mathit{S}_{ij}$ is the \textbf{screen shear tensor}, $\mathit{R}%
_{ij}$ the \textbf{screen rotation tensor} and $\mathit{E}$ \ the
\textbf{screen \ expansion} of the vector field $n^{i}.$

From their definition it is easy to show that the kinematic quantities \ of
$\ n^{i}$ satisfy the properties:
\begin{align*}
\mathit{R}_{ij}  &  =-\mathit{R}_{ji}\,,\,\mathit{R}_{ij}n^{j}=0\\
\mathit{S}_{ij}  &  =\mathit{S}_{ji}\,,\,\mathit{S}_{i}^{i}=0\,,\,\mathit{S}%
_{ij}n^{j}=0\\
\mathit{R}_{ij}  &  =h_{i}^{k}h_{j}^{r}\mathit{R}_{kr}\\
\mathit{S}_{ij}  &  =h_{i}^{k}h_{j}^{r}\mathit{S}_{kr}%
\end{align*}
Coming back to the decomposition formula (\ref{d.two.7}) we find: \
\begin{align}
n_{i;j}  &  =-\left(  \dot{u}^{k}n_{k}\right)  u_{i}u_{j}+\left(  \sigma
_{rk}n^{r}n^{k}+\frac{1}{3}\theta\right)  u_{i}n_{j}-P_{i}^{k}\dot{n}_{k}%
u_{j}+P_{j}^{k}\left(  \omega_{rk}n^{r}+\sigma_{jk}n^{j}\right)  u_{i}%
+P_{i}^{k}\overset{\ast}{n}_{k}n_{j}\nonumber\\
&  +\mathit{R}_{ij}+\mathit{J}_{ij}+\frac{1}{2}P_{ij}\mathit{E}\nonumber\\
&  =-\left(  \dot{u}^{k}n_{k}\right)  u_{i}u_{j}+\left(  \sigma_{rk}n^{r}%
n^{k}+\frac{1}{3}\theta\right)  u_{i}n_{j}-P_{i}^{k}\dot{n}_{k}u_{j}+P_{j}%
^{k}\left(  \omega_{rk}n^{r}\right)  u_{i}+P_{j}^{k}\omega_{kr}n^{r}%
u_{i}+P_{i}^{k}\overset{\ast}{n}_{k}n_{j}\nonumber\\
&  +P_{j}^{k}\dot{n}_{k}u_{i}-N_{k}u_{i}+\mathit{R}_{ij}+\mathit{J}_{ij}%
+\frac{1}{2}P_{ij}\mathit{E}\nonumber\\
&  =-\left(  \dot{u}^{k}n_{k}\right)  u_{i}u_{j}+\left(  \sigma_{rk}n^{r}%
n^{k}+\frac{1}{3}\theta\right)  u_{i}n_{j}-P_{i}^{k}\dot{n}_{k}u_{j}+P_{j}%
^{k}\left(  2\omega_{rk}n^{r}-N_{k}\right)  u_{i}+P_{i}^{k}\overset{\ast}%
{n}_{k}n_{j}+P_{j}^{k}\dot{n}_{k}u_{i}\nonumber\\
&  +\mathit{R}_{ij}+\mathit{J}_{ij}+\frac{1}{2}P_{ij}\mathit{E}
\label{DSoT.38}%
\end{align}
where, the \textbf{Screen or Greenberg vector} $N^{a}$ is defined as follows:%
\begin{equation}
N^{i}=P_{j}^{i}(\dot{n}^{j}-\overset{\ast}{u}^{j})=P_{j}^{i}L_{n}u^{j}.
\label{DSoT.38a}%
\end{equation}
Finally we have:%
\begin{align}
n_{i;j}  &  =-\left(  \dot{u}^{k}n_{k}\right)  u_{i}u_{j}+\left(  \sigma
_{rk}n^{r}n^{k}+\frac{1}{3}\theta\right)  u_{i}n_{j}-P_{i}^{k}\dot{n}_{k}%
u_{j}+\label{DSoT.39}\\
&  +P_{j}^{k}\left(  2\omega_{jk}n^{j}-N_{j}\right)  u_{i}+P_{i}^{k}%
\overset{\ast}{n}_{k}n_{j}+P_{j}^{k}\dot{n}_{k}u_{i}+\mathit{R}_{ij}%
+\mathit{J}_{ij}+\frac{1}{2}P_{ij}\mathit{E}%
\end{align}

\bigskip

and in matrix form:%
\begin{equation}
n_{i;j}\rightarrow%
\begin{pmatrix}
-\dot{u}^{i}n_{i} & \left(  \sigma_{rk}n^{r}n^{k}+\frac{1}{3}\theta\right)   &
P_{j}^{k}\dot{n}_{k}+P_{j}^{k}\left(  2\omega_{jk}n^{j}-N_{j}\right)  \\
0 & 0 & 0\text{ \ }\\
-P_{i}^{k}\dot{n}_{k} & P_{i}^{k}\overset{\ast}{n}_{k} & \mathit{R}%
_{ij}+\mathit{J}_{ij}+\frac{1}{2}P_{ij}\mathit{E}%
\end{pmatrix}
\end{equation}
The Greenberg vector is important because when $N^{a}=0$ the $L_{n}u^{j}$ is a
linear combination of the vectors $u^{a},n^{a}$ which is the condition that
the integral curves of the vector fields $u^{a},n^{a}$ form a surface. This
condition is used in the RMHD approximation as the condition that the magnetic
field is frozen in wrt the observers $u^{a},$ that is a charge moves always on
the same magnetic field line.

\subsubsection{The 1+1+2 decomposition of the energy-momentum tensor}

The energy-momentum tensor is a symmetric second rank tensor. We take
$Y_{ij}=T_{ij}$ and find that the 1+1+2 decomposition of this tensor ( for
$u^{i}u_{i}=-1\,\,,n^{i}n_{i}=1$ and \thinspace$\phi=0)$ is as follows:%
\begin{equation}
T_{ij}=\mu u_{i}u_{j}+vn_{i}s_{j}+\nu s_{i}n_{j}+\left(  p+\gamma\right)
n_{i}n_{j}+Q_{j}u_{i}+P_{j}n_{i}+Q_{i}u_{j}+P_{i}n_{j}+D_{ij}+\left(
p-\frac{1}{2}\gamma\right)  P_{ij}.\label{dem.1}%
\end{equation}

We know that the 1+3 decomposition of the energy stress tensor wrt the vector
$u^{i}$ is given by the expression:%
\begin{equation}
T_{ij}=\mu u_{i}u_{j}+ph_{ij}+2u_{(i}q_{j)}+\pi_{ij}\label{dem.2}%
\end{equation}

where:%
\[
q_{i}u^{i}=0,\pi_{ij}u^{j}=0,\pi_{ij}=\pi_{ji},\pi_{i}^{i}=0.
\]
The quantities of the 1+3 decomposition are related to the quantities of the
1+1+2 decomposition as follows:%
\begin{align}
q_{i}  &  =vn_{i}+Q_{i}\label{dem.3}\\
\pi_{ij}  &  =\gamma\left(  n_{i}n_{j}-\frac{1}{2}P_{ij}\right)
+2P_{(i}n_{j)}+D_{ij}. \label{dem.4}%
\end{align}

More on the 1+1+2 decomposition of the energy stress tensor we shall mention below.

\subsubsection{The 1+1+2 decomposition of the tensors $(\lambda u_{a})_{;b}$
and $(\phi n_{a})_{;b}$}

We shall need in our calculations the 1+1+2 decomposition of the quantities
$(\lambda u_{i});_{j}$ \ and $(\phi n_{i});_{j}$ \ where $\lambda$,$\phi$
\ are scalars (invariants). One derives easily the following results:
\begin{align}
(\lambda u_{a});_{b} &  =\lambda,_{b}u_{a}+\lambda u_{a;b}\nonumber\\
&  =\lambda_{,b}u^{b}u_{a}u_{b}+\lambda_{,b}n^{b}u_{a}n_{b}+P_{b}^{k}%
\lambda_{,k}u_{a}+\lambda u_{a;b}\nonumber\\
&  =-\dot{\lambda}u_{a}u_{b}+\overset{\ast}{\lambda}u_{a}n_{b}+P_{b}%
^{k}\lambda,_{k}u_{a}+\lambda u_{a;b}\label{DSoT.39a}%
\end{align}
and:
\begin{align}
\left(  \phi n_{a}\right)  _{;b} &  =\phi_{,b}n_{a}+\phi n_{a;b}\nonumber\\
&  =-\phi_{,b}u^{b}n_{a}u_{b}+\phi_{,b}n^{b}n_{a}n_{b}+P_{b}^{k}\phi_{,k}%
n_{a}+\phi n_{a;b}\nonumber\\
&  =-\dot{\phi}n_{a}u_{b}+\overset{\ast}{\phi}n_{a}n_{b}+P_{b}^{k}\phi
_{,k}n_{a}+\phi n_{a;b}\label{DSoT.39b}%
\end{align}

\subsubsection{The case of a general vector $\xi^{a}$}

We consider the vector field $\xi_{a}=-\lambda u_{a}+\phi n_{a}$ where
$u_{a},n_{a}$ are the vectors defining the the double congruence.  We
calculate the 1+1+2 decomposition of the covariant derivative $\xi_{a;b}$ of
the vector $\xi_{a}$ in the double congruence defined by the pair $u_{a}%
,n_{a}.$ We have:%
\begin{equation}
\xi_{a;b}=-\left(  \lambda u_{a}\right)  _{;b}+\left(  \phi n_{a}\right)
_{;b}\label{DSoT.40}%
\end{equation}
where we have computed $\left(  \lambda u_{a}\right)  _{;b},\left(  \phi
n_{a}\right)  _{;b}~$ in (\ref{DSoT.39a}) and (\ref{DSoT.39b}).

After some algebra we find that
\begin{align}
\xi_{a;b} &  =\left[  \dot{\lambda}-\phi\dot{u}^{k}n_{k}\right]  u_{a}%
u_{b}+\left[  -\overset{\ast}{\lambda}+\phi\left(  \sigma_{rk}n^{r}n^{k}%
+\frac{1}{3}\theta\right)  \right]  u_{a}n_{b}+\left[  \lambda\dot{u}^{k}%
n_{k}-\dot{\phi}\right]  n_{a}u_{b}\nonumber\\
&  +\left[  \overset{\ast}{\phi}-\lambda\left(  \sigma_{ab}n^{a}n^{b}+\frac
{1}{3}\theta\right)  \right]  n_{a}n_{b}\nonumber\\
&  +\left[  -\lambda_{,k}+\phi\dot{n}_{k}+\phi\left(  2\omega_{bk}n^{b}%
-N_{b}\right)  \right]  u_{a}P_{b}^{k}+\left[  \lambda\dot{u}_{k}-\phi\dot
{n}_{k}\right]  P_{a}^{k}u_{b}\nonumber\\
&  +\left[  \phi_{,k}+\lambda\left(  \omega_{kr}n^{r}-\sigma_{kr}n^{r}\right)
\right]  n_{a}P_{b}^{k}+\left[  \phi\overset{\ast}{n}_{k}-\lambda(\sigma
_{kr}+\omega_{kr})n^{r}\right]  n_{b}P_{a}^{k}\nonumber\\
&  +\left[  -\lambda\sigma_{kr}+\phi\mathit{S}_{ab}\right]  \left(  P_{(a}%
^{k}P_{b)}^{r}-\frac{1}{2}P_{ab}P^{kr}\right)  +\left[  \frac{1}{2}%
\lambda\left(  \sigma_{kr}n^{k}n^{r}-\frac{2}{3}\theta\right)  +\frac{1}%
{2}\phi\mathit{E}\right]  P_{ab}\label{DSoT.43}\\
&  +\left[  -\lambda\omega_{kr}+\phi\mathit{R}_{kr}\right]  P_{[a}^{k}%
P_{b]}^{r}.\text{ \ \ \ }\nonumber
\end{align}

Having computed $\xi_{a;b}$ we compute the symmetric and the antisymmetric
part (for reasons to be seen later). \ For the symmetric part we have:%
\begin{align}
\xi_{\left(  a;b\right)  }  &  =\left[  \dot{\lambda}-\phi\dot{u}^{k}%
n_{k}\right]  u_{a}u_{b}+\left[  -\overset{\ast}{\lambda}+\phi\left(
\sigma_{rk}n^{r}n^{k}+\frac{1}{3}\theta\right)  +\lambda\dot{u}^{k}n_{k}%
-\dot{\phi}\right]  u_{(a}n_{b)}\nonumber\\
&  +\left[  \overset{\ast}{\phi}-\lambda\left(  \sigma_{ab}n^{a}n^{b}+\frac
{1}{3}\theta\right)  \right]  n_{a}n_{b}\nonumber\\
&  +\left[  -\lambda_{,k}+\phi\left(  2\omega_{bk}n^{b}-N_{b}\right)
+\lambda\dot{u}_{k}\right]  u_{(a}P_{b)}^{k}\nonumber\\
&  +\left[  \phi_{,k}-2\lambda\sigma_{kr}n^{r}+\phi\overset{\ast}{n}%
_{k}\right]  n_{(a}P_{b)}^{k}\nonumber\\
&  +\left[  -\lambda\sigma_{kr}+\phi\mathit{S}_{ab}\right]  \left(  P_{(a}%
^{k}P_{b)}^{r}-\frac{1}{2}P_{ab}P^{kr}\right) \nonumber\\
&  +\frac{1}{2}\left[  \lambda\left(  \sigma_{kr}n^{k}n^{r}-\frac{2}{3}%
\theta\right)  +\phi\mathit{E}\right]  P_{ab}.\text{ \ } \label{DSoT.44}%
\end{align}

For the antisymmetric part $\xi_{\left[  a;b\right]  }$ we find
\begin{align}
\xi_{\left[  a;b\right]  } &  =\left[  -\overset{\ast}{\lambda}+\phi\left(
\sigma_{rk}n^{r}n^{k}+\frac{1}{3}\theta\right)  -\lambda\dot{u}^{k}n_{k}%
+\dot{\phi}\right]  u_{[a}n_{b]}+\label{DSoT.45}\\
&  +\left[  -\lambda_{,k}+2\phi\dot{n}_{k}+\phi\left(  2\omega_{bk}n^{b}%
-N_{b}\right)  -\lambda\dot{u}_{k}\right]  u_{[a}P_{b]}^{k}+\\
&  +\left[  \phi_{,k}+2\lambda\omega_{kr}n^{r}-\phi\overset{\ast}{n}%
_{k}\right]  n_{[a}P_{b]}^{k}+\left[  -\lambda\omega_{kr}+\phi\mathit{R}%
_{kr}\right]  P_{[a}^{k}P_{b]}^{r}.
\end{align}
Finally the trace $\xi_{;a}^{a}$ is found to be:
\begin{equation}
\xi_{;a}^{a}=-\dot{\lambda}-\lambda\theta+\phi\dot{u}^{b}n_{b}+\phi
\mathit{E}+\overset{\ast}{\phi}.\label{DSoT.46}%
\end{equation}

\section{Lie derivative and the 1+1+2 decomposition}

As in the case of 1+3 decomposition, we consider the 1+1+2 decomposition of
the Lie derivative of a vector and a second rank tensor in order to study the
effects of a collineation in the kinematics and the dynamics of a model spacetime.

\subsection{1+1+2 decomposition wrt a double congruence}

Consider the double congruence defined by the timelike vector field
$u^{i}\,,\left(  u^{i}u_{i}=-1\right)  $ and a spacelike vector field
$n^{i},\left(  n^{i}n_{i}=1\right)  $ and assume that $\phi=u^{i}n_{i}=0$.

Let $\xi^{i}$ a vector field that is given by the expression:%
\[
\xi^{i}=-\lambda u^{i}+\varphi n^{i}%
\]
where $\lambda=\xi^{i}u_{i}$ and $\varphi=\xi^{i}n_{i}$. From the 1+1+2
decomposition formula for vectors [see (\ref{d.two.6})] we have:%
\begin{align*}
\lambda_{;i}  &  =-\dot{\lambda}u_{i}+\lambda^{\ast}n_{i}+P_{i}^{j}%
\lambda_{;j}\\
\varphi_{;i}  &  =-\dot{\varphi}u_{i}+\varphi^{\ast}n_{i}+P_{i}^{j}%
\varphi_{;j}.%
\end{align*}

\begin{proposition}
The following identities are true:
\begin{align}
L_{\xi}u_{a}  &  =-\left[  \dot{\lambda}-\phi\dot{u}^{k}n_{k}\right]
u_{a}+\left[  \overset{\ast}{\lambda}-\lambda\left(  \dot{u}^{k}n_{k}\right)
\right]  n_{a}+P_{a}^{k}\left[  \lambda_{,k}-\lambda\dot{u}_{k}-2\phi
\omega_{bk}n^{b}\right] \label{d.S.T.V.1}\\
L_{\xi}u^{a}  &  =\left[  \dot{\lambda}-\phi\dot{u}^{k}n_{k}\right]
u^{a}-\left[  \dot{\phi}-\phi\left(  \sigma_{cd}n^{c}n^{d}+\frac{1}{3}%
\theta\right)  \right]  n^{a}-\phi N^{a}\label{d.S.T.V.1b}\\
L_{\xi}n_{a}  &  =\left[  -\dot{\phi}+\phi\left(  \sigma_{rk}n^{r}n^{k}%
+\frac{1}{3}\theta\right)  \right]  u_{a}+\left[  \overset{\ast}{\phi}%
-\lambda\left(  \sigma_{ab}n^{a}n^{b}+\frac{1}{3}\theta\right)  \right]
n_{a}+\label{d.S.T.V.2}\\
&  +P_{a}^{k}\left[  \lambda\left(  2\omega_{tk}n^{t}-N_{k}\right)
+2\lambda\dot{n}_{k}-\phi\overset{\ast}{n}_{k}-\phi_{;c}\right] \\
L_{\xi}n^{a}  &  =\left[  \overset{\ast}{\lambda}-\lambda\dot{u}^{k}%
n_{k}\right]  u^{a}-\left[  \overset{\ast}{\phi}-\lambda\left(  \sigma
_{ab}n^{a}n^{b}+\frac{1}{3}\theta\right)  \right]  n^{a}-\lambda N^{a}
\label{d.S.T.V.2b}%
\end{align}

\end{proposition}

where $N^{i}$ is the \textbf{Greenberg} or the \textbf{screen vector} which is
given by the expression
\begin{equation}
N^{i}=P_{j}^{i}L_{u}n^{j}=P_{j}^{i}(\dot{n}^{j}-\overset{\ast}{u}^{j}).
\label{d.S.T.V.3}%
\end{equation}

\begin{proof}
We have for (\ref{d.S.T.V.1}):%
\[
L_{\xi}u_{a}=u_{a;b}\xi^{b}+\xi_{~b;a}u^{b}%
\]%
\begin{align}
\xi_{a;b}  &  =\left[  \dot{\lambda}-\phi\dot{u}^{k}n_{k}\right]  u_{a}%
u_{b}+\left[  -\overset{\ast}{\lambda}+\phi\left(  \sigma_{rk}n^{r}n^{k}%
+\frac{1}{3}\theta\right)  \right]  u_{a}n_{b}+\left[  \lambda\dot{u}^{k}%
n_{k}-\dot{\phi}\right]  n_{a}u_{b}\nonumber\\
&  +\left[  \overset{\ast}{\phi}-\lambda\left(  \sigma_{ab}n^{a}n^{b}+\frac
{1}{3}\theta\right)  \right]  n_{a}n_{b}\nonumber\\
&  +\left[  -\lambda_{,k}+\phi\dot{n}_{k}+\phi\left(  2\omega_{bk}n^{b}%
-N_{b}\right)  \right]  u_{a}P_{b}^{k}+\left[  \lambda\dot{u}_{k}-\phi\dot
{n}_{k}\right]  P_{a}^{k}u_{b}\nonumber\\
&  +\left[  \phi_{,k}+\lambda\left(  \omega_{kr}n^{r}-\sigma_{kr}n^{r}\right)
\right]  n_{a}P_{b}^{k}+\left[  \phi\overset{\ast}{n}_{k}-\lambda(\sigma
_{kr}+\omega_{kr})n^{r}\right]  n_{b}P_{a}^{k}\nonumber\\
&  +\left[  -\lambda\sigma_{kr}+\phi\mathit{S}_{ab}\right]  \left(  P_{(a}%
^{k}P_{b)}^{r}-\frac{1}{2}P_{ab}P^{kr}\right)  +\left[  \frac{1}{2}%
\lambda\left(  \sigma_{kr}n^{k}n^{r}-\frac{2}{3}\theta\right)  +\frac{1}%
{2}\phi\mathit{E}\right]  P_{ab}\\
&  +\left[  -\lambda\omega_{kr}+\phi\mathit{R}_{kr}\right]  P_{[a}^{k}%
P_{b]}^{r}.\nonumber
\end{align}
The term:%
\[
\xi_{b;a}u^{b}=-\left[  \dot{\lambda}-\phi\dot{u}^{k}n_{k}\right]
u_{a}+\left[  \overset{\ast}{\lambda}-\phi\left(  \sigma_{rk}n^{r}n^{k}%
+\frac{1}{3}\theta\right)  \right]  n_{a}+\left[  \lambda_{,k}-\phi\dot{n}%
_{k}-\phi\left(  2\omega_{bk}n^{b}-N_{k}\right)  \right]  P_{a}^{k}%
\]
so that:%
\[
L_{\xi}u_{a}=u_{a;b}\xi^{b}-\left[  \dot{\lambda}-\phi\dot{u}^{k}n_{k}\right]
u_{a}-\left[  -\overset{\ast}{\lambda}+\phi\left(  \sigma_{rk}n^{r}n^{k}%
+\frac{1}{3}\theta\right)  \right]  n_{a}+\left[  \lambda_{,k}-\phi\dot{n}%
_{k}-\phi\left(  2\omega_{bk}n^{b}-N_{k}\right)  \right]  P_{a}^{k}%
\]
Similarly from (\ref{du.1}) we have:%
\begin{equation}
u_{a;b}\xi^{b}=-\lambda\left(  \dot{u}^{k}n_{k}\right)  n_{i}-\lambda
P_{i}^{k}\dot{u}_{k}+\phi\left(  \sigma_{rk}n^{r}n^{k}+\frac{1}{3}%
\theta\right)  n_{i}+\phi P_{i}^{k}\overset{\ast}{u}_{k}%
\end{equation}
Collecting terms we find:%
\begin{equation}
L_{\xi}u_{a}=-\left[  \dot{\lambda}-\phi\dot{u}^{k}n_{k}\right]  u_{a}+\left[
\overset{\ast}{\lambda}-\lambda\left(  \dot{u}^{k}n_{k}\right)  \right]
n_{a}+P_{a}^{k}\left[  \lambda_{,k}-\lambda\dot{u}_{k}-\phi2\omega_{bk}%
n^{b}\right]
\end{equation}
\qquad\qquad\newline We have for (\ref{d.S.T.V.1b})
\[
L_{\xi}u^{a}=\left(  L_{\xi}g^{ab}\right)  u_{b}+g^{ab}L_{\xi}u_{b}%
=-2\xi^{\left(  a;b\right)  }u_{b}+g^{ab}L_{\xi}u_{b}%
\]
The first term gives:%
\begin{align*}
-2\xi^{\left(  a;b\right)  }u_{b}  &  =2\left[  \dot{\lambda}-\phi\dot{u}%
^{k}n_{k}\right]  u^{a}+\left[  -\overset{\ast}{\lambda}+\phi\left(
\sigma_{rk}n^{r}n^{k}+\frac{1}{3}\theta\right)  +\lambda\dot{u}^{k}n_{k}%
-\dot{\phi}\right]  n^{a}+\\
&  +\left[  -\lambda_{,k}+\phi\left(  2\omega_{bk}n^{b}-N_{b}\right)
+\lambda\dot{u}_{k}\right]  P^{ka}%
\end{align*}
Replacing $L_{\xi}u_{b}$ from (\ref{d.S.T.V.1}) and collecting terms we find:%
\begin{equation}
L_{\xi}u^{a}=\left[  \dot{\lambda}-\phi\dot{u}^{k}n_{k}\right]  u^{a}+\left[
\phi\left(  \sigma_{rk}n^{r}n^{k}+\frac{1}{3}\theta\right)  -\dot{\phi
}\right]  n^{a}-\phi N^{a}%
\end{equation}
Concerning $L_{\xi}n_{a}$ we have:%
\[
L_{\xi}n_{a}=n_{a;b}\xi^{b}+\xi_{~b;a}n^{b}.
\]
From (\ref{d.S.T.V.2}) we find:%
\begin{equation}
\xi_{b;a}n^{b}=\left[  \lambda\dot{u}^{k}n_{k}-\dot{\phi}\right]
u_{a}+\left[  \overset{\ast}{\phi}-\lambda\left(  \sigma_{ab}n^{a}n^{b}%
+\frac{1}{3}\theta\right)  \right]  n_{a}+\left[  \phi_{,k}+\lambda\left(
\omega_{kr}n^{r}-\sigma_{kr}n^{r}\right)  \right]  P_{a}^{k}%
\end{equation}
Similarity from (\ref{DSoT.39}) we have:%
\begin{equation}
n_{a;b}\xi^{b}=-\lambda\left(  \dot{u}^{k}n_{k}\right)  u_{j}-\lambda
P_{j}^{k}\dot{n}_{k}+\phi\left(  \sigma_{rk}n^{r}n^{k}+\frac{1}{3}%
\theta\right)  u_{i}+\phi P_{i}^{k}\overset{\ast}{n}_{k}%
\end{equation}
Replacing we find:%
\begin{align*}
L_{\xi}n_{a}  &  =\left[  -\dot{\phi}+\phi\left(  \sigma_{rk}n^{r}n^{k}%
+\frac{1}{3}\theta\right)  \right]  u_{a}+\left[  \overset{\ast}{\phi}%
-\lambda\left(  \sigma_{ab}n^{a}n^{b}+\frac{1}{3}\theta\right)  \right]
n_{a}+\\
&  -P_{a}^{k}\left[  \lambda\left(  2\omega_{tk}n^{t}-N_{k}\right)
+2\lambda\dot{n}_{k}-\phi\overset{\ast}{n}_{k}-\phi_{;c}\right]
\end{align*}
where $N^{a}=P_{b}^{a}L_{n}u^{a}$ is the Greenberg Vector. \

Concerning $L_{\xi}n^{a}$ we have:%
\[
L_{\xi}n^{a}=\left(  L_{\xi}g^{ab}\right)  n_{b}+g^{ab}L_{\xi}n_{a}%
=-2\xi^{\left(  a;b\right)  }n_{b}+g^{ab}L_{\xi}n_{b}%
\]
Again from (\ref{d.S.T.V.2}) we have for the first term:%
\begin{align}
-2\xi^{\left(  a;b\right)  }n_{b}  &  =-\left[  -\overset{\ast}{\lambda}%
+\phi\left(  \sigma_{rk}n^{r}n^{k}+\frac{1}{3}\theta\right)  +\lambda\dot
{u}^{k}n_{k}-\dot{\phi}\right]  u^{a}+\\
&  -2\left[  \overset{\ast}{\phi}-\lambda\left(  \sigma_{ab}n^{a}n^{b}%
+\frac{1}{3}\theta\right)  \right]  n^{a}-\left[  \phi_{,k}-2\lambda
\sigma_{kr}n^{r}+\phi\overset{\ast}{n}_{k}\right]  P^{ka}%
\end{align}
Using (\ref{d.S.T.V.2}) \ and collecting terms we find:%
\[
L_{\xi}n^{a}=\left[  \overset{\ast}{\lambda}-\lambda\dot{u}^{k}n_{k}\right]
u^{a}-\left[  \overset{\ast}{\phi}-\lambda\left(  \sigma_{ab}n^{a}n^{b}%
+\frac{1}{3}\theta\right)  \right]  n^{a}-\lambda N^{a}.
\]

\end{proof}

\begin{proposition}
Let $X^{i}$ be an arbitrary vector field and $u^{i}$ a unit timelike vector
field. Then we have the identity:%
\begin{equation}
L_{u}X^{i}=-\left(  u^{r}u^{j}L_{u}g_{rj}\right)  u^{i}+\left[  \left(
X^{j}n_{j}\right)  ^{\cdot}+2\mathit{R}_{rj}X^{r}u^{j}-\phi_{;r}X^{r}\right]
n^{i}+P_{j}^{i}L_{u}X^{i} \label{d.S.T.V.4}%
\end{equation}
where $\phi=u^{i}n_{i}$.
\end{proposition}

\begin{proof}
From equation (\ref{d.two.6}) \ we have:%
\[
L_{u}X^{i}=-\left(  (L_{u}X^{j})u_{j}\right)  u^{i}+\left(  n_{j}L_{u}%
X^{j}\right)  n^{i}+P_{j}^{i}L_{u}X^{j}%
\]
Contracting with $u^{i},n^{i}$ we find:%
\begin{align*}
(L_{u}X^{j})u_{j}  &  =X_{(j;r)}u^{r}u^{i}=u^{r}u^{i}L_{u}g_{ij}\\
n_{j}L_{u}X^{j}  &  =\left(  X^{j}n_{j}\right)  ^{\cdot}+2n_{\left[
j;r\right]  }X^{r}u^{j}-\phi_{;r}X^{r}%
\end{align*}
Hence%
\[
L_{u}X^{i}=-\left(  u^{r}u^{j}L_{u}g_{rj}\right)  u^{i}+\left[  \left(
X^{j}n_{j}\right)  ^{\cdot}+2n_{\left[  j;r\right]  }X^{r}u^{j}-\phi_{;r}%
X^{r}\right]  n^{i}+P_{j}^{i}L_{u}X^{i}%
\]

\end{proof}

Working in a similar manner it is easy to prove the relations:
\begin{align}
L_{n}X^{i}  &  =-\left[  \left(  X^{j}u_{j}\right)  ^{\cdot}+2u_{\left[
j;r\right]  }X^{r}u^{j}-\phi_{;r}X^{r}\right]  u^{i}+\left(  n^{r}n^{j}%
L_{u}g_{rj}\right)  n^{i}+P_{j}^{i}L_{u}X^{i}\label{d.S.T.V.5}\\
L_{u}n^{i}  &  =-\left(  u^{r}u^{j}L_{u}g_{rj}\right)  u^{i}+P_{j}^{i}%
n^{j}\label{d.S.T.V.6}\\
L_{n}u^{i}  &  =\left(  n^{r}n^{j}L_{u}g_{rj}\right)  n^{i}+P_{j}^{i}u^{j}.
\label{d.S.T.V.7}%
\end{align}

\subsection{The Lie derivative of the projection tensors}

In the case $\xi^{a}=-\lambda u^{a}+\phi n^{a}$ we compute the Lie
derivative of the projection tensors $h_{ab}=g_{ab}+u_{a}u_{a}$ and
$P_{ab}=h_{ab}-n_{a}n_{a}.$

We have:%
\begin{align}
L_{\xi}h_{ab}  &  =L_{\xi}g_{ab}+\left(  L_{\xi}u_{a}\right)  u_{b}%
+u_{a}\left(  L_{\xi}u_{b}\right) \nonumber\\
&  =2\left[  -\dot{\phi}+\phi\left(  \sigma_{rk}n^{r}n^{k}+\frac{1}{3}%
\theta\right)  \right]  u_{(a}n_{b)}-2\phi N_{(b}u_{a)}+2\left[  \overset
{\ast}{\phi}-\lambda\left(  \sigma_{ab}n^{a}n^{b}+\frac{1}{3}\theta\right)
\right]  n_{(a}n_{b)}\nonumber\\
&  +2\left[  \phi,_{c}-2\lambda\sigma_{cd}n^{d}+\phi\overset{\ast}{n}%
_{c}\right]  n_{(a}P_{b)}^{c}\text{ \ }\nonumber\\
&  +2\left[  -\lambda\sigma_{kr}+\phi\mathit{S}_{kr}\right]  \left(
P_{(a}^{k}P_{b)}^{r}-\frac{1}{2}P_{ab}P^{kr}\right)  +\left[  \lambda\left(
\sigma_{kr}n^{k}n^{r}-\frac{2}{3}\theta\right)  +\phi\mathit{E}\right]
P_{ab}~ \label{DSoT.49}%
\end{align}

We also have:%
\begin{align}
L_{\xi}h_{b}^{a}  &  =\left(  L_{\xi}u^{a}\right)  u_{b}+u^{a}\left(  L_{\xi
}u_{b}\right) \nonumber\\
&  =\left[  -\dot{\phi}+\phi\left(  \sigma_{rk}n^{r}n^{k}+\frac{1}{3}%
\theta\right)  \right]  n^{a}u_{b}+\left[  \overset{\ast}{\lambda}-\lambda
\dot{u}^{k}n_{k}\right]  u^{a}n_{b}\nonumber\\
&  -\phi N^{a}u_{b}+P_{b}^{c}\left[  \lambda_{;c}-\lambda\dot{u}_{c}%
-2\phi\omega_{kc}n^{k}\right]  u^{a}\text{ } \label{DSoT.50}%
\end{align}
and
\begin{align*}
L_{\xi}h^{ab}  &  =L_{\xi}g^{ab}+\left(  L_{\xi}u^{a}\right)  u^{b}%
+u^{a}\left(  L_{\xi}u^{b}\right) \\
&  =2\left[  \overset{\ast}{\lambda}-\lambda\dot{u}^{k}n_{k}\right]
n^{(a}u^{b)}+2u^{(a}P^{b)c}\left[  \lambda_{;c}-\lambda\dot{u}_{c}-2\phi
\omega_{dc}n^{d}\right] \\
&  -2\left[  \overset{\ast}{\phi}-\lambda\left(  \sigma_{ab}n^{a}n^{b}%
+\frac{1}{3}\theta\right)  \right]  n^{a}n^{b}-2\left[  \phi,_{c}%
-2\lambda\sigma_{cd}n^{d}+\phi\overset{\ast}{n}_{c}\right]  n^{(a}P^{b)c}+\\
&  -\left[  \lambda\left(  \sigma_{kr}n^{k}n^{r}-\frac{2}{3}\theta\right)
+\phi\mathit{E}\right]  P^{ab}-2\left[  -\lambda\sigma^{cd}+\phi
\mathit{S}^{cd}\right]  \left(  P_{c}^{a}P_{d}^{b}-\frac{1}{2}P_{ab}%
P^{kr}\right)
\end{align*}
Working with a similar manner we compute the Lie derivative of the screen
projection tensor $P_{ab}=h_{ab}-n_{a}n_{b}.$

We find:%
\begin{align}
L_{\xi}P_{ab} &  =L_{\xi}h_{ab}-2L_{\xi}n_{(a}n_{b)}\nonumber\\
&  =2\lambda N_{(a}n_{b)}-2\phi N_{(a}u_{b)}+2\left[  -\lambda\sigma_{kr}%
+\phi\mathit{S}_{kr}\right]  \left(  P_{(a}^{k}P_{b)}^{r}-\frac{1}{2}%
P_{ab}P^{kr}\right)  \nonumber\\
&  +\left[  \lambda\left(  \sigma_{kr}n^{k}n^{r}-\frac{2}{3}\theta\right)
+\phi\mathit{E}\right]  P_{ab}\text{ }\label{DSoT.51}%
\end{align}%
\begin{align}
L_{\xi}P_{b}^{a} &  =L_{\xi}h_{b}^{a}-\left(  L_{\xi}n^{a}\right)  n_{b}%
-n^{a}\left(  L_{\xi}n_{b}\right)  \nonumber\\
&  =\left[  \lambda_{,k}-\lambda\dot{u}_{k}-2\phi\omega_{rk}n^{r}\right]
P_{b}^{k}u^{a}-\phi N^{a}u_{b}+\left[  2\lambda\sigma_{kd}n^{d}-\phi_{,k}%
-\phi\overset{\ast}{n}_{k}\right]  P_{b}^{k}n^{a}+\lambda(N^{a}n_{b}%
+n^{a}N_{b})\label{DSoT.52}%
\end{align}

\begin{align}
L_{\xi}P^{ab} &  =2[\phi N_{c}-2\phi\omega_{tc}n^{t}-\lambda\dot{u}%
_{c}+\lambda_{;c}]u^{(a}P^{b)c}-2\phi u^{(a}N^{b)}\nonumber\\
&  -2[\phi_{;c}-2\lambda\sigma_{cd}n^{d}+\phi\overset{\ast}{n_{c}}%
]n^{(a}P^{b)c}+2\lambda n^{(a}N^{b)}\nonumber\\
&  -2[\phi\mathit{S}^{cd}-\lambda\sigma^{cd}]\left(  P_{(c}^{a}P_{d)}%
^{b}-\frac{1}{2}P_{cd}P^{ab}\right)  -\left[  \lambda\left(  \sigma_{kr}%
n^{k}n^{r}-\frac{2}{3}\theta\right)  +\phi\mathit{E}\right]  P^{ab}%
.\label{DSoT.53}%
\end{align}

\subsection{The components of the generic symmetry in terms of the kinematic
parameters}

We express $L_{\xi}g_{ab}$ (the generic symmetry) in terms of the kinematic
quantities of the double congruence. We introduce the trace and the trace free
part of the $L_{\xi}g_{ab}$ be means of the identity:%
\begin{equation}
L_{\xi}g_{ab}=2\psi g_{ab}+2H_{ab}\text{ },\,\,\,H_{a}^{a}=0\text{ and
}H_{\left[  ab\right]  }=0. \label{G.E.S.01}%
\end{equation}
The quantities $\psi,H_{ab}$ we call the components of the generic symmetry.

Now for a general vector $\xi^{a}$ we have the decomposition:%
\[
\xi_{a;b}=\xi_{\left(  a;b\right)  }+\xi_{\left[  a;b\right]  }%
\]
But we know that:%
\begin{equation}
2\xi_{\left(  a;b\right)  }=L_{\xi}g_{ab} \label{G.E.S.01a}%
\end{equation}
therefore we obtain:%
\[
\xi_{a;b}=\psi g_{ab}+H_{ab}+\xi_{\left[  a;b\right]  }.
\]
From (\ref{G.E.S.01a}) we compute the quantities $\psi$ and $H_{ab}$ in terms
of the kinematic quantities of the double congruence. Indeed writing
(\ref{G.E.S.01a}) in the form:%
\begin{equation}
2\xi_{\left(  a;b\right)  }=L_{\xi}g_{ab}=2\psi g_{ab}+2H_{ab}\text{ }
\label{G.E.S.01b}%
\end{equation}
and taking the trace we find:%
\begin{equation}
2\xi_{;a}^{a}=8\psi\Rightarrow\psi=\frac{1}{4}\left[  -\dot{\lambda}%
-\lambda\theta+\phi\dot{u}^{b}n_{b}+\phi\mathit{E}+\overset{\ast}{\phi
}\right]  \label{G.E.S.02}%
\end{equation}
where we have used (\ref{DSoT.46}).

Concerning the tensor $H_{ab}$ we have from (\ref{G.E.S.01b}):%
\[
H_{ab}=\xi_{\left(  a;b\right)  }-\psi g_{ab}.
\]
Replacing $\xi_{\left(  a;b\right)  }$ from (\ref{G.E.S.01}) and $\psi$ \ from
(\ref{G.E.S.02}) we find after standard calculations:%
\begin{align}
H_{ab}  &  =\frac{1}{2}\left[  3\dot{\lambda}-3\phi\left(  \dot{u}^{k}%
n_{k}\right)  -\lambda\theta+\phi\mathit{E}+\overset{\ast}{\phi}\right]
u_{a}u_{b}\nonumber\\
&  +\left[  -\overset{\ast}{\lambda}+\phi\left(  \sigma_{kr}n^{k}n^{r}%
+\frac{1}{3}\theta\right)  +\lambda\left(  \dot{u}^{k}n_{k}\right)  -\dot
{\phi}\right]  u_{(a}n_{b)}\nonumber\\
&  +\frac{1}{2}\left[  3\overset{\ast}{\phi}-4\lambda\sigma_{kr}n^{k}%
n^{r}-\frac{1}{3}\lambda\theta-\phi\dot{u}^{k}n_{k}-\phi\mathit{E}%
+\dot{\lambda}\right]  n_{a}n_{b}\nonumber\\
&  +\left[  -\lambda_{,k}+\phi\left(  2\omega_{kr}n^{k}-N_{k}\right)
+\lambda\dot{u}_{k}\right]  u_{(a}P_{b)}^{r}\nonumber\\
&  +\left[  \phi_{,k}-2\lambda\sigma_{kr}n^{r}+\phi\overset{\ast}{n}%
_{k}\right]  n_{(a}P_{b)}^{k}\nonumber\\
&  +\left[  -\lambda\sigma_{kr}+\phi\mathit{S}_{ab}\right]  \left(  P_{(a}%
^{k}P_{b)}^{r}-\frac{1}{2}P_{ab}P^{kr}\right) \label{G.E.S.03}\\
&  +\frac{1}{2}\left[  \dot{\lambda}-\phi\dot{u}^{k}u_{k}-\frac{1}{3}%
\lambda\theta+2\lambda\sigma_{kr}n^{k}n^{r}+\phi\mathit{E}-\overset{\ast}%
{\phi}\right]  P_{ab}.\nonumber
\end{align}
In matrix form this result is written as follows:%
\begin{equation}
H_{ab}\rightarrow%
\begin{pmatrix}
\frac{1}{2}\left[  3\dot{\lambda}-3\phi\left(  \dot{u}^{k}n_{k}\right)
-\lambda\theta+\phi\mathit{E}+\overset{\ast}{\phi}\right]  & \frac{1}%
{2}\left[  -\overset{\ast}{\lambda}+\phi\left(  \sigma_{kr}n^{k}n^{r}+\frac
{1}{3}\theta\right)  +\lambda\left(  \dot{u}^{k}n_{k}\right)  -\dot{\phi
}\right]  & H_{13}\\
& \frac{1}{2}\left[  3\overset{\ast}{\phi}-4\lambda\sigma_{kr}n^{k}n^{r}%
-\frac{1}{3}\lambda\theta-\phi\dot{u}^{k}n_{k}-\phi\mathit{E}+\dot{\lambda
}\right]  & H_{23}\\
&  & H_{33}%
\end{pmatrix}
. \label{G.E.S.03a}%
\end{equation}

where:%
\begin{align*}
H_{13}  &  =\frac{1}{2}\left[  -\lambda_{,k}+\phi\left(  2\omega_{kr}%
n^{k}-N_{k}\right)  +\lambda\dot{u}_{k}\right]  P_{b}^{r}\\
H_{23}  &  =\frac{1}{2}\left[  \phi_{,k}-2\lambda\sigma_{kr}n^{r}+\phi
\overset{\ast}{n}_{k}\right]  P_{b}^{k}\\
H_{33}  &  =\left[  -\lambda\sigma_{kr}+\phi\mathit{S}_{ab}\right]  \left(
P_{(a}^{k}P_{b)}^{r}-\frac{1}{2}P_{ab}P^{kr}\right)  +\frac{1}{2}\left[
\dot{\lambda}-\phi\dot{u}^{k}u_{k}-\frac{1}{3}\lambda\theta+2\lambda
\sigma_{kr}n^{k}n^{r}+\phi\mathit{E}-\overset{\ast}{\phi}\right]  P_{ab}.
\end{align*}

These expressions can be used to answer the question: given $\xi^{a},n^{a}$
what type of collineation $\xi^{a}$ can be  and under what conditions? Below,
we express the kinematic quantities of the double congruence in terms of the
symmetry parameters $\psi,H_{ab}.$

\subsection{The kinematic implications of a collineation}

The following Theorem gives the kinematic implications of a collineation, that
is, expresses the kinematic quantities in terms of the collineation parameters
$\psi,H_{ab}$ in the 1+1+2 decomposition

\begin{theorem}
Suppose $\xi^{i}=\xi n^{i}$, $n^{i}n_{i}=\varepsilon\left(  n\right)  $ and
$s^{i}s_{i}=\varepsilon\left(  s\right)  $ two non-null normalized vector
fields such that $s^{i}n_{i}=\phi$ where $\varepsilon\left(  s\right)
,\varepsilon\left(  n\right)  =\pm1$ \ are the signatures of the vectors, the
sign $+\,\ $applying to a spacelike vector and the $-$ sign to a timelike
vector. Let $P_{ij}=P_{ij}\left(  s,n\right)  $ be the projective tensor
associated with the double congruence consisting of the vector fields $s^{i}$
and $n^{i}$. Then equation (\ref{DSoT.38}) is equivalent to the following
conditions:
\begin{align}
\mathit{S}_{ij} &  =\frac{1}{\xi}\left[  P_{i}^{k}P_{j}^{r}-\frac{1}{2}%
P^{kr}P_{ij}\right]  H_{kr}\label{ColDecom.2}\\
\dot{n}^{i}s_{i} &  =\frac{1}{\xi}\varepsilon\left(  s\right)  \psi
-\phi\left(  \ln\xi\right)  ^{\bullet}+\frac{1}{\xi}H_{11}\label{ColDecom.3}\\
n_{i}^{\ast} &  =\frac{\phi n_{i}-\varepsilon\left(  n\right)  s_{i}}{\phi
^{2}-\varepsilon\left(  s\right)  \varepsilon\left(  n\right)  }\left[
\frac{1}{\xi}\phi\psi-\varepsilon\left(  n\right)  \left(  \ln\xi\right)
^{\bullet}+\frac{2}{\xi}H_{21}-\frac{\varepsilon\left(  n\right)  }{\xi}\phi
H_{22}\right]  +\label{ColDecom.4}\\
&  +P_{i}^{j}\left[  -\varepsilon\left(  n\right)  \left(  \ln\xi\right)
_{;j}+\frac{2}{\xi}H_{j2}\right]  \,\\
\xi^{\ast} &  =\psi+\varepsilon(n)H_{22}\label{ColDecom.5}\\
\mathit{E} &  =\frac{2\psi}{\xi}+\frac{1}{\xi}P^{ij}H_{ij}\label{ColDecom.6}\\
N_{i} &  =-2\omega_{ij}n^{j}-P_{i}^{j}\phi\left[  \ln\left(  \xi\left\vert
\phi\right\vert \right)  \right]  _{;b}+\frac{2}{\xi}P_{i}^{j}H_{j1}%
-\varepsilon\left(  s\right)  \phi P_{i}^{j}\dot{s}_{j}\label{ColDecom.7}%
\end{align}
where $\dot{\upsilon}=\upsilon_{;k}s^{k},\upsilon^{\ast}=\upsilon_{;k}n^{k}$
and $N_{i}=\left(  n_{i;j}u^{j}-u_{i;j}n^{j}\right)  =\left(  \dot{n}%
_{i}-u_{i}^{\ast}\right)  $ is the Greenberg vector.
\end{theorem}

\begin{corollary}
Let $\xi^{i}=\xi n^{i}$, $n^{i}n_{i}=\varepsilon\left(  n\right)  $ be a CKV
and $s^{i},$ $s^{i}s_{i}=\varepsilon\left(  s\right)  $ be a unit vector such
that $n^{i}s_{i}=\phi$. Then the conformal symmetry condition $L_{\xi}%
g_{ij}=2\psi g_{ij}$ is equivalent to the conditions:%
\begin{align}
\mathit{S}_{ij}  &  =0\label{ColDecom.20}\\
\dot{n}^{i}s_{i}  &  =\frac{1}{\xi}\left[  \varepsilon\left(  s\right)
\psi-\phi\left(  \ln\xi\right)  ^{\bullet}\right] \label{ColDecom.21}\\
n_{i}^{\ast}  &  =\frac{1}{\phi^{2}-\varepsilon\left(  s\right)
\varepsilon\left(  n\right)  }\left(  \phi n_{i}-\varepsilon\left(  n\right)
s_{i}\right)  \left[  \frac{1}{\xi}\phi\psi-\varepsilon\left(  n\right)
\left(  \ln\xi\right)  ^{\bullet}\right]  -\varepsilon\left(  n\right)
P_{i}^{j}\left(  \ln\xi\right)  _{;j}\,\label{ColDecom.22}\\
\xi^{\ast}  &  =\psi\label{ColDecom.23}\\
\mathit{E}  &  =\frac{2\psi}{\xi}\label{ColDecom.24}\\
N_{i}  &  =-2\omega_{ij}n^{j}-P_{i}^{j}\phi\left[  \ln\left(  \xi\left\vert
\phi\right\vert \right)  \right]  _{;b}-\varepsilon\left(  s\right)  \phi
P_{i}^{j}\dot{s}_{j}. \label{ColDecom.25}%
\end{align}

\end{corollary}

\begin{corollary}
In the special case $n^{i}n_{i}=1,s^{i}s_{i}=-1,n^{i}s_{i}=0$ \ we have that
$\xi^{i}=\xi n^{i}$ is a spacelike CKV iff:%
\begin{align}
\mathit{S}_{ij}  &  =0\\
\dot{n}^{i}s_{i}  &  =-\frac{\psi}{\xi}\\
n_{i}^{\ast}  &  =s_{i}\left(  \ln\xi\right)  ^{\bullet}-P_{i}^{j}\left(
\ln\xi\right)  _{;j}\,\\
\xi^{\ast}  &  =\psi\\
\mathit{E}  &  =\frac{2\psi}{\xi}\\
N_{i}  &  =-2\omega_{ij}n^{j}.
\end{align}

\end{corollary}

\section{The Lie derivative of $R_{ab}$ wrt a collineation $\xi_{a}$.}

In this section we compute the Lie derivative of the Einstein equations in the
1+1+2 decomposition. To do that we write these equations in the form:
\begin{equation}
R_{ab}=T_{ab}-\frac{1}{2}g_{ab}T+\Lambda g_{ab}.\label{LR.01}%
\end{equation}
The Lie derivative of Einstein field equations wrt a general vector
field\footnote{A vector $\xi^{a}$ can always be written as
$\xi^{a}=-\lambda u^{a}+\phi n^{a}$ where $n^{a} // h^{a}_{b}\xi^{b}$.} $\xi^{a}$ is:%
\begin{equation}
L_{\xi}R_{ab}=L_{\xi}\left[  T_{ab}-\frac{1}{2}g_{ab}T+\Lambda g_{ab}\right]
.\label{LR.05a}%
\end{equation}

The lhs $L_{\xi}R_{ab}$ is of a pure geometric nature and can be computed in
terms of the generic symmetry $L_{\xi}g_{ab}$ or,  in
terms of the geometric parameters $\psi,H_{ab.}$ If we do that then we shall
have expressed the Lie derivative of energy momentum tensor in terms of
geometric variables. Because the generic symmetry describes any symmetry, this
means that by doing this we have solved Einstein field equations for the
symmetry defined by the symmetry parameters $\psi,H_{ab.}$

We consider next the dynamic variables $\mu,p,q^{a},\pi_{ab}$ defined by the
1+3 decomposition of the energy momentum tensor wrt any observers $u^{a}. $ As
we have shown (see equation (\ref{TE.5}) ) the 1+3 decomposition of the energy
momentum tensor wrt the 4-velocity $u^{a}$ is given by the relation:%
\begin{align}
T_{ab}  &  =\mu u_{a}u_{b}+ph_{ab}+q_{a}u_{b}+q_{b}u_{a}+\pi_{ab}%
\label{LR.02}\\
T  &  =g_{ab}T^{ab}=3p-\mu. \label{LR.03}%
\end{align}
Replacing in equation (\ref{LR.01}) we find:%
\begin{align}
R_{ab}  &  =\left(  \mu+p\right)  u_{a}u_{b}+\frac{1}{2}\left(  \mu
-p+2\Lambda\right)  g_{ab}+q_{a}u_{b}+q_{b}u_{a}+\pi_{ab}\label{LR.04}\\
&  \text{or}\nonumber\\
R_{ab}  &  =\left(  \frac{1}{2}\mu+\frac{3}{2}p-\Lambda\right)  u_{a}%
u_{b}+\left(  \frac{1}{2}\mu-\frac{1}{2}p+\Lambda\right)  h_{ab}+q_{a}%
u_{b}+q_{b}u_{a}+\pi_{ab}. \label{LR.05}%
\end{align}

In the 1+1+2 decomposition wrt a double congruence $u^{a},n^{a}$ we have shown
(see equation (\ref{dem.1})) that:
\begin{equation}
T_{ab}=\mu u_{a}u_{b}+vn_{a}u_{b}+\nu u_{a}n_{b}+\left(  p+\gamma\right)
n_{a}n_{b}+Q_{b}u_{a}+P_{b}n_{a}+Q_{a}u_{b}+P_{a}n_{b}+D_{ab}+\left(
p-\frac{1}{2}\gamma\right)  P_{ab}\label{LR.05b}%
\end{equation}
where:%
\begin{align}
q_{a}u^{a} &  =0,\pi_{ab}u^{b}=0,\pi_{ab}=\pi_{ba},\pi_{a}^{a}=0\nonumber\\
q_{a} &  =vn_{a}+Q_{a}\label{LR.05c}\\
\pi_{ab} &  =\gamma\left(  n_{a}n_{b}-\frac{1}{2}P_{ab}\right)  +2P_{(a}%
n_{b)}+D_{ab}.\label{LR.05d}%
\end{align}
Therefore for this decomposition the Ricci tensor reads:%
\begin{align}
R_{ab} &  =\left(  \frac{1}{2}\mu+\frac{3}{2}p-\Lambda\right)  u_{a}%
u_{b}+vn_{a}u_{b}+\nu u_{a}n_{b}+\left(  \frac{1}{2}\mu-\frac{1}{2}%
p+\gamma+\Lambda\right)  n_{a}n_{b}+\nonumber\\
&  +Q_{b}u_{a}+P_{b}n_{a}+Q_{a}u_{b}+P_{a}n_{b}+D_{ab}+\left(  \frac{1}{2}%
\mu-\frac{1}{2}p-\frac{1}{2}\gamma+\Lambda\right)  P_{ab}.\label{LR.05e}%
\end{align}
The conclusion is that by working in this manner we will express the field
equations as Lie derivatives of the dynamic variables either in the 1+3 or in
the 1+1+2 decomposition in terms of the geometric parameters $\psi,H_{ab.}%
$defining the symmetry vector $\xi^{a}$.

Before we compute the Lie derivative $L_{\xi}R_{ab}$ we note that any vector
field $\xi^{a}$ can be written as $\xi^{a}=-\lambda u^{a}+\phi n^{a}$ (where
$n^{a}||h_{b}^{a}\xi^{b})$ therefore by the linearity of the Lie derivative we
have:%
\begin{align}
L_{\xi}R_{ab} &  =L_{-\lambda u}R_{ab}+L_{\phi n}R_{ab}\nonumber\\
&  =\left[  R_{ab;c}\left[  -\lambda u^{c}\right]  +R_{cb}\left[  -\lambda
u^{c}\right]  _{;a}+R_{ac}\left[  -\lambda u^{c}\right]  _{;b}\right]
+\left[  R_{ab;c}\left[  \phi n^{c}\right]  +R_{cb}\left[  \phi n^{c}\right]
_{;a}+R_{ac}\left[  \phi n^{c}\right]  _{;b}\right]  \nonumber\\
&  =\left[  -\lambda\dot{R}_{ab}+R_{cb}\left[  -\lambda u^{c}\right]
_{;a}+R_{ac}\left[  -\lambda u^{c}\right]  _{;b}\right]  +\left[  \phi
\overset{\ast}{R}_{ab}+R_{cb}\left[  \phi n^{c}\right]  _{;a}+R_{ac}\left[
\phi n^{c}\right]  _{;b}\right]  .\label{LR.07}%
\end{align}
This implies that we should brake the calculation of $L_{\xi}R_{ab}$ in three
steps. First we calculate the Lie derivative wrt a spacelike vector $\xi
^{a}=\xi n^{a}$, then wrt a timelike vector $\xi^{a}=\xi u^{a}$ and finally
wrt the general vector $\xi^{a}=-\lambda u^{a}+\phi n^{a}.\qquad$

Before we start the calculation we recall the relation:%
\begin{equation}
L_{\xi}R_{ab}=R_{ab;c}\xi^{c}+R_{cb}\xi_{;a}^{c}+R_{ac}\xi_{;a}^{c}%
=\mathring{R}_{ab}+2R_{c(a}\xi_{;b)}^{c} \label{LR.06}%
\end{equation}
where $\mathring{R}_{ab}=R_{ab;c}\xi^{c}$and a $\circ$ over a symbol means
covariant differentiation wrt $\xi^{a}.$

\subsection{1+1+2 decomposition of the terms $\mathring{R}_{ab},~2R_{c(a}%
\xi_{;b)}^{c}$}

For the term $\mathring{R}_{ab}$ we have:%
\begin{align*}
R_{ab} &  =\left(  \frac{1}{2}\mu+\frac{3}{2}p-\Lambda\right)  u_{a}%
u_{b}+vn_{a}u_{b}+\nu u_{a}n_{b}+\left(  \frac{1}{2}\mu-\frac{1}{2}%
p+\gamma+\Lambda\right)  n_{a}n_{b}+\\
&  +Q_{b}u_{a}+P_{b}n_{a}+Q_{a}u_{b}+P_{a}n_{b}+D_{ab}+\left(  \frac{1}{2}%
\mu-\frac{1}{2}p-\frac{1}{2}\gamma+\Lambda\right)  P_{ab}\Rightarrow\\
& \\
\mathring{R}_{ab} &  =\left(  \frac{1}{2}\mu+\frac{3}{2}p-\Lambda\right)
^{\circ}u_{a}u_{b}+\left(  \frac{1}{2}\mu+\frac{3}{2}p-\Lambda\right)
2\mathring{u}_{(a}u_{b)}+2\nu\left[  \mathring{n}_{(a}u_{b)}+n_{(a}%
\mathring{u}_{b)}\right]  \\
&  +\left(  \frac{1}{2}\mu-\frac{1}{2}p+\gamma+\Lambda\right)  ^{\circ}%
n_{a}n_{b}+\left(  \frac{1}{2}\mu-\frac{1}{2}p+\gamma+\Lambda\right)
2\mathring{n}_{(a}n_{b)}+\left(  \frac{1}{2}\mu-\frac{1}{2}p-\frac{1}{2}%
\gamma+\Lambda\right)  ^{\circ}P_{ab}\\
&  +\left(  \frac{1}{2}\mu-\frac{1}{2}p-\frac{1}{2}\gamma+\Lambda\right)
2\mathring{u}_{(a}u_{b)}-\left(  \frac{1}{2}\mu-\frac{1}{2}p-\frac{1}{2}%
\gamma+\Lambda\right)  2\mathring{n}_{(a}n_{b)}%
\end{align*}
finally:%
\begin{align*}
\mathring{R}_{ab} &  =\frac{1}{2}\left(  \mu+3p\right)  ^{\circ}u_{a}%
u_{b}+2\left(  \mu+p-\frac{1}{2}\gamma\right)  \mathring{u}_{(a}u_{b)}+\left(
\frac{1}{2}\mu-\frac{1}{2}p+\gamma\right)  ^{\circ}n_{a}n_{b}+2\frac{3}%
{2}\gamma\mathring{n}_{(a}n_{b)}+2\mathring{\nu}n_{(a}u_{b)}+\\
&  +2\nu\left(  \mathring{n}_{(a}u_{b)}+n_{(a}\mathring{u}_{b)}\right)
+2\mathring{Q}_{(b}u_{a)}+2Q_{(b}\mathring{u}_{a)}+2\mathring{P}_{(b}%
n_{a)}+2P_{(b}\mathring{n}_{a)}+\mathring{D}_{ab}+\left(  \frac{1}{2}\mu
-\frac{1}{2}p-\frac{1}{2}\gamma\right)  ^{\circ}P_{ab}%
\end{align*}
From this  we compute the quantities%
\begin{align*}
u_{a}u_{b} &  :\mathring{R}_{ab}u^{a}u^{b}=\frac{1}{2}\left(  \mu+3p\right)
^{\circ}+2\nu\mathring{u}_{c}n^{c}-2\mathring{Q}_{c}u^{c}\\
u_{a}n_{b} &  :\mathring{R}_{ab}u^{a}n^{b}=-\left(  \mu+p+\gamma\right)
\mathring{u}_{c}n^{c}-\mathring{\nu}-\mathring{Q}_{c}n^{c}-P^{c}\mathring
{u}_{c}\\
u_{a}P_{b}^{c} &  :\mathring{R}_{ac}u^{a}P_{b}^{c}=-\left(  \mu+p-\frac{1}%
{2}\gamma\right)  \mathring{u}_{d}P_{c}^{d}-\left(  \nu\mathring{n}%
_{d}+\mathring{Q}_{d}\right)  P_{c}^{d}-P_{d}\left(  \mathring{u}_{k}%
n^{k}\right)  P_{c}^{d}-D_{kd}\mathring{u}^{k}P_{c}^{d}\\
n_{a}n_{b} &  :\mathring{R}_{ab}n^{a}n^{b}=\left(  \frac{1}{2}\mu-\frac{1}%
{2}p+\gamma\right)  ^{\circ}+2\nu\mathring{u}_{c}n^{c}-2P_{k}\mathring{n}%
^{k}\\
n_{a}P_{b}^{c} &  :\mathring{R}_{ac}n^{a}P_{b}^{c}=\nu\mathring{u}_{d}%
P_{c}^{d}+\left(  \mathring{u}_{k}n^{k}\right)  Q_{d}P_{c}^{d}+\mathring
{P}_{d}P_{c}^{d}+\frac{3}{2}\gamma\mathring{n}_{c}-D_{dc}\mathring{n}^{d}\\
P_{a}^{c}P_{b}^{d} &  :\mathring{R}_{cd}P_{a}^{c}P_{b}^{d}=\frac{1}{2}\left(
\mu-p-\gamma\right)  ^{\circ}P_{cd}+2Q_{(k}\mathring{u}_{r)}P_{c}^{k}P_{d}%
^{r}+2P_{(c}\mathring{n}_{d)}+\mathring{D}_{kr}P_{c}^{k}P_{d}^{r}%
\end{align*}
which lead to the 1+1+2 decomposition%
\begin{align*}
\overset{\circ}{R}_{ab} &  =\left[  \frac{1}{2}\left(  \mu+3p\right)  ^{\circ
}+2\nu\mathring{u}_{c}n^{c}-2\mathring{Q}_{c}u^{c}\right]  u_{a}u_{b}-2\left[
-\left(  \mu+p+\gamma\right)  \mathring{u}_{c}n^{c}-\mathring{\nu}%
-\mathring{Q}_{c}n^{c}-P^{c}\mathring{u}_{c}\right]  u_{(a}n_{b)}+\\
&  -2\left[  -\left(  \mu+p-\frac{1}{2}\gamma\right)  \mathring{u}_{d}%
P_{c}^{d}-\left(  \nu\mathring{n}_{d}+\mathring{Q}_{d}\right)  P_{c}^{d}%
-P_{d}\left(  \mathring{u}_{k}n^{k}\right)  P_{c}^{d}-D_{kd}\mathring{u}%
^{k}P_{c}^{d}\right]  u_{(a}P_{b)}^{c}+\\
&  +\left[  \left(  \frac{1}{2}\mu-\frac{1}{2}p+\gamma\right)  ^{\circ}%
+2\nu\mathring{u}_{c}n^{c}-2P_{k}\mathring{n}^{k}\right]  n_{a}n_{b}+\\
&  +2\left[  \nu\mathring{u}_{d}P_{c}^{d}+\left(  \mathring{u}_{k}%
n^{k}\right)  Q_{d}P_{c}^{d}+\mathring{P}_{d}P_{c}^{d}+\frac{3}{2}%
\gamma\mathring{n}_{c}-D_{dc}\mathring{n}^{d}\right]  n_{(a}P_{b)}^{c}+\\
&  +\frac{1}{2}\left(  \mu-p-\gamma\right)  ^{\circ}P_{cd}+2Q_{(k}\mathring
{u}_{r)}P_{c}^{k}P_{d}^{r}+2P_{(c}\mathring{n}_{d)}+\mathring{D}_{kr}P_{c}%
^{k}P_{d}^{r}.
\end{align*}

For the term $R_{ac}\xi_{;b}^{c}$ we have
\begin{align}
R_{ac}\xi_{;b}^{c} &  =\left(  \frac{1}{2}\mu+\frac{3}{2}p-\Lambda\right)
u_{a}u_{c}\xi_{;b}^{c}+vn_{a}u_{c}\xi_{;b}^{c}+\nu u_{a}n_{c}\xi_{;b}%
^{c}+\left(  \frac{1}{2}\mu-\frac{1}{2}p+\gamma+\Lambda\right)  n_{a}n_{c}%
\xi_{;b}^{c}+\nonumber\\
&  +Q_{c}\xi_{;b}^{c}u_{a}+P_{c}\xi_{;b}^{c}n_{a}+Q_{a}u_{c}\xi_{;b}^{c}%
+P_{a}n_{c}\xi_{;b}^{c}+D_{ac}\xi_{;b}^{c}+\left(  \frac{1}{2}\mu-\frac{1}%
{2}p-\frac{1}{2}\gamma+\Lambda\right)  P_{ac}\xi_{;b}^{c}.%
\end{align}
We compute the contractions:%
\begin{align*}
u_{a}u_{b} &  :R_{ac}\xi_{;b}^{c}=-\frac{1}{2}\left(  \mu+3p-2\Lambda\right)
\left(  \dot{\xi}_{c}u^{c}\right)  -\nu n_{c}\dot{\xi}^{c}-Q_{c}\dot{\xi}%
^{c}\\
u_{a}n_{b} &  :R_{ac}\xi_{;b}^{c}=-\frac{1}{2}\left(  \mu+3p-2\Lambda\right)
u_{c}\overset{\ast}{\xi}^{c}-\nu n_{c}\overset{\ast}{\xi}^{c}-Q_{c}%
\overset{\ast}{\xi}^{c}\\
n_{a}u_{b} &  :R_{ac}\xi_{;b}^{c}=\frac{1}{2}\left(  \mu-p+2\gamma
+2\Lambda\right)  \dot{\xi}_{c}n^{c}+\nu\left(  \dot{\xi}^{c}u_{c}\right)
+P_{c}\dot{\xi}^{c}\\
u_{a}P_{b}^{c} &  :R_{ac}\xi_{;b}^{c}=\frac{1}{2}\left(  \mu+3p-2\Lambda
\right)  u_{k}\xi_{;d}^{k}P_{c}^{d}-\nu n_{k}\xi_{;d}^{k}P_{c}^{d}+Q_{k}%
\xi_{;d}^{k}P_{c}^{d}\\
P_{a}^{c}u_{b} &  :R_{ac}\xi_{;b}^{c}=\frac{1}{2}\left(  \mu-p-\gamma
+2\Lambda\right)  \dot{\xi}_{d}P_{c}^{d}+Q_{d}\left(  \dot{\xi}^{k}%
u_{k}\right)  P_{c}^{d}+P_{d}\left(  n_{k}\dot{\xi}^{k}\right)  P_{c}%
^{d}+D_{dk}\dot{\xi}^{k}P_{c}^{d}\\
n_{a}n_{b} &  :R_{ac}\xi_{;b}^{c}=\frac{1}{2}\left(  \mu-p+2\gamma
+2\Lambda\right)  \left(  \overset{\ast}{\xi}_{c}n^{c}\right)  +\nu\left(
u^{c}\overset{\ast}{\xi}_{c}\right)  +P^{c}\overset{\ast}{\xi}_{c}\\
n_{a}P_{b}^{c} &  :R_{ac}\xi_{;b}^{c}=\frac{1}{2}\left(  \mu-p+2\gamma
+2\Lambda\right)  n_{k}\xi_{;d}^{k}P_{c}^{d}+\nu u_{c}\xi_{;d}^{c}P_{e}%
^{d}+P_{c}\xi_{;d}^{c}P_{e}^{d}\\
P_{a}^{c}n_{b} &  :R_{ac}\xi_{;b}^{c}=\frac{1}{2}\left(  \mu-p-\gamma
+2\Lambda\right)  \overset{\ast}{\xi}_{c}+Q_{c}\left(  u^{k}\overset{\ast}%
{\xi}_{k}\right)  +P_{d}n_{k}\overset{\ast}{\xi}^{k}P_{c}^{d}+D_{dk}%
\overset{\ast}{\xi}^{k}P_{c}^{d}\\
P_{a}^{c}P_{b}^{d} &  :R_{ac}\xi_{;b}^{c}=\frac{1}{2}\left(  \mu
-p-\gamma+2\Lambda\right)  \xi_{k;r}P_{c}^{k}P_{d}^{r}+Q_{c}\left(  u_{k}%
\xi_{;d}^{k}\right)  +P_{c}\left(  n_{k}\xi_{;d}^{k}\right)  +D_{ck}\xi
_{;d}^{k}%
\end{align*}
from which follows:%
\begin{align}
R_{ac}\xi_{;b}^{c} &  =\left[  -\frac{1}{2}\left(  \mu+3p-2\Lambda\right)
\left(  \dot{\xi}_{c}u^{c}\right)  -\nu n_{c}\dot{\xi}^{c}-Q_{c}\dot{\xi}%
^{c}\right]  u_{a}u_{b}+\nonumber\\
&  -\left[  -\frac{1}{2}\left(  \mu+3p-2\Lambda\right)  u_{c}\overset{\ast
}{\xi}^{c}-\nu n_{c}\overset{\ast}{\xi}^{c}-Q_{c}\overset{\ast}{\xi}%
^{c}\right]  u_{a}n_{b}+\nonumber\\
&  -\left[  \frac{1}{2}\left(  \mu-p+2\gamma+2\Lambda\right)  \dot{\xi}%
_{c}n^{c}+\nu\left(  \dot{\xi}^{c}u_{c}\right)  +P_{c}\dot{\xi}^{c}\right]
u_{b}n_{a}+\nonumber\\
&  -\left[  \frac{1}{2}\left(  \mu+3p-2\Lambda\right)  u_{k}\xi_{;d}^{k}%
P_{c}^{d}-\nu n_{k}\xi_{;d}^{k}P_{c}^{d}+Q_{k}\xi_{;d}^{k}P_{c}^{d}\right]
u_{a}P_{b}^{c}+\nonumber\\
&  -\left[  \frac{1}{2}\left(  \mu-p-\gamma+2\Lambda\right)  \dot{\xi}%
_{d}P_{c}^{d}+Q_{d}\left(  \dot{\xi}^{k}u_{k}\right)  P_{c}^{d}+P_{d}\left(
n_{k}\dot{\xi}^{k}\right)  P_{c}^{d}+D_{dk}\dot{\xi}^{k}P_{c}^{d}\right]
u_{b}P_{a}^{c}+\nonumber\\
&  +\left[  \frac{1}{2}\left(  \mu-p+2\gamma+2\Lambda\right)  \left(
\overset{\ast}{\xi}_{c}n^{c}\right)  +\nu\left(  u^{c}\overset{\ast}{\xi}%
_{c}\right)  +P^{c}\overset{\ast}{\xi}_{c}\right]  n_{a}n_{b}+\nonumber\\
&  +\left[  \frac{1}{2}\left(  \mu-p+2\gamma+2\Lambda\right)  n_{k}\xi
_{;d}^{k}P_{c}^{d}+\nu u_{c}\xi_{;d}^{c}P_{e}^{d}+P_{c}\xi_{;d}^{c}P_{e}%
^{d}\right]  n_{a}P_{b}^{c}+\nonumber\\
&  +\left[  \frac{1}{2}\left(  \mu-p-\gamma+2\Lambda\right)  \overset{\ast
}{\xi}_{d}P_{c}^{d}+Q_{c}\left(  u^{k}\overset{\ast}{\xi}_{k}\right)
+P_{d}n_{k}\overset{\ast}{\xi}^{k}P_{c}^{d}+D_{dk}\overset{\ast}{\xi}^{k}%
P_{c}^{d}\right]  n_{b}P_{a}^{c}+\nonumber\\
&  +\left[  \frac{1}{2}\left(  \mu-p-\gamma+2\Lambda\right)  \xi_{k;r}%
P_{c}^{k}P_{d}^{r}+Q_{c}\left(  u_{k}\xi_{;d}^{k}\right)  +P_{c}\left(
n_{k}\xi_{;d}^{k}\right)  +D_{ck}\xi_{;d}^{k}\right]  P_{a}^{c}P_{b}^{d}.
\end{align}
Calculation of the term $R_{cb}\xi_{;a}^{c}.$

We have:%
\begin{align}
R_{cb}\xi_{;a}^{c}  &  =\left[  -\frac{1}{2}\left(  \mu+3p-2\Lambda\right)
\left(  \dot{\xi}_{c}u^{c}\right)  -\nu n_{c}\dot{\xi}^{c}-Q_{c}\dot{\xi}%
^{c}\right]  u_{a}u_{b}+\nonumber\\
&  -\left[  \frac{1}{2}\left(  \mu-p+2\gamma+2\Lambda\right)  \dot{\xi}%
_{c}n^{c}+\nu\left(  \dot{\xi}^{c}u_{c}\right)  +P_{c}\dot{\xi}^{c}\right]
u_{a}n_{b}+\nonumber\\
&  -\left[  -\frac{1}{2}\left(  \mu+3p-2\Lambda\right)  u_{c}\overset{\ast
}{\xi}^{c}-\nu n_{c}\overset{\ast}{\xi}^{c}-Q_{c}\overset{\ast}{\xi}%
^{c}\right]  u_{b}n_{a}+\nonumber\\
&  -\left[  \frac{1}{2}\left(  \mu-p-\gamma+2\Lambda\right)  \dot{\xi}%
_{d}P_{c}^{d}+Q_{d}\left(  \dot{\xi}^{k}u_{k}\right)  P_{c}^{d}+P_{d}\left(
n_{k}\dot{\xi}^{k}\right)  P_{c}^{d}+D_{dk}\dot{\xi}^{k}P_{c}^{d}\right]
u_{a}P_{b}^{c}+\nonumber\\
&  -\left[  \frac{1}{2}\left(  \mu+3p-2\Lambda\right)  u_{k}\xi_{;d}^{k}%
P_{c}^{d}-\nu n_{k}\xi_{;d}^{k}P_{c}^{d}+Q_{k}\xi_{;d}^{k}P_{c}^{d}\right]
u_{b}P_{a}^{c}+\nonumber\\
&  +\left[  \frac{1}{2}\left(  \mu-p+2\gamma+2\Lambda\right)  \left(
\overset{\ast}{\xi}_{c}n^{c}\right)  +\nu\left(  u^{c}\overset{\ast}{\xi}%
_{c}\right)  +P^{c}\overset{\ast}{\xi}_{c}\right]  n_{a}n_{b}+\nonumber\\
&  +\left[  \frac{1}{2}\left(  \mu-p-\gamma+2\Lambda\right)  \overset{\ast
}{\xi}_{d}P_{c}^{d}+Q_{c}\left(  u^{k}\overset{\ast}{\xi}_{k}\right)
+P_{d}n_{k}\overset{\ast}{\xi}^{k}P_{c}^{d}+D_{dk}\overset{\ast}{\xi}^{k}%
P_{c}^{d}\right]  n_{a}P_{b}^{c}+\nonumber\\
&  +\left[  \frac{1}{2}\left(  \mu-p+2\gamma+2\Lambda\right)  n_{k}\xi
_{;d}^{k}P_{c}^{d}+\nu u_{c}\xi_{;d}^{c}P_{e}^{d}+P_{c}\xi_{;d}^{c}P_{e}%
^{d}\right]  n_{b}P_{a}^{c}+\nonumber\\
&  +\left[  \frac{1}{2}\left(  \mu-p-\gamma+2\Lambda\right)  \xi_{k;r}%
P_{c}^{k}P_{d}^{r}+Q_{c}\left(  u_{k}\xi_{;d}^{k}\right)  +P_{c}\left(
n_{k}\xi_{;d}^{k}\right)  +D_{ck}\xi_{;d}^{k}\right]  P_{a}^{c}P_{b}^{d}%
\end{align}

Adding we find%
\begin{align}
R_{c(a}\xi_{;b)}^{c} &  =\left[  -\left(  \mu+3p-2\Lambda\right)  \left(
\dot{\xi}_{c}u^{c}\right)  -2\nu n_{c}\dot{\xi}^{c}-2Q_{c}\dot{\xi}%
^{c}\right]  u_{a}u_{b}+\nonumber\\
&  -2\left[
\begin{array}
[c]{c}%
\left(  \frac{1}{2}\left(  \mu-p+2\gamma+2\Lambda\right)  \dot{\xi}_{c}%
n^{c}+\nu\left(  \dot{\xi}^{c}u_{c}\right)  +P_{c}\dot{\xi}^{c}\right)  +\\
\\
+\left(  -\frac{1}{2}\left(  \mu+3p-2\Lambda\right)  u_{c}\overset{\ast}{\xi
}^{c}-\nu n_{c}\overset{\ast}{\xi}^{c}-Q_{c}\overset{\ast}{\xi}^{c}\right)
\end{array}
\right]  u_{(a}n_{b)}+\nonumber\\
&  -2\left[
\begin{array}
[c]{c}%
\left(  \frac{1}{2}\left(  \mu+3p-2\Lambda\right)  u_{k}\xi_{;d}^{k}P_{c}%
^{d}-\nu n_{k}\xi_{;d}^{k}P_{c}^{d}+Q_{k}\xi_{;d}^{k}P_{c}^{d}\right)  +\\
\\
+\left(
\begin{array}
[c]{c}%
\frac{1}{2}\left(  \mu-p-\gamma+2\Lambda\right)  \dot{\xi}_{d}P_{c}^{d}%
+Q_{d}\left(  \dot{\xi}^{k}u_{k}\right)  P_{c}^{d}+\\
\\
+P_{d}\left(  n_{k}\dot{\xi}^{k}\right)  P_{c}^{d}+D_{dk}\dot{\xi}^{k}%
P_{c}^{d}%
\end{array}
\right)
\end{array}
\right]  u_{(a}P_{b)}^{c}+\\
&  +\left[  \left(  \mu-p+2\gamma+2\Lambda\right)  \left(  \overset{\ast}{\xi
}_{c}n^{c}\right)  +2\nu\left(  u^{c}\overset{\ast}{\xi}_{c}\right)
+2P^{c}\overset{\ast}{\xi}_{c}\right]  n_{a}n_{b}+\nonumber\\
&  +2\left[
\begin{array}
[c]{c}%
\left(  \frac{1}{2}\left(  \mu-p+\gamma+2\Lambda\right)  n_{k}\xi_{;d}%
^{k}P_{c}^{d}+\nu u_{c}\xi_{;d}^{c}P_{e}^{d}+P_{c}\xi_{;d}^{c}P_{e}%
^{d}\right)  +\\
\\
+\left(
\begin{array}
[c]{c}%
\frac{1}{2}\left(  \mu-p-\gamma+2\Lambda\right)  \overset{\ast}{\xi}_{d}%
P_{c}^{d}+\\
\\
+Q_{c}\left(  u^{k}\overset{\ast}{\xi}_{k}\right)  +P_{d}n_{k}\overset{\ast
}{\xi}^{k}P_{c}^{d}+D_{dk}\overset{\ast}{\xi}^{k}P_{c}^{d}%
\end{array}
\right)
\end{array}
\right]  n_{(a}P_{b)}^{c}+\nonumber\\
&  +\left[
\begin{array}
[c]{c}%
\left(  \mu-p-\gamma+2\Lambda\right)  \xi_{k;r}P_{c}^{k}P_{d}^{r}+\\
+2Q_{c}\left(  u_{k}\xi_{;d}^{k}\right)  +2P_{c}\left(  n_{k}\xi_{;d}%
^{k}\right)  +2D_{ck}\xi_{;d}^{k}%
\end{array}
\right]  P_{a}^{c}P_{b}^{d}\nonumber
\end{align}
Finally we write for the quantity $L_{\xi}R_{ab}$ in the 1+1+2 decomposition:%
\begin{align}
L_{\xi}R_{ab} &  =\left[
\begin{array}
[c]{c}%
\frac{1}{2}\left(  \mu+3p\right)  ^{\circ}+2\nu\mathring{u}_{c}n^{c}%
-2\mathring{Q}_{c}u^{c}+\\
\\
-\left(  \mu+3p-2\Lambda\right)  \left(  \dot{\xi}_{c}u^{c}\right)  -2\nu
n_{c}\dot{\xi}^{c}-2Q_{c}\dot{\xi}^{c}%
\end{array}
\right]  u_{a}u_{b}+\nonumber\\
&  -2\left[
\begin{array}
[c]{c}%
-\left(  \mu+p+\gamma\right)  \mathring{u}_{c}n^{c}-\mathring{\nu}%
-\mathring{Q}_{c}n^{c}-P^{c}\mathring{u}_{c}+\\
\\
+\left(  \frac{1}{2}\left(  \mu-p+2\gamma+2\Lambda\right)  \dot{\xi}_{c}%
n^{c}+\nu\left(  \dot{\xi}^{c}u_{c}\right)  +P^{c}\dot{\xi}_{c}\right)  +\\
\\
+\left(  -\frac{1}{2}\left(  \mu+3p-2\Lambda\right)  u_{c}\overset{\ast}{\xi
}^{c}-\nu n_{c}\overset{\ast}{\xi}^{c}-Q_{c}\overset{\ast}{\xi}^{c}\right)
\end{array}
\right]  u_{(a}n_{b)}+\nonumber\\
&  -2\left[
\begin{array}
[c]{c}%
-\left(  \mu+p-\frac{1}{2}\gamma\right)  \mathring{u}_{d}P_{c}^{d}-\left(
\nu\mathring{n}_{d}+\mathring{Q}_{d}\right)  P_{c}^{d}-P_{d}\left(
\mathring{u}_{k}n^{k}\right)  P_{c}^{d}-D_{kd}\mathring{u}^{k}P_{c}^{d}\\
\\
+\left(  \frac{1}{2}\left(  \mu+3p-2\Lambda\right)  u_{k}\xi_{;d}^{k}P_{c}%
^{d}-\nu n_{k}\xi_{;d}^{k}P_{c}^{d}+Q_{k}\xi_{;d}^{k}P_{c}^{d}\right)  +\\
\\
+\left(  \frac{1}{2}\left(  \mu-p-\gamma+2\Lambda\right)  \dot{\xi}_{d}%
P_{c}^{d}+Q_{d}\left(  \dot{\xi}^{k}u_{k}\right)  P_{c}^{d}+P_{d}\left(
n_{k}\dot{\xi}^{k}\right)  P_{c}^{d}+D_{dk}\dot{\xi}^{k}P_{c}^{d}\right)
\end{array}
\right]  u_{(a}P_{b)}^{c}+\nonumber\\
&  +\left[
\begin{array}
[c]{c}%
\left(  \frac{1}{2}\mu-\frac{1}{2}p+\gamma\right)  ^{\circ}+2\nu\mathring
{u}_{c}n^{c}-2P_{k}\mathring{n}^{k}\\
\\
+\left(  \mu-p+2\gamma+2\Lambda\right)  \left(  \overset{\ast}{\xi}_{c}%
n^{c}\right)  +2\nu\left(  u^{c}\overset{\ast}{\xi}_{c}\right)  +2P^{c}%
\overset{\ast}{\xi}_{c}%
\end{array}
\right]  n_{a}n_{b}+\label{LR.09}\\
&  +2\left[
\begin{array}
[c]{c}%
\nu\mathring{u}_{d}P_{c}^{d}+\left(  \mathring{u}_{k}n^{k}\right)  Q_{d}%
P_{c}^{d}+\mathring{P}_{d}P_{c}^{d}+\frac{3}{2}\gamma\mathring{n}_{c}%
-D_{dc}\mathring{n}^{d}\\
\\
+\left(  \frac{1}{2}\left(  \mu-p+2\gamma+2\Lambda\right)  n_{k}\xi_{;d}%
^{k}P_{c}^{d}+\nu u_{c}\xi_{;d}^{c}P_{e}^{d}+P_{c}\xi_{;d}^{c}P_{e}%
^{d}\right)  +\\
\\
+\left(  \frac{1}{2}\left(  \mu-p-\gamma+2\Lambda\right)  \overset{\ast}{\xi
}_{d}P_{c}^{d}+Q_{c}\left(  u^{k}\overset{\ast}{\xi}_{k}\right)  +P_{d}%
n_{k}\overset{\ast}{\xi}^{k}P_{c}^{d}+D_{dk}\overset{\ast}{\xi}^{k}P_{c}%
^{d}\right)
\end{array}
\right]  n_{(a}P_{b)}^{c}+\nonumber\\
&  +\left[
\begin{array}
[c]{c}%
\frac{1}{2}\left(  \mu-p-\gamma\right)  ^{\circ}P_{cd}+2Q_{(k}\mathring
{u}_{r)}P_{c}^{k}P_{d}^{r}+2P_{(c}\mathring{n}_{d)}+\mathring{D}_{kr}P_{c}%
^{k}P_{d}^{r}+\\
\\
\left(  \mu-p-\gamma+2\Lambda\right)  \xi_{k;r}P_{c}^{k}P_{d}^{r}%
+2Q_{c}\left(  u_{k}\xi_{;d}^{k}\right)  +2P_{c}\left(  n_{k}\xi_{;d}%
^{k}\right)  +2D_{ck}\xi_{;d}^{k}%
\end{array}
\right]  P_{a}^{c}P_{b}^{d}\nonumber
\end{align}

If we replace $\xi^{a}=-\lambda u^{a}+\phi n^{a}$ we will find the complete
answer in the final 1+1+2 form. In the following we work with the particular
cases $\xi^{a}=-\lambda u^{a}$ and $\xi^{a}=\phi n^{a}$.

\subsubsection{The case $\xi^{a}=-\lambda u^{a}$}

When $\xi^{a}=-\lambda u^{a}$ we find%
\begin{align}
L_{\left(  -\lambda u^{a}\right)  }R_{ab} &  =-\lambda\left[  \frac{1}%
{2}\left(  \mu+3p\right)  ^{\cdot}+\left(  \mu+3p-2\Lambda\right)  \left(
\ln\lambda\right)  ^{\cdot}\right]  u_{a}u_{b}+\nonumber\\
&  +\lambda\left[
\begin{array}
[c]{c}%
-\left(  \mu+3p-2\Lambda\right)  \left(  \dot{u}_{c}n^{c}-\left(  \ln
\lambda\right)  ^{\ast}\right)  -2\nu\left(  \ln\lambda\right)  ^{\cdot}%
\qquad\\
\\
-2\dot{\nu}+2Q^{c}N_{c}-2\nu\left(  \sigma_{dc}n^{c}n^{d}+\frac{1}{3}%
\theta\right)
\end{array}
\right]  u_{(a}n_{b)}+\nonumber\\
& \nonumber\\
&  +\lambda\left[
\begin{array}
[c]{c}%
-\left(  \mu+3p-2\Lambda\right)  \left[  \dot{u}_{c}-\left(  \ln
\lambda\right)  _{,c}\right]  -2Q_{c}\left(  \ln\lambda\right)  ^{\cdot}+\\
\\
-2\nu N_{c}-4\nu\sigma_{dc}n^{d}-2\dot{Q}_{c}-2Q_{d}\left(  \omega_{.c}%
^{d}+\sigma_{.c}^{d}+\frac{1}{3}\theta h_{c}^{d}\right)
\end{array}
\right]  u_{(a}P_{b)}^{c}+\label{LR.10}\\
& \nonumber\\
&  -\lambda\left[
\begin{array}
[c]{c}%
\frac{1}{2}\left(  \mu-p+2\gamma\right)  ^{\cdot}+\left(  \mu-p+2\gamma
+2\Lambda\right)  \left(  \sigma_{cd}n^{c}n^{d}+\frac{1}{3}\theta\right)  +\\
\\
-2P_{c}N^{c}+2\nu\left[  \dot{u}_{d}n^{d}-\left(  \ln\lambda\right)  ^{\ast
}\right]
\end{array}
\right]  n_{a}n_{b}+\nonumber\\
& \nonumber\\
&  -\lambda\left[
\begin{array}
[c]{c}%
2\left(  \mu-p+2\gamma+2\Lambda\right)  \sigma_{cd}n^{d}+\frac{6}{2}\gamma
N_{c}+2\dot{P}_{c}-2D_{cd}N^{d}+\frac{4}{3}\theta P_{c}+\\
\\
+2Q_{c}\left[  \dot{u}_{d}n^{d}-\left(  \ln\lambda\right)  ^{\ast}\right]
+2\nu\left[  \dot{u}_{c}-\left(  \ln\lambda\right)  _{,c}\right]  +\\
\\
+2P_{c}\left(  \sigma_{de}n^{d}n^{e}\right)  +2P_{d}\left(  \sigma_{.c}%
^{d}+\omega_{.c}^{d}\right)
\end{array}
\right]  n_{(a}P_{b)}^{d}+\nonumber\\
& \nonumber\\
&  -\lambda\left[
\begin{array}
[c]{c}%
\frac{1}{2}\left(  \mu-p-\gamma\right)  ^{\cdot}P_{cd}+\left(  \mu
-p-\gamma+2\Lambda\right)  \left(  \sigma_{cd}-\frac{1}{3}\theta
P_{cd}\right)  +\\
\\
+2P_{(c}P_{d)}^{e}\left[  N_{e}+\left(  \sigma_{ef}+\omega_{ef}\right)
n^{f}\right]  +\dot{D}_{cd}+\frac{2}{3}\theta D_{cd}+\\
\\
+2Q_{(c}\left[  \dot{u}_{d)}-\left(  \ln\lambda\right)  _{;d)}\right]
+2P_{(c}\left(  \omega_{.d)}^{e}+\sigma_{.d)}^{e}\right)  n_{e}+2D_{e(c}%
\left(  \omega_{.d)}^{e}+\sigma_{.d)}^{e}\right)
\end{array}
\right]  P_{a}^{c}P_{b}^{d}\nonumber
\end{align}

The last term can be decomposed in trace and a trace-free part. Contracting
with $P^{ab}$ we have the trace part:%
\[
I=-\frac{\lambda}{2}\left[
\begin{array}
[c]{c}%
\left(  \mu-p-\gamma\right)  ^{\cdot}-\left(  \mu-p-\gamma+2\Lambda\right)
\left(  \sigma_{cd}n^{c}n^{d}-\frac{2}{3}\theta\right)  +\\
+2P^{e}\left[  N_{e}+\sigma_{ef}n^{f}\right]  +2Q^{e}\left[  \dot{u}%
_{e}-\left(  \ln\lambda\right)  _{;e)}\right]  +2D_{ec}\sigma^{ec}%
\end{array}
\right]
\]
and the trace free part:
\[
H_{ab}=-\lambda\left[
\begin{array}
[c]{c}%
\left(  \mu-p-\gamma+2\Lambda\right)  \sigma_{cd}+2P_{(c}P_{d)}^{e}\left[
N_{e}+\left(  \sigma_{ef}+\omega_{ef}\right)  n^{f}\right]  -P^{e}\left[
N_{e}+\sigma_{ef}n^{f}\right]  P_{cd}+\\
+\dot{D}_{cd}+\frac{2}{3}\theta D_{cd}+2Q_{(c}\left[  \dot{u}_{d)}-\left(
\ln\lambda\right)  _{;d)}\right]  -Q^{e}\left[  \dot{u}_{e}-\left(  \ln
\lambda\right)  _{;e)}\right]  P_{cd}+\\
+2P_{(c}\left(  \omega_{.d)}^{e}+\sigma_{.d)}^{e}\right)  n_{e}+2D_{e(c}%
\left(  \omega_{.d)}^{e}+\sigma_{.d)}^{e}\right)  -\left(  D_{ec}\sigma
^{ec}\right)  P_{cd}%
\end{array}
\right]
\]

\subsubsection{The case $\xi^{a}=\phi n^{a}$}

When $\xi^{a}=\phi n^{a}$ we calculate%
\begin{align}
&  L_{\left(  \phi n^{a}\right)  }R_{ab}=\phi\left[
\begin{array}
[c]{c}%
\frac{1}{2}\left(  \mu+3p\right)  ^{\ast}+\left(  \mu+3p-2\Lambda\right)
\left(  \dot{u}^{c}n_{c}\right)  +\\
\\
-2\nu\left[  \left(  \ln\phi\right)  ^{\cdot}-\sigma_{cd}n^{c}n^{d}+\frac
{1}{3}\theta\right]  -2Q_{c}N^{c}%
\end{array}
\right]  u_{a}u_{b}+\nonumber\\
& \nonumber\\
&  -2\phi\left[
\begin{array}
[c]{c}%
\frac{1}{2}\left(  \mu-p+2\Lambda+2\gamma\right)  \left[  \left(  \ln
\phi\right)  ^{\cdot}-\left(  \sigma_{cd}n^{c}n^{d}+\frac{1}{3}\theta\right)
\right]  +\\
\\
+P^{c}N_{c}-\overset{\ast}{\nu}-\nu\left[  \left(  \ln\phi\right)  ^{\ast
}+\dot{u}^{c}n_{c}\right]
\end{array}
\right]  u_{(a}n_{b)}\text{ }\label{LR.11}\\
& \nonumber\\
&  -2\phi\left[
\begin{array}
[c]{c}%
\frac{1}{2}\left(  \mu-p-\gamma+2\Lambda\right)  N_{c}+\left(  \mu
+3p-2\Lambda\right)  \omega_{dc}n^{d}+\\
\\
+P_{c}\left[  \left(  \ln\phi\right)  ^{\cdot}-\left(  \sigma_{ef}n^{e}%
n^{f}+\frac{1}{3}\theta\right)  \right]  +D_{dc}N^{d}-\nu\left(  \overset
{\ast}{n}_{c}+\left(  \ln\phi\right)  _{;c}\right)  +\\
\\
-Q_{c}\left(  \dot{u}^{d}n_{d}\right)  -\overset{\ast}{Q}_{d}P_{c}^{d}%
-Q_{r}\left(  \mathit{R}_{.d}^{r}\mathit{+S}_{.d}^{r}\mathit{+}\frac{1}%
{2}\mathit{E}P_{d}^{r}\right)
\end{array}
\right]  u_{(a}P_{b)}^{c}\nonumber\\
& \nonumber\\
&  +\phi\left[  \frac{1}{2}\left(  \mu-p+2\gamma\right)  ^{\ast}+\left(
\mu-p+2\Lambda+2\gamma\right)  \left(  \ln\phi\right)  ^{\ast}\right]
n_{a}n_{b}\nonumber\\
& \nonumber\\
&  +2\phi\left[
\begin{array}
[c]{c}%
\overset{\ast}{P}_{c}+\frac{1}{2}\left(  \mu-p+2\gamma+2\Lambda\right)
\left(  \overset{\ast}{n}_{c}+\left(  \ln\phi\right)  _{;c}\right)
-2\nu\omega_{dc}n^{d}+\\
+P_{r}\left[  \left(  \mathit{R}_{.c)}^{r}\mathit{+S}_{.c)}^{r}\mathit{+}%
\frac{1}{2}\mathit{E}P_{c)}^{r}\right)  +\left(  \ln\phi\right)  ^{\ast}%
P_{c}^{r}\right]
\end{array}
\right]  n_{(a}P_{b)}^{c}\nonumber\\
& \nonumber\\
&  +\phi\left[
\begin{array}
[c]{c}%
\frac{1}{2}\left(  \mu-p-\gamma\right)  ^{\ast}P_{cd}+\left(  \mu
-p-\gamma+2\Lambda\right)  \left(  \mathit{S}_{cd}+\frac{1}{2}P_{cd}%
\mathit{E}\right)  +\\
\\
+\overset{\ast}{D}_{cd}+2D_{r(c}\left(  \mathit{R}_{.d)}^{r}\mathit{+S}%
_{.d)}^{r}\mathit{+}\frac{1}{2}\mathit{E}P_{d)}^{r}\right)  +\\
\\
+4Q_{(c}\omega_{d)r}n^{r}+2P_{(c}P_{d)}^{e}\left(  \overset{\ast}{n}%
_{e}+\left(  \ln\phi\right)  _{;e}\right)
\end{array}
\right]  P_{a}^{c}P_{b}^{d}\nonumber
\end{align}

The trace and the trace free parts of the last term are:%
\[
I=\frac{\phi}{2}\left[
\begin{array}
[c]{c}%
\left(  \mu-p-\gamma\right)  ^{\ast}+\left(  \mu-p-\gamma+2\Lambda\right)
\mathit{E}+2D_{rc}\mathit{S}^{rc}+\\
+4Q^{d}\omega_{dr}n^{r}+2P^{e}\left(  \overset{\ast}{n}_{e}+\left(  \ln
\phi\right)  _{;e}\right)
\end{array}
\right]
\]%
\[
H_{ab}=+\phi\left[
\begin{array}
[c]{c}%
\left(  \mu-p-\gamma+2\Lambda\right)  \mathit{S}_{cd}+\overset{\ast}{D}%
_{cd}+2D_{r(c}\left(  \mathit{R}_{.d)}^{r}\mathit{+S}_{.d)}^{r}\mathit{+}%
\frac{1}{2}\mathit{E}P_{d)}^{r}\right)  +\\
-\left(  2D_{rc}\mathit{J}^{rc}\right)  P_{cd}+4Q_{(c}\omega_{d)r}%
n^{r}-2\left(  Q^{d}\omega_{dr}n^{r}\right)  P_{cd}+\\
+2P_{(c}P_{d)}^{e}\left(  \overset{\ast}{n}_{e}+\left(  \ln\phi\right)
_{;e}\right)  -P^{e}\left(  \overset{\ast}{n}_{e}+\left(  \ln\phi\right)
_{;e}\right)  P_{cd}%
\end{array}
\right].
\]
These general results can be used to study the kinematics and he dynamics of
all spacetime models for all types of observers,  all types of symmetries and
for general matter fields.

\textbf{Acknowledgements}

AP acknowledges the financial support of FONDECYT grant no. 3160121.


\newpage

\begin{thebibliography}{99}                                                                                               %


\bibitem {ein1}A. Einstein, Grundgedanken der allgemeinen
Relativit\"{a}tstheorie und Anwendung dieser Theorie in der Astronomie,
Preussische Akademie der Wissenschaften, Sitzungsberichte, 315, (1915)

\bibitem {ein2}A. Einstein, Zur allgemeinen Relativit\"{a}tstheorie (On the
General Theory of Relativity), Preussische Akademie der Wissenschaften,
Sitzungsberichte, 778, (1915)

\bibitem {ein3}A. Einstein, Erkl\"{a}rung der Perihelbewegung des Merkur aus
der allgemeinen Relativit\"{a}tstheorie, Preussische Akademie der
Wissenschaften, Sitzungsberichte, 831, (1915)

\bibitem {ein4}A. Einstein, Feldgleichungen der Gravitation, Preussische
Akademie der Wissenschaften, Sitzungsberichte, 844, (1915)

\bibitem {exact}H. Stephani, D. Kramer, M. MacCallum, C. Hoenselaers and E.
Herlt, Exact Solutions of Einstein's Field Equations, Cambridge University
Press, Cambridge (2003)

\bibitem {Kransinski}A. Krasi\'{n}ski, Inhomogeneous Cosmological Models,
Cambridge University Press, New York (2006)

\bibitem {nakamura}M. Nakamura, Geometry, Topology and Physics, Taylor \&
Francis Group, Florida (2003)

\bibitem {lie1}S. Lie, Theorie der Transformationsgruppen I, Leipzig: B. G.
Teubner (1888)

\bibitem {anderson}J.L. Anderson, Principles of Relativity Physics, Academi
Press, New York, (1973)

\bibitem {spinors}R. Penrose, Annals of Physiscs, 10, 171 (1960)

\bibitem {diracgr}P. Dirac, General Theory of Relativity, Princeton University
Press, Wiley, New York (1975)

\bibitem {M. Tsamparlis 1992}M. Tsamparlis,\ J. Math. Phys. \textbf{33}, 1472 (1992)

\bibitem {katzin1}G.H. Katzin, J. Levine and W.R. Davis, J. Math. Phys. 10,
617 (1969)

\bibitem {ellis}G.F.R. Ellis and H.V. Elst, Cosmological Models, Carg\`{e}se
lectures [gr-qc/9812046]

\bibitem {Herrera et all (1984)}L. Herrera, J. Jemenez, L. Leal, J. Ponce de
Leon, M. Esculpi and V. Galina, J. Math. Phys. 25, 3274 (1984)

\bibitem {Maartens Mason Tsamparlis  1985}R. Maartens, D.P. Mason and M.
Tsamparlis,\ J. Math. Phys. 27, 2987 (1986)

\bibitem {Saridakis Tsamparlis 1991}E. Saridakis and M. Tsamparlis M,\ J.
Math. Phys. 32, 1541 (1991)

\bibitem {Yano 1956}K. Yano The theory of Lie derivatives and its
Applications, Amsterdam: North Holland, (1956)

\bibitem {Tsamparlis 1992}Tsamparlis M 1992, \textquotedblleft Geometrization
of a General Collineation\textquotedblright\ J. Math. Phys. 33, 1472 - 1479

\bibitem {Ellis1}G.F.R. Ellis, J.\ Math. Phys. 8, 1171 (1967)

\bibitem {Stewart-Ellis}J. M. Stewart and G.F.R. Ellis, J. Math. Phys. 9, 1072 (1968)

\bibitem {Ray  (1978)}D. Ray,\ Phys Rev D, 18, 3879 (1978)

\bibitem {Lund F Regge T (1976)}F. Lund and T. Regge, Phys Rev D, 14, 1524 (1976)

\bibitem {Letelier 1980}P. Letelier,\ Phys Rev D 22, 807 (1980)

\bibitem {Letelier 1981}P. Letelier,\ Nuovo Cimento B63, 519 (1981)

\bibitem {Letelier 1983}P. Letelier, Phys Rev D 28, 2414 (1983)

\bibitem {Yavuz Yilmaz (1997)}I. Yavuz and I. Yilmaz, Gen. Rel. Grav. 9, 1295 (1997)

\bibitem {Yilmaz 2001}I. Yilmaz, IJMPD 10, 681 (2001)

\bibitem {Baysal Yilmaz 2002}H. Baysal and I. Yilmaz, Class Quantum Grav 19,
6435 (2002)

\bibitem {U Camci 2002}U. Camci, IJMPD 11, 353 (2002)

\bibitem {Baysal Camci et all 2002}H. Baysal, U. Camci, I. Tarhan and I.
Yilmaz, IJMPD 11, 463 (2002)

\bibitem {Sharif M Sheikh U 2005}M. Sharif and\ U. Sheikh , IJMPA 21, 3213 (2006)

\bibitem {M. Tsamparlis 2006}M. Tsamparlis, Gen. Relativ. Grav. 38, 311 (2006)
\end{thebibliography}
\end{document}